\documentclass[a4paper,10pt]{article}
\usepackage{amsthm,amsmath,amssymb,bbm,geometry,epsfig,hyperref,enumerate,comment,nicefrac,listings,graphicx,color,relsize}

\usepackage[ruled,vlined]{algorithm2e}
\usepackage{algorithmic}
\usepackage[
backend=bibtex,
style=numeric,
sorting=nyt,
maxbibnames=9
]{biblatex}
\usepackage{bm}
\usepackage{subfig}
\usepackage{hyperref}
\usepackage{stmaryrd}
\usepackage{dsfont}
\usepackage{listings,newtxtt}
\usepackage{caption}
\usepackage{mathtools}
\usepackage{amsfonts}
\usepackage{lipsum}
\usepackage{multicol}
\usepackage{braket}
\usepackage[dvipsnames]{xcolor}
\usepackage{yfonts}
\usepackage[font={small,it}]{caption}
\usepackage{lmodern}
\usepackage{fancyhdr}
\usepackage{pdfpages}
\usepackage[shortlabels]{enumitem}
\usepackage{mathrsfs}
\usepackage{quantikz}
\usepackage{placeins}
\newtheoremstyle{break}{5pt}{5pt}{\em}{}{\bfseries}{}{ }{\thmname{#1}\thmnumber{ #2}\thmnote{ (#3)}}
\theoremstyle{break}
\newtheorem{theorem}{Theorem}
\newtheorem{definition}{Definition}[section]

\newtheorem{assumption}[definition]{Assumption}

\newtheorem{proposition}[definition]{Proposition}
\newtheorem{lemma}[definition]{Lemma}
\newtheorem{corollary}[definition]{Corollary}
\newtheorem{remark}[definition]{Remark}
\newtheorem{example}[definition]{Example}

\newcommand{\R}{{\mathbb R}}

\newcommand{\C}{{\mathbb C}}
\newcommand{\Prob}{{\mathbb P}}
\newcommand{\N}{{\mathbb N}}

\newcommand{\Z}{{\mathbb Z}}

\newcommand{\E}{{\mathbb E}}
\newcommand{\eps}{{\varepsilon}}
\newcommand{\Tau}{\mathcal{T}}
\newcommand{\Beta}{\mathcal{B}}
\newcommand{\vertiii}[1]{{\left\vert\kern-0.25ex\left\vert\kern-0.25ex\left\vert #1 
    \right\vert\kern-0.25ex\right\vert\kern-0.25ex\right\vert}}
\newcommand{\anc}{\mathrm{anc}}
\newcommand{\CNOT}{\mathrm{CNOT}}
\newcommand{\SWAP}{\mathrm{SWAP}}

\usepackage{physics}
\newcommand*{\bigboxplus}{%
    \DOTSB  
    \mathop{%
        \mathchoice
        {\boxplus \huge}%
        {\boxplus \LARGE}%
        {\boxplus {}}%
        {\boxplus \footnotesize}%
    }
}

\begin{document}
\title{\vspace{-1.0cm} 
	Quantum Monte Carlo algorithm for option pricing\\
	and its
	complexity analysis 
}

\author{Jianjun Chen$^1$, Yongming Li$^2$, Ariel Neufeld$^3$
	\bigskip
	\\
	\small{$^1$ Division of Physics and Applied Physics, School of Physical and Mathematical Sciences,}
	\\
	\small{Nanyang Technological University, Singapore,
		e-mail: chen1554@e.ntu.edu.sg.}
    \\
    \small{$^2$ Department of Mathematics, Texas A\&M University,}
    \\
    \small{Texas, USA, 
        e-mail: liyo0008@tamu.edu}
	\\
	\small{$^3$ Division of Mathematical Sciences, School of Physical and Mathematical Sciences,}
	\\
	\small{Nanyang Technological University, Singapore,
		e-mail: ariel.neufeld@ntu.edu.sg}
}
\maketitle
\begin{abstract}
In this paper we provide a quantum Monte Carlo algorithm to solve multidimensional Black-Scholes PDEs with correlation for option pricing. The payoff function of the option is of general form and is only required to be continuous and piecewise affine, which covers most of the relevant payoff functions used in finance. We provide a rigorous error analysis and complexity analysis of our algorithm. In particular, we prove that the computational complexity of our algorithm is bounded polynomially in the space
dimension $d$ of the PDE and the reciprocal of the prescribed accuracy $\varepsilon$. 
Moreover, we show that for payoff functions which are bounded, our algorithm indeed has a speed-up compared to classical Monte Carlo methods. Furthermore, we provide numerical simulations in two dimensions using our developed package within the  \texttt{Qiskit} framework tailored to price continuous piecewise affine options with respect to the Black-Scholes model, as well as discuss the potential extension of the numerical simulations to arbitrary space dimension.

\end{abstract}

\vspace{0.5cm}
{
	\footnotesize
\tableofcontents
}
\section{Introduction}
It is well known that, under the Black-Scholes framework, the price of a financial derivative depending on 
$d$ underlying assets can be characterized as the unique solution of the corresponding 
$d$-dimensional Black-Scholes PDE, with the terminal condition given by the derivative’s payoff function \cite{black1973pricing}.
 However, in the multidimensional setting of the Black-Scholes model, 
 there are no explicit expressions for the solution of the corresponding PDE, and hence numerical methods are necessary to approximately solve these PDEs. It is crucial both from a theoretical, but especially also from a practical point of view to provide a rigorous and comprehensive error analysis and complexity analysis of these numerical methods. Indeed, while the precise error analysis allows, in the context of option pricing, to precisely quantify the exact deviation between the numerical approximation and the theoretical option price,
 a rigorous complexity analysis 
 reveals how the computational cost of an algorithm scales with respect to the  space-dimension of the PDE. Ideally, one aims to design numerical algorithms whose error and computational complexity can be characterized \textit{precisely}, and whose computational cost grows only \textit{polynomially} with the dimension~$d$ and the reciprocal of the desired accuracy~$\varepsilon$. 

Monte Carlo based algorithms have demonstrated both theoretically and practically to be efficient for pricing multidimensional options under the Black-Scholes model. In particular, the computational complexity of Monte Carlo methods typically does not 
grow  exponentially in the dimension $d$ and the reciprocal of the accuracy $\varepsilon$, see, e.g., \cite{boyle1977options}.
 Moreover, recently, deep learning based methods have been developed using neural networks which can approximately solve the multidimensional Black-Scholes PDE 
  \cite{beck2021deep,berner2020analysis,elbrachter2022dnn,grohs2018proof, han2018solving}.
 However, to date, there is no rigorous and comprehensive error and complexity analysis available for these algorithms, as the training of neural networks consists of solving non-convex optimization problems for which there is no theoretical guarantee that the optimization error can be controlled to be arbitrarily small \cite{Goodfellow-et-al-2016}.

In recent years, there has been a rapid development of numerical methods dealing with problems in
quantitative finance using quantum computers. The motivation comes from the fact that qubits, compared to classical bits, are allowed quantum mechanically to be in a state of superposition, from which one anticipates that quantum computers should be able to achieve much higher computational power than classical (super-) computers. We also refer to \cite{gonon2023universal} for a universal approximation theorem for quantum neural networks.
The applications of quantum algorithms in finance include portfolio optimization  \cite{rebentrost2018quantum}, 
the computation of risk measures such as \textit{Value at Risk (VAR)} \cite{QC6_risk_analysis},
volatility modeling \cite{assouel2022quantum}, optimal stopping \cite{doriguello2022quantum},
and option pricing,  particularly in the Black-Scholes model
\cite{chakrabarti2021threshold, fontanela2021quantum,fuchs2023hybrid, kubo2022pricing,ramos2021quantum, QC5_Patrick,QC4_optionpricing,QC16_GAN}. We also refer to the monograph \cite{jacquier2022quantum} and surveys \cite{egger2020quantum, jacquieroverview2023,  orus2019quantum} for (further) applications of quantum computing in finance. 
Moreover, \cite{ an2021quantum,  arrazola2019quantum, childs2021high, dees2023unsupervised,   QC2_Heat, montanaro2016quantum} proposed quantum algorithms to approximately solve different PDEs than the Black-Scholes PDE used in finance for option pricing. Furthermore, recently, \cite{crew2025quantum} connected the theory of path signatures to  matrix model distributions in quantum field theory, allowing them to introduce the concept of  quantum path signature feature map together with a corresponding quantum signature kernel.

Let us discuss more in detail the various quantum methodologies developed to price financial derivatives under the Black-Scholes model.
\cite{QC16_GAN} used \textit{quantum Generative Adversarial Networks} (qGANs) for learning and loading of probability distributions, including the multivariate log-normal distribution, and applied it to price financial derivatives under the Black-Scholes model. Recently,  \cite{fuchs2023hybrid} introduced \textit{Wasserstein}  qGAN (qWGAN) for learning and loading of probability distributions employing the Wasserstein loss, and numerically demonstrated the advantage over the standard qGAN methodology. 
\cite{kubo2022pricing} employed variational quantum simulation in order to solve the Black-Scholes PDE.
\cite{fontanela2021quantum} proposed an innovative hybrid quantum-classical algorithm to approximately solve the one-dimensional Black--Scholes PDE by exploiting its relation to the Schr\"odinger equation in imaginary time. More precisely, they first transform the Black-Scholes PDE into the heat equation and then interpret its solution as a wave function of the corresponding Schr\"odinger equation. To learn this wave function, they approximate it using a sequence of parametrized quantum gates, where the approximately optimal parameters are obtained by solving a suitable variational principle. 
However, as highlighted in \cite{fontanela2021quantum}, this approach requires an \textit{ansatz circuit} that should, on the one hand, accurately approximate the wave function, but on the other hand, allow for an efficient solution of the associated minimization problem that determines the parameters. Consequently, it remains unclear how to concretely construct such an \textit{ansatz circuit} in higher dimensions. Moreover, since the underlying minimization problem is non-convex, there is no theoretical guarantee of finding the optimal parameters.
Furthermore, \cite{ramos2021quantum} employed the unary basis of the asset's value to design a quantum algorithm to price under the Black-Scholes model.

Most literature uses quantum Monte Carlo methods to approximately solve the Black-Scholes PDE in order to price financial options. More precisely, these works rely on the \textit{Quantum Amplitude Estimation algorithm~(QAE)} \cite{QC1_QAE} which estimates the expected value of a random parameter (see Section~\ref{section: QAE} for a detailed discussion) based on an extension of  \textit{Grover's search algorithm} \cite{QC19_Grover_search}. Several variations of the
Quantum Amplitude Estimation algorithm have been proposed recently, see e.g.\  \cite{aaronson2019quantum,  
	fukuzawa2022modified, giurgica2022low,
	grinko2021iterative, 
	manzano2023real,
	Nakaji20,
	 plekhanov2022variational, rall2022amplitude,suzuki2020amplitude,wie2019simpler,zhao2022adaptive}. 
Quantum Monte Carlo methods can ideally achieve a quadratic speed-up \cite{herbert2022quantum}, \cite{QC_8_Montanaro_QMC} compared to classical (i.e.\ non-quantum) Monte Carlo methods. However, the quadratic speedup can only be achieved if there is a so-called \textit{oracle} quantum circuit which can correctly upload the corresponding distribution in  rotated quantum form, without any approximation errors (caused, e.g., from discretization and rotation), such that it is applicable to  a quantum amplitude estimation algorithm. This assumption however in most cases cannot be justified in practice, as highlighted, e.g., in \cite{chakrabarti2021threshold,QC16_GAN}.

In this paper, we propose a quantum Monte Carlo algorithm for solving multidimensional Black-Scholes PDEs with correlated assets and general payoff functions that are continuous and piecewise affine. This class of payoffs encompasses most of the standard financial derivatives used in practice (see also Section~\ref{subsec:CPWA}).  Our algorithm conceptually builds upon the quantum Monte Carlo methods proposed in~\cite{chakrabarti2021threshold, QC5_Patrick, QC4_optionpricing}, which first encode the multivariate log-normal distribution and the payoff function in a rotated quantum form, and then employ a Quantum Amplitude Estimation (QAE) algorithm to approximately solve the Black-Scholes PDE and compute option prices.

Our main contribution lies in a rigorous and comprehensive error analysis as well as complexity analysis of our algorithm. To that end, we first introduce quantum circuits that can perform arithmetic operations on two complement's numbers representing signed dyadic rational numbers, together with its complexity analysis. This allows us to provide a rigorous error and complexity analysis when uploading first a truncated and discretized approximation of the multivariate log-normal distribution and then uploading an approximation of the continuous piecewise affine payoff function in rotated form, where the approximation consists of truncation as well as the rounding of the coefficients of the continuous piecewise affine payoff function.
This together with a rigorous error and complexity analysis when applying the modified iterative quantum amplitude estimation algorithm \cite{fukuzawa2022modified} allows us to control the output error of our algorithm to be bounded by the pre-specified accuracy level $\varepsilon \in (0,1)$, while bounding its computational complexity; we refer to Theorem~\ref{main theorem} for the precise statement of our main result. 

In particular, we prove that the computational complexity of our algorithm only grows polynomially in the space dimension $d$ of the Black-Scholes PDE and in the (reciprocal of the) accuracy level $\varepsilon$. 
This is in line with classical (i.e.\ non-quantum) Monte Carlo methods, where the number of Monte Carlo samples required in order to approximate the solution of the Black-Scholes PDE only grows polynomially in the dimension $d$ and $\varepsilon^{-1}$, see, e.g., \cite{HutzenthalerJentzenVonWurstemberger}.
Moreover, we show that for payoff functions which are bounded, our algorithm indeed has a \textit{speed-up} compared to classical Monte Carlo methods.
To the best of our knowledge, this is the first work in the literature which 
 provides a rigorous and comprehensive mathematical error and complexity analysis for a quantum Monte Carlo algorithm which approximately solves multidimensional PDEs. We refer to Remark~\ref{rem:complexity} for a detailed discussion of the complexity analysis.

Furthermore, we provide numerical simulations in two dimensions for three different payoff functions. To that end, we developed a package we named \texttt{qfinance} within the  \texttt{Qiskit} framework tailored to price continuous piecewise affine options with respect to the Black-Scholes model. Moreover, we discuss the potential extension of the numerical simulations to arbitrary space dimension.

The rest of this paper is organized as follows. 
In Section~\ref{sec:SettingMainResults}, we introduce the main setting of this paper,  present  our algorithm, and state our main theorem, as well as provide a detailed discussion of our complexity analysis.
In Section~\ref{sec: numerics} we present our numerical simulations in two space dimensions as well as discuss their potential extension to higher dimensions.
In Section~\ref{sec:quantum_circuits}, we introduce and analyze all relevant quantum circuits we need in our quantum Monte Carlo algorithm.
In Section~\ref{section: error estimates}, we provide a detailed error analysis of the steps of our algorithm outlined in Section~\ref{section: outline of algorithm}.
Finally, in Section~\ref{sec:AlgoandProof}, we provide the proof of Theorem~\ref{main theorem}. 
\\

\noindent
\textbf{Notation.} We denote the set of real numbers and  positive real numbers by $\R$ and $\R_+:= (0,\infty)$, respectively. The set of natural numbers is denoted by $\N:= \{1,2,\ldots\}$, and we use $\N_0:= \N \cup \{0\}$. The set of complex numbers is denoted by $\C$, and we define $\mathrm{i} := \sqrt{-1}$. Moreover, we denote by $I_2$ 
the  identity matrix in $\C^{2\times 2}$,
whereas for each $n \in \N$ and matrix $M$ we denote by $M^{\otimes n}$ its $n$-fold tensor product.
Furthermore,  for each $n \in \N$ we denote by $\mathcal{U}(2^n)$  the set of unitary matrices in  $\C^{2^n \times 2^n}$, i.e.\ matrices $U \in \C^{2^n \times 2^n}$ satisfying $UU^\dagger = U^\dagger U = I_{2}^{\otimes n}$, where  $U^\dagger$ denotes the \textit{conjugate transpose} of $U$.

\section{Setting and Main result}\label{sec:SettingMainResults}
\subsection{Black-Scholes PDE for option pricing}
Let $r \in (0,\infty)$ be the risk-free interest rate, let $T \in (0,\infty)$ be a finite time horizon determining the maturity, and let $d \in \N$ be the number of assets. We consider the multiple-asset Black-Scholes PDE
\begin{equation}\label{eqn: PDE}
    \frac{\partial u}{\partial t} + \frac{1}{2} \sum_{i,j=1}^d C_{ij} x_i x_j \frac{\partial^2 u}{\partial x_i \partial x_j} + \sum_{i=1}^d r x_i \frac{\partial u}{\partial x_i} - ru = 0, \quad \text{in } [0,T) \times \R_+^d
\end{equation}
subjected to a terminal condition $u(T,\cdot) = h(\cdot)$. Here, $h: \R_+^d \to \R$ represents the payoff function and $u(t,\bm{x})$ represents the option price at time $t$ with spot price $\bm{x}$. The covariance matrix $\bm{C}=(C_{i,j})_{i,j=1}^d \in \R^{d\times d}$ is assumed to be symmetric positive definite with a Cholesky factorization $\bm{C} = \bm{\sigma} \bm{\sigma}^\top$, where $\bm{\sigma} \in \R^{d \times d}$ is the volatility matrix, 
so that there is a unique risk-neutral measure (see, e.g., \cite{campolieti2014financial}).
Note that 
the PDE \eqref{eqn: PDE} has a unique solution
whenever $h: \R_+^d \to \R$ is continuous and at most polynomially growing,
see, e.g., \cite[Proposition~2.23, Corollary 4.5]{grohs2018proof}.
 

\subsubsection{Geometric Brownian motion process for the price evolution of  multiple assets}
In the multidimensional Black-Scholes model,  the prices of the $d$ stocks under consideration are modeled by a multidimensional geometric Brownian motion (GBM) having constant growth rate and volatility, see, e.g., \cite{campolieti2014financial}. 
We briefly describe the dynamics of the geometric Brownian motion process for multiple assets.

Let $(\Omega,\mathcal{F},\Prob)$ be a probability space and let $\bm{W}=(W^1,\dots,W^d): [0,T] \times \Omega \to \R^d$ be a standard $d$-dimensional Brownian motion. For a volatility matrix $\bm{\sigma} \in \R^{d \times d}$ assumed to be invertible, let $\bm{\sigma}_1,\ldots,\bm{\sigma}_d \in \R^d$ denote the row vectors of matrix $\bm{\sigma}$, and let $\sigma_i := \lVert \bm{\sigma_i} \rVert_{\ell^2(\R^d)}$. Let $\bm{S}=(S^1,\dots,S^d):[0,T] \times \Omega \to \R_+^d$ be the stock price process governed by the following stochastic differential equation
\begin{equation}
\begin{aligned}
dS_t^{i} &= S_t^{i}\bigg(rdt + \sum_{j=1}^d \sigma_{ij} dW_t^{j}\bigg),\quad \text{for }  i=1,\ldots,d, \label{eqn: GBM SDE}
\end{aligned}
\end{equation}
with some initial spot price $\bm{S}_0 \in \R_{+}^d$. 
Here $\bm{S}_t=(S^1_t,\dots,S^d_t)$ represents the values of each stock $i=1,\dots,d$ at time $0\leq t \leq T$.
Let $\bm{R} = (R^{1},\ldots,R^{d}):[0,T] \times \Omega \to \R^d$ be the log-return process defined component-wise by $R_t^{i} = \ln(S_t^{i}/S_0^{i})$ for $i=1,\ldots,d$. It follows from It\^o's formula for all $t \in [0,T]$ that 
\begin{equation}
    d R_t^{i} = (r-\tfrac{1}{2}\sigma_i^2) dt + \sum_{j=1}^d \sigma_{ij} dW_t^{j}, \quad \text{for } i=1,\ldots,d, \label{eqn: log-transform process}
\end{equation}
with initial condition $R_0^i = 0$ for $i=1,\ldots,d$. Let $\bm{\hat{\mu}} = (\hat{\mu}_1,\ldots,\hat{\mu}_d) \in \R^d$ be a vector satisfying $\hat{\mu}_{i} = (r - \frac{1}{2}\sigma_i^2)$ for $i = 1,\ldots,d$. From equation \eqref{eqn: log-transform process}, it holds that $\bm{R}_T$ is a multivariate normal distribution with mean $T\bm{\hat{\mu}}$ and covariance $T\bm{C}$. Hence, by taking the inverse of the log transform, we observe that the law of the stock price process $\bm{S}_T = \bm{S}_0\exp(\bm{R}_T)$ is a multivariate log-normal distribution with log-mean $T\bm{\hat{\mu}}$ and log-covariance $T\bm{C}$. In general, for a given fixed initial condition $(t,x) \in [0,T] \times \R_+^d$, there is a well-known formula for the probability transition density function of $\bm{S}_T$, subjected to the condition that $\bm{S}_t = \bm{x}$. 

\begin{lemma}[Density formula]\label{lemma: density} 
Let $d \in \N$, let $\bm{x}=(x_1,\ldots,x_d) \in \R_+^d$, and let $t \in [0,T)$. Let $\bm{\mu} = (\mu_1,\ldots,\mu_d) \in \R^d$ be given by $\mu_i = \ln(x_i) + (r - \tfrac{1}{2}\sigma_i^2)(T-t)$ for $i=1,\ldots,d$.
Then, the stock price process $\bm{S}_T$ introduced by \eqref{eqn: GBM SDE} conditional on $\bm{S}_t = \bm{x}$ follows a multivariate log-normal distribution with log-mean $\bm{\mu}$ and log-covariance $(T-t)\bm{C}$, 
and the joint transition probability density function is given by
\begin{equation}
    p(\bm{y},T;\bm{x},t) := \frac{\exp( -\frac{1}{2(T-t)} (\log(\bm{y}) - \bm{\mu})^\top {\bm{C}}^{-1} (\log(\bm{y}) - \bm{\mu}))}{ (2\pi(T-t))^{d/2} (\det \bm{C})^{1/2}\prod_{i=1}^d y_i}  , \label{eqn: density formula}
\end{equation}
where for $\bm{y} = (y_1,\ldots,y_d) \in \R_+^d$, $\log(\bm{y}) \in \R^d$ is given by 
\begin{equation}
    (\log(\bm{y}))_i = \ln(y_i), \quad i = 1,\ldots,d. \label{eqn: density change variable}
\end{equation}
\begin{proof}
See, e.g., \citeauthor{campolieti2014financial} \cite[page 485-486]{campolieti2014financial}.
\end{proof}
\end{lemma}
Throughout the paper we impose the following assumption on the covariance matrix $\bm{C} \equiv \bm{C}_d$ in dependence of the dimension $d$.
\begin{assumption}[Covariance matrix]\label{assumption: cov matrix} 
    There is a constant $C_1 \in [1,\infty) $ not depending on the dimension $d\in \N$ such that for every $d\in \N$ the corresponding covariance matrix $\bm{C} \equiv \bm{C}_d = ((\bm{C}_d)_{i,j})_{i,j=1}^d \in \R^{d\times d}$ defined as in \eqref{eqn: PDE} satisfies for every $i,j=1,\dots, d$ that
    \begin{equation}
        \vert (\bm{C}_d)_{i,j}\vert \leq C_1. 
    \end{equation}
\end{assumption}

\subsubsection{Continuous piecewise affine payoff functions} \label{subsec:CPWA}
For any $d \in \N$, we consider a payoff function $h: \R^d_+ \to \R$, which takes the stock prices $\bm{S}_T \in \R_+^d$ at terminal time~$T$ as input. The option price $u(t,\bm{x}) \in \R$ at time $t \in [0,T)$ given that the spot price satisfies $\bm{S}_t = \bm{x} \in \R_+^d$ is characterized by the following Feynman-Kac formula (see, e.g., \cite[Equation (13.33)]{campolieti2014financial})
\begin{equation}\label{eqn: option price formula}
    u(t,\bm{x}) = e^{-r(T-t)}\E[h(\bm{S}_T)\mid \bm{S}_t = \bm{x} ] = e^{-r(T-t)}\int_{\R_+^d} h(\bm{y}) p(\bm{y},T;\bm{x},t)\,d\bm{y},
\end{equation}
where $p(\cdot,T;\bm{x},t)$ is the transition density formula given in Lemma \ref{lemma: density}. In this paper, we consider payoff functions restricted to the class of continuous piecewise affine functions\footnote{In particular, any continuous piecewise affine function is linearly growing, see, e.g., Lemma~\ref{lemma: lin growth}. Hence the PDE \eqref{eqn: PDE} has a unique solution.}. This type of function represents most of the payoff functions seen in financial mathematics literature \cite{neufeld2020modelfree}; see also the examples below. 

\begin{definition}[continuous piecewise affine payoff]Let $d \in \N$. A function $h: \R^d_+ \to \R$ is a continuous piecewise affine function if it can be represented as 
\begin{equation}
    h(\bm{x}) = \sum_{k=1}^{K} \xi_k \max\{\bm{a}_{k,l} \cdot \bm{x} + b_{k,l}: l = 1,\ldots,I_k\}, \label{eqn: CPWA payoff}
\end{equation}
where $K ,I_k \in \N$ and $\xi_k \in \{-1,1\}$ for $k = 1,\ldots,K$, and where $\bm{a}_{k,l} \in \R^d$, $b_{k,l} \in \R$ for $k =1,\ldots,K$, $l = 1,\ldots, I_k$.
\end{definition}
Throughout the paper we impose the following assumptions on the continuous piecewise affine payoff function $h:\R^d_+ \to \R$ in dependence of the dimension $d$.
\begin{assumption}[continuous piecewise affine]\label{assumption: CPWA} 
    There is a constant $C_2 \in [1,\infty)$ not depending on the dimension $d\in \N$ such that the continuous piecewise affine function $h:\R^d_+ \to \R$ defined as in \eqref{eqn: CPWA payoff} satisfies both 
    \begin{equation}\label{eqn: bound on cpwa coefficients}
        \max\left\{\Vert \bm{a}_{k,l}\Vert_\infty, \vert b_{k,l}\vert: k=1,\ldots,K,l=1,\ldots,I_k\right\}\leq C_2
    \end{equation}
    and
    \begin{equation}
        K \cdot \max\{I_1,\ldots,I_K\}\leq C_2 d.
    \end{equation}
\end{assumption}

\begin{example}\label{example: call options}
The list below contains examples showcasing that many popular payoff functions $h: \R^d_+ \to \R$ used in finance are continuous piecewise affine, see also \cite[Appendix EC.2]{neufeld2020modelfree}. In the following, we denote $\bm{e}_i$ the $i$-th unit vector in~$\R^d$. 
\begin{enumerate}
    \item Call option on the $i$-th asset with strike $\kappa$: setting $K = 1$, $\xi_1 = 1$, $I_1 = 2$, $\bm{a}_{1,1} = \bm{e}_i$, $\bm{a}_{1,2} = \bm{0}$, $b_{1,1} = - \kappa$, $b_{1,2} = 0$, we have 
    \begin{equation}
        h(\bm{x}) = \max\{ x_i - \kappa, 0 \}.
    \end{equation}
    \item Basket call option with weights $\bm{w}$ and strike $\kappa$: setting $K = 1$, $\xi_1 = 1$, $I_1 = 2$, $\bm{a}_{1,1} = \bm{w}$, $\bm{a}_{1,2} = \bm{0}$, $b_{1,1} = - \kappa$, $b_{1,2} = 0$, we have 
    \begin{equation}
        h(\bm{x}) = \max\{ \bm{w} \cdot \bm{x} - \kappa, 0\}.
    \end{equation}
    \item Spread call option: using setting 2., but by replacing $\bm{a}_{1,1}$ with $\bm{a}_{1,1} = \sum_{i \in \mathcal{I}} \bm{e}_i - \sum_{j \in \mathcal{I}'} \bm{e}_j$ for $\mathcal{I},\mathcal{I}' \subset \{1,\ldots,d\}$ and $\mathcal{I} \cap \mathcal{I}' = \emptyset$, we have 
    \begin{equation}
        h(\bm{x}) = \max\bigg\{ \sum_{i \in \mathcal{I}} x - \sum_{j \in \mathcal{I}'} x_j - \kappa, 0\bigg\}.
    \end{equation}

    \item Call-on-max option with strike $\kappa$: setting $K = 1$, $\xi_1 = 1$, $I_1 = d+1$, $\bm{a}_{1,j} = \bm{e}_j$, $b_{1,j} = - \kappa$ for all $j = 1,\ldots,d$, $\bm{a}_{1,d+1} = \bm{0}$, $b_{1,d+1} = 0$, we have 
    \begin{equation}
        h(\bm{x}) = \max\{ x_1 - \kappa,\ldots,x_d - \kappa, 0 \}.
    \end{equation}
    
    \item Call-on-min option with strike $\kappa$: setting $K = 2$, $\xi_1 = 1$, $\xi_2 = -1$, $I_1 = d$, $I_2= d+1$, $\bm{a}_{1,j} = \bm{a}_{2,j} = - \bm{e}_j$, $b_{1,j} = b_{2,j} = \kappa$ for all $j = 1,\ldots,d$, $\bm{a}_{1,d+1} = \bm{0}$, $b_{1,d+1} = 0$, we have 
    \begin{equation}
        h(\bm{x}) = \max\{ \kappa - x_1,\ldots, \kappa - x_d, 0 \} - \max\{\kappa - x_1,\ldots, \kappa - x_d \}
        =\max\{\min\{x_1,\dots,x_d\}-\kappa,0\}.
    \end{equation}
    
    \item Best-of-call option with strikes $\kappa_1,\ldots,\kappa_d$: setting $K = 1$, $\xi_1 = 1$, $I_1 = d+1$, $\bm{a}_{1,j} = \bm{e}_j$, $b_{1,j} = - \kappa_j$ for all $j = 1,\ldots,d$, $\bm{a}_{1,d+1} = \bm{0}$, $b_{1,d+1} = 0$, we have 
    \begin{equation}
        h(\bm{x}) = \max\{x_1 - \kappa_1,\ldots,x_d - \kappa_d,0 \}.
    \end{equation}
\end{enumerate} 
We note that all of the above examples satisfy Assumption \ref{assumption: CPWA} provided that the coefficients $(\bm{a}_{k,l},b_{k,l})$ are bounded  by some constant $C_2 \in [1,\infty)$ uniformly in the dimension $d$, c.f.\ \eqref{eqn: bound on cpwa coefficients}. 
\end{example}

\subsection{Brief Introduction to Quantum Computing}
In this section, we briefly recall notions and concepts in quantum computing which we frequently use in this paper. Classic references for this subject are the textbooks by \citeauthor{jacquier2022quantum}  \cite{jacquier2022quantum} and \citeauthor{QC0_textbook} \cite{QC0_textbook}.
\\

\noindent
Let $\mathcal{H}$ be a finite dimensional complex Hilbert space. A vector $v \in \mathcal{H}$, also referred as a \textit{state}, is denoted by the \textit{ket} notation $\ket{v}$. The inner product of two vectors $v,w \in \mathcal{H}$ is denoted by the \textit{bra-ket} notation $ \braket{v}{w} := \langle v,w\rangle  \in \C$. Elements $u \in \mathcal{H}^*$ of the dual space $\mathcal{H}^*$ are denoted by the \textit{bra} notation $\bra{u}$. The action of the dual vector $u \in \mathcal{H}^*$ on a vector $v \in \mathcal{H}$ is also denoted by the \textit{bra-ket} notation $\braket{u}{v}$. The action of a linear operator $A: \mathcal{H} \to \mathcal{H}$ on a vector $\ket{v}$ is denoted by $A\ket{v}$. The operator $A$ acts on dual vectors $\bra{u} \in \mathcal{H}^*$ by the rule $(\bra{u}A) \ket{v} := \bra{u}(A\ket{v}) := \langle u,Av \rangle$ for all $v \in \mathcal{H}$ which is also denoted by $\mel{u}{A}{v}$. 

For the Hilbert space $\mathcal{H} = \C^2$, we consider the $n$-fold tensor product Hilbert space $\mathcal{H}^{\otimes n} := \mathcal{H} \otimes \cdots \otimes \mathcal{H} \simeq \C^{2^n}$. We denote a state ${\psi} \in \mathcal{H}^{\otimes n}$ by $\ket{\psi}_n$, where the subscript $n$ emphasizes the \textit{($\log_2$)-dimension} of the tensor product Hilbert space $\mathcal{H}^{\otimes n}$. We use the orthonormal basis $\Beta_n = \{\ket{i}_n : i=(i_{1},i_{2},\ldots,i_n) \in \{0,1\}^n\} \subset \C^{2^n}$, where $\ket{i}_n := \ket{i_{1}} \otimes \ket{i_{2}} \otimes \cdots \otimes\ket{i_n} := \ket{i_{1}}\ket{i_{2}}\cdots\ket{i_n} $, $i \in \{0,1\}^n$. The basis $\Beta_n$ is referred as the \textit{computational basis} in the literature. 
For example the standard orthonormal basis $\{\ket{0},\ket{1}\} \subset \mathcal{H} = \C^2$ is given by 
\begin{equation}
    \ket{0} := \begin{bmatrix} 1 \\ 0\end{bmatrix}, \quad \ket{1} := \begin{bmatrix} 0 \\ 1 \end{bmatrix}.
\end{equation} 

 The (single) \textit{qubit}  
 represents a unit of quantum information. 
An arbitrary $n$-qubit state $\ket{\psi}_n$ is represented by a normalized vector in $\C^{2^n}$ which can be described by a $\C$-linear combination in the computational basis $\Beta_n$, i.e.
\begin{equation}\label{eqn: n qubit superposition}
    \ket{\psi}_n = \sum_{i \in \{0,1\}^n} \alpha_i \ket{i}_n,\quad \text{with }\sum_{i \in \{0,1\}^n} \lvert \alpha_i \rvert^2 = 1. 
\end{equation}
The coefficients $\alpha_i \in \C$ in \eqref{eqn: n qubit superposition} are referred as \textit{probability amplitudes} (or simply \textit{amplitudes}) due to the fact that for each $\ket{i}_n \in \mathcal{B}_n$ we have that the square amplitude $\lvert \alpha_i \rvert^2 = \lvert \braket{i}{\psi}_n\rvert^2$ is the probability of observing that the state $\ket{\psi}_n$ collapses during a projective measurement to the state  $\ket{i}_n$, according to \textit{Born's rule}; see, e.g., \cite[Chapter~2.5]{marinescu2011classical}.

Qubits on a quantum computer are manipulated by a sequence of  \textit{quantum gates}. The state evolution of qubits, according to the axioms of quantum mechanics, is unitary. \textit{Elementary} quantum gates act on either one or two qubits at a time. 
A \textit{quantum circuit} is a finite sequence of  compositions (and tensor products) of quantum gates and wires (where each wire represents the identity operator $I_2 \in \C^{2\times 2}$). Thus, quantum circuits correspond to unitary operators, represented by 
\textit{unitary matrices} in $\mathcal{U}(2^n)\subset \C^{2^n\times 2^n}$, 
where $n$ is the number of qubits on which the circuit acts.

Let us now define the set of elementary quantum gates, which allows us to describe an \textit{algebraic definition} for a quantum circuit.

\begin{definition}[Elementary quantum gate set]\label{def: elementary gate set}
We call any element of the set of 1-qubit and 2-qubits quantum gates $\mathbb{G} \subset \C^{2 \times 2} \cup \C^{2^2 \times 2^2}$ defined by
\begin{equation}\label{eqn: elementary gate set}
    \mathbb{G} := \Big\{X,Y,Z,H,
    P(\phi), R_x(\theta),R_y(\theta),R_z(\theta): \phi \in (0,2\pi), \theta \in (0,4\pi) \Big\} \cup \Big\{\CNOT,\SWAP \Big\}
\end{equation}
as elementary gates, 
where $X := 
\begin{bmatrix}0 & 1 \\ 1 & 0 \end{bmatrix},\quad 
Y := 
\begin{bmatrix}0 & -\mathrm{i} \\ \mathrm{i} & 0 \end{bmatrix}, \quad 
Z := 
\begin{bmatrix}1 & 0 \\ 0 & -1 \end{bmatrix}$
denote the Pauli transformation gates, 
 $H := \frac{1}{\sqrt{2}}\begin{bmatrix}1 & 1 \\ 1 & -1 \end{bmatrix}$ denotes the Hadamard gate, 
 $P(\phi) := \begin{bmatrix}1 & 0 \\ 0 & e^{\mathrm{i}\phi} \end{bmatrix}$,  $\phi \in (0,2\pi)$, denote the phase shift gates, and 
 $R_x(\theta) := \exp(-\mathrm{i}\theta X/2)$, 
 $R_y(\theta) := \exp(-\mathrm{i}\theta Y/2)$,
 $R_z(\theta) := \exp(-\mathrm{i}\theta Z/2)$, $\theta \in (0,4\pi)$, denote the three families of rotation gates,
  whereas      $\CNOT := \begin{bmatrix} 1 & 0 & 0 & 0 \\ 0 & 1 & 0 & 0 \\ 0 & 0 & 0 & 1 \\ 0 & 0 & 1 & 0 \end{bmatrix}$ \
  and \ 
  $    \SWAP := \begin{bmatrix} 1 & 0 & 0 & 0 \\ 0 & 0 & 1 & 0 \\ 0 & 1 & 0 & 0 \\ 0 & 0 & 0 & 1 \end{bmatrix}$
  denote the Controlled-NOT gate and the swap gate, respectively.
  

\end{definition}

\begin{remark}[Universality of $\mathbb{G}$]
    We note that the set of elementary quantum gates $\mathbb{G}$ in \eqref{eqn: elementary gate set} is \textit{universal} in the sense of the Solovay-Kitaev theorem \cite[Appendix 3]{QC0_textbook}. Moreover, the set $\mathbb{G}$ consists of quantum gates used in practical quantum computing softwares, such as IBM's Qiskit \cite{Qiskit} and Google's Cirq \cite{cirq_developers_2022_7465577}.
\end{remark}

\begin{definition}[Quantum circuit]\label{def: quantum circuit}
Let $n,M,L \in \N$. A quantum circuit $\mathcal{Q}$ acting on $n$ qubits is a $2^n\times 2^n$ unitary matrix of the form\footnote{Note that for any $n \in \N$ and any $\mathcal{Q}_1\dots,\mathcal{Q}_L \in \mathcal{U}(2^n)$, 
	their product is defined by $\prod_{l=1}^L \mathcal{Q}_l :=  \mathcal{Q}_L \mathcal{Q}_{L-1}\cdots \mathcal{Q}_2\mathcal{Q}_1\in \mathcal{U}(2^n)$.}
\begin{equation}
    \mathcal{Q} = \prod_{l=1}^L \left( G_{l,1} \otimes G_{l,2} \otimes \cdots \otimes G_{l,n_l}\right) \in \mathcal{U}(2^n),
\end{equation}
where $(G_{l,1},G_{l,2},\ldots,G_{l,n_l})_{l=1}^L \subset \mathbb{G} \cup \{I_2\}$, and  $\prod_{m=1}^{n_l} \dim(G_{l,m}) = 2^n$ for all $l=1,\ldots,L$. We define the following quantum circuit complexities\footnote{The indicator function $\mathbbm{1}_{\mathbb{S}}$ for a non-empty subset ${\mathbb{S}} \subset {\mathbb{H}}$ is the unique function satisfying $\mathbbm{1}_{\mathbb{S}}:{\mathbb{H}} \to \{0,1\}$ such that $\mathbbm{1}_{\mathbb{S}}(x) = 1$ if $x \in {\mathbb{S}}$ and $\mathbbm{1}_{\mathbb{S}}(x) = 0$ if $x \not \in {\mathbb{S}}$.  }:
\begin{itemize}
    \item $M=$ number of elementary quantum gates used to construct quantum circuit $\mathcal{Q}$, i.e.
    \begin{equation}
        M := \sum_{l=1}^L \sum_{j=1}^{n_l} \mathbbm{1}_{\mathbb{G}}(G_{l,j}),
    \end{equation}
    \item $n =$ number of qubits used in quantum circuit $\mathcal{Q}$, and
    \item $L=$ depth of quantum circuit $\mathcal{Q}$.
\end{itemize}
\end{definition}

\begin{remark}[Depth]
    For any quantum circuit $\mathcal{Q}$, we may bound its depth complexity by the number of elementary gates in the quantum circuit. In this paper, we focus only on the number of elementary gates and qubits used in constructing quantum circuits. 
\end{remark}
\begin{remark}[Ancilla qubits]
In most quantum circuits, auxiliary qubits are used as additional memory to perform necessary quantum computations but may not be used for the output. One calls these qubits \textit{ancilla qubits} or simply \textit{ancillas} and  denotes them by $\ket{\anc}_\star$, where the subscript $\star$ usually indicates in the literature that the amount of ancillas used are not specified precisely.  The \textit{complexity of a quantum circuit} is usually described by the number of elementary quantum gates used, and the number of qubits and ancilla qubits used. For simplicity, we count both qubits and ancilla qubits together as the number of qubits used in a circuit. 
\end{remark}


\subsection{Quantum amplitude estimation algorithms}\label{section: QAE}
In this section, we briefly review quantum algorithms for solving the quantum amplitude estimation (QAE) problem. Given an unitary operator $\mathcal{A}$ acting on $n+1$ qubits, defined by 
\begin{equation}\label{eqn: A00}
    \mathcal{A}\ket{0}_n\ket{0} = \sqrt{1-a}\ket{\psi_0}_n \ket{0} + \sqrt{a}\ket{\psi_1}_n \ket{1},
\end{equation}
where the so-called bad state is $\ket{\psi_0}_n\ket{0}$ and the good state is $\ket{\psi_1}_n\ket{1}$, \citeauthor{QC1_QAE} introduced in \cite{QC1_QAE} the amplitude estimation problem where the goal is to estimate the unknown amplitude $a \in [0,1]$, which is the probability of measuring the good state $\ket{\psi_1}_n\ket{1}$ according to Born's rule. Let $a = \sin^2(\theta_a)$ for some $\theta_a \in [0,\tfrac{\pi}{2}]$ so that we can rewrite \eqref{eqn: A00} as 
\begin{equation}\label{eqn: A00'}
    \mathcal{A}\ket{0}_n\ket{0} = \cos(\theta_a)\ket{\psi_0}_n \ket{0} + \sin(\theta_a)\ket{\psi_1}_n\ket{1}.
\end{equation}
To achieve a quantum speed-up, they introduced in \cite{QC1_QAE} the amplitude amplification operator (also known as the Grover operator)
\begin{equation}\label{eqn: Grover operator}
    \mathcal{Q} := \mathcal{A}S_0\mathcal{A}^\dagger\mathcal{S}_{\psi_0},
\end{equation}
where $\mathcal{S}_0 := I_2^{\otimes n+1} - 2\ket{0}_{n+1}\bra{0}_{n+1}$, and $\mathcal{S}_{\psi_0} := I_2^{\otimes n}\otimes Z$. Note that for every $k\in \N$, it holds that
\begin{equation}\label{eqn: Q^kA00}
    \mathcal{Q}^k\mathcal{A}\ket{0}_n \ket{0} = \cos((2k+1)\theta_a)\ket{\psi_0}_n\ket{0} + \sin((2k+1)\theta_a)\ket{\psi_1}_n\ket{1},
\end{equation}
(c.f. \cite[Section 2]{QC1_QAE}). We observe that measuring \eqref{eqn: Q^kA00} boosts the probability of obtaining the good state to  $\sin^2((2k+1)\theta_a)$ which is larger than $\sin^2(\theta_a)$ when measuring \eqref{eqn: A00'} directly, provided $\theta_a$ is sufficiently small so that $(2k+1)\theta_a \leq \tfrac{\pi}{2}$. Using quantum Fourier transform and a number of multi-controlled operators for $\mathcal{Q}^k$, \citeauthor{QC1_QAE} designed the QAE  algorithm \cite[Algorithm(Est\_Amp)]{QC1_QAE} to estimate $a$ with high probability using only $O(\eps^{-1})$ queries of $\mathcal{A}$ (see \cite[Theorem 12]{QC1_QAE}). It is noted that \citeauthor{QC1_QAE}'s algorithm enables a quadratic speed-up for many approximation problems which are solved classically by Monte Carlo simulations under the assumption that the corresponding distribution can be uploaded in rotated form \eqref{eqn: A00'}. However, due to difficulties in implementing large number of controlled unitary operators as well as the quantum Fourier transform (QFT) operator on quantum computers, several variants of the QAE algorithm without using QFT have been proposed recently; see e.g., \cite{aaronson2019quantum, giurgica2022low, grinko2021iterative, manzano2023real,Nakaji20,plekhanov2022variational, rall2022amplitude, suzuki2020amplitude, wie2019simpler,  zhao2022adaptive}. In this paper, we use the modified iterative quantum amplitude estimation algorithm (Modified IQAE) \cite[Algorithm 1]{fukuzawa2022modified}  introduced recently by \citeauthor{fukuzawa2022modified} \cite{fukuzawa2022modified}, which is a modification of the IQAE algorithm presented by \citeauthor{suzuki2020amplitude} \cite{suzuki2020amplitude}. In brief, the Modified IQAE algorithm consists of several rounds where for each round $i$, the algorithm maintains a confidence interval $[\theta_l^{(i)},\theta_u^{(i)}]$ so that $\theta_a$ lies inside this interval with a certain probability. The confidence interval is narrowed in each subsequent round until the terminating condition $\theta_u-\theta_l <2 \eps$ for prespecified $\eps \in (0,1)$ is satisfied. The return output of the Modified IQAE algorithm is the confidence interval $[a_l,a_u]$ for $a$, where $a_l := \sin^2(\theta_l)$ and $a_u := \sin^2(\theta_u)$. The following statement is a direct rephrase of the main results in \cite{fukuzawa2022modified}.

\begin{proposition}\label{prop: m-IQAE}
    Let $\alpha, \eps \in (0,1)$ and let  $\mathcal{A}$ be an $(n+1)$-qubit quantum circuit satisfying 
    \begin{equation}\label{rotated_form}
        \mathcal{A} \ket{0}_{n}\ket{0} = \sqrt{1-a} \ket{\psi_0}_n\ket{0} + \sqrt{a} \ket{\psi_1}_n\ket{1}
    \end{equation}
    where $n \in \N$, $\ket{\psi_0}_n,\ket{\psi_1}_n$ are normalized states, $a \in [0,1]$, and where $\mathcal{A}$ can be constructed with $N_\mathcal{A} \in \N$ number of elementary gates. Then, the following holds.
    \begin{enumerate}
        \item The Modified IQAE Algorithm \cite[Algorithm 1]{fukuzawa2022modified} outputs a confidence interval $[a_l,a_u]$ that satisfies 
        \begin{equation}
           a \not \in [a_l,a_u], \quad \mbox{with probability at most } \alpha,
        \end{equation}
        where $0\leq a_u-a_l < 2\eps$. In particular, the estimator $\widehat{a} := \frac{a_u+a_l}{2}$ of $a$ satisfies
        \begin{equation}
        	\vert a - \widehat{a} \vert < \eps, \quad \mbox{with probability at least } 1-\alpha,
        \end{equation}
        \item the Modified IQAE Algorithm uses at most 
        \begin{equation}\label{eqn: bound on queries of A in mIQAE}
            \frac{62}{\eps}\ln\left(\frac{21}{\alpha}\right)
        \end{equation}
        applications of $\mathcal{A}$, and 
        \item the Modified IQAE Algorithm uses $n+1$ qubits and requires at most 
        \begin{equation}\label{eqn: number of gates for m-IQAE}
            \frac{\pi}{4\eps}(8n^2 + 23 + N_{\mathcal{A}})
        \end{equation}
        number of elementary gates.
    \end{enumerate}
    \begin{proof}
        Item 1.\ and Item 2.\ are proven in \cite[Theorem 3.1]{fukuzawa2022modified} and  \cite[Lemma 3.7]{fukuzawa2022modified}, respectively. For Item~3., we note that \cite[Algorithm 1 Modified IQAE]{fukuzawa2022modified} uses quantum circuits $\mathcal{Q}^k\mathcal{A}$ which is defined by \eqref{eqn: Q^kA00}, where $k \in \N$. Let us construct the operator $\mathcal{Q}$ (c.f. \eqref{eqn: Grover operator}), as outlined in \cite[Section 3]{rao2020quantum}. We note that the operator $\mathcal{S}_{\psi_0}$ be constructed using one $Z$ gate. 
        By direct computation (or see \cite[Figure 5]{rao2020quantum}), the operator $\mathcal{S}_0$ satisfy the identity
        \begin{equation}\label{eqn: S_0}
            \mathcal{S}_0 = X^{\otimes (n+1)} (I_2^{\otimes n} \otimes H)C^n(X)(I_2^{\otimes n} \otimes H)X^{\otimes (n+1)},
        \end{equation}
        where $C^n(X)$ is the generalized version of Toffoli gate, which uses the first $n$-qubits for control. The multi-control gate $C^n(X)$ is constructed as a quantum circuit in \cite[Theorem 2]{saeedi2013linear} using $2n^2-6n+5$  controlled $X$-rotation gates $CR_x(\theta)$, where each $CR_x(\theta)$ gate can be constructed with 4 elementary gates by the following definition
        \begin{equation}
            CR_x(\theta) = (I_2\otimes R_x(\tfrac{\theta}{2}))\CNOT (I_2\otimes R_x(\tfrac{\theta}{2})) \CNOT,
        \end{equation}
    see also the proof of Lemma \ref{lemma: controlled y-rotation}.
        Thus, by \eqref{eqn: S_0}, the operator $\mathcal{S}_{0}$ can be constructed using 
        \begin{equation}
            \begin{split}
                (n+1) + 1 + 4(2n^2-6n+5) + 1+ (n+1)\leq 8n^2 + 22
            \end{split}
        \end{equation}
        elementary gates. Finally, since the set of elementary gates is closed under  inversion, the number of elementary gates used to construct $\mathcal{A}^\dagger$ is the same as the number of elementary gates used for $\mathcal{A}$. Hence, the number of elementary gates used to construct $\mathcal{Q}$ is at most 
        \begin{equation}
            N_\mathcal{A}+1+N_\mathcal{A} + (8n^2+22) = 8n^2 + 23 + 2N_\mathcal{A}.
        \end{equation}
        Next, by \cite[Lemma 3.1]{fukuzawa2022modified}, the integer $k$ satisfies the bound $2k+1 \leq \tfrac{\pi}{4\eps}$. Thus,  the number of elementary gates used to construct the operator $\mathcal{Q}^k\mathcal{A}$ is at most 
        \begin{equation}
            \begin{split}
                &k(8n^2+23+2N_\mathcal{A})+N_\mathcal{A} \\
                &=k(8n^2+23) + (2k+1)N_\mathcal{A} \\
                &\leq (2k+1)(8n^2+23+N_\mathcal{A}) \\
                &\leq \frac{\pi}{4\eps}(8n^2+23+N_\mathcal{A}).
            \end{split}
        \end{equation}
    \end{proof}
\end{proposition}

\begin{remark}\label{remark: circuit Q^kA}
  Let us remark the following on the Modified IQAE algorithm \cite{fukuzawa2022modified}. 
    \begin{enumerate}
        \item The total number of rounds $t \in \N$ in \cite[Algorithm 1 Modified IQAE]{fukuzawa2022modified} is bounded by $\log_3(\tfrac{\pi}{4\eps})$, see \cite[Section 3.1]{fukuzawa2022modified}. For each round $i=1,\ldots,t$, the quantum circuit $\mathcal{Q}^{k_i}\mathcal{A}$ is prepared on a quantum computer, where each $k_i \in \N$ are found recursively by using the subroutine \cite[Algorithm 2, FindNextK]{fukuzawa2022modified}. Note that each of these quantum circuits require the same number of qubits as with the quantum circuit $\mathcal{A}$. Moreover, each $k_i$ satisfy $2k_i+1 \leq \tfrac{\pi}{4\eps}$, \cite[Lemma 3.1]{fukuzawa2022modified}.
        \item Since the number $k_t$ is the maximum among the $k_i$'s, we infer that \cite[Algorithm 1, Modified IQAE]{fukuzawa2022modified} requires the number of elementary gates used to construct the quantum circuit $\mathcal{Q}^{k_t}\mathcal{A}$ on a quantum computer in order to run the Modified IQAE algorithm, which can be bounded by \eqref{eqn: number of gates for m-IQAE}.
        
        \item The query complexity (i.e.\ number of applications) of $\mathcal{A}$ in Proposition \ref{prop: m-IQAE} is defined to be the number of times the operator $\mathcal{Q}$ is applied in the algorithm, which is 
        \begin{equation}
            \sum_{i=1}^t k_i N_i, 
        \end{equation}
        where $N_i$ is the number of measurements made on $\mathcal{Q}^{k_i}\mathcal{A}\ket{0}_n\ket{0}$ in round $i$. Hence, the query complexity of $\mathcal{A}$ can be interpreted as the computational running time for the Modified IQAE algorithm. It was shown in \cite[Lemma 3.7]{fukuzawa2022modified} that this number is bounded by \eqref{eqn: bound on queries of A in mIQAE}. 
        \item We emphasize that other versions of the QAE (or IQAE) algorithm offer essentially the same query complexity (i.e.\ $O(\frac{1}{\eps}\ln(\frac{1}{\alpha}))$ up to logarithmic factors of $\eps^{-1}$). In this paper, we chose the Modified IQAE for the quantum ampltitude estimation subroutine in Algorithm~\ref{quantum algorithm} since the bounds on the query complexities in \cite{fukuzawa2022modified} were explicit. We refer the reader to Section 3.2 in \cite{manzano2023} for a detailed comparison for the query complexities of the different QAE algorithms that are available in the literature. 
    \end{enumerate}
\end{remark}

\subsection{Algorithm~\ref{quantum algorithm} and main result}\label{subsec:mainresult}

In this section, we first present our quantum Monte Carlo algorithm named Algorithm~\ref{quantum algorithm} to solve  Black-Scholes PDEs \eqref{eqn: PDE} with corresponding continuous piecewise affine payoff function \eqref{eqn: CPWA payoff}. Moreover, we then outline Algorithm~\ref{quantum algorithm} and present our main result in Theorem~\ref{main theorem}, namely a convergence and complexity analysis of our algorithm.

\vspace*{0.2cm}
\begin{algorithm}[H]
	\KwIn{ $\eps \in (0,1)$, $\alpha \in (0,1)$, $d \in \N$, $r,T \in (0,\infty)$, $(t,\bm{x}) \in [0,T) \times \R_+^d$, covariance matrix $\bm{C}_d \in \R^{d \times d}$, and continuous piecewise affine function 
		\[\R_+^d \ni \bm{x} \mapsto h(\bm{x}) = \sum_{k=1}^{K} \xi_k \max\{\bm{a}_{k,l} \cdot \bm{x} + b_{k,l}: l = 1,\ldots,I_k\} \in \R\]}
	\KwOut{$\widetilde{U}_{t,\bm{x}} \in \R$}
	\nl Set $C_1,C_2,C_3 \in [1,\infty)$ to be the constants given by Assumption \ref{assumption: cov matrix}, Assumption \ref{assumption: CPWA}, and Assumption \ref{assumption: distribution loading}, respectively.\\
	\nl Set $$
	n_1 :=  \lceil n_{1,d,\eps}\rceil,
	\quad
	n_2 := 1 + \lceil \log_2(C_2) \rceil, 
	\quad
	m_1 := \lceil m_{1,d,\eps}\rceil,
	\quad
	m_2 := \lceil m_{2,d,\eps} \rceil,$$ where $n_{1,d,\eps}, m_{1,d,\eps}, m_{2,d,\eps}$ are defined in \eqref{eqn: step6-n_1}-\eqref{eqn: step6-m_2} in Proposition \ref{proposition: final error estimate}. \\
	\nl Set
	$$N = \eqref{eqn: def N, q_k}, \quad \gamma = \eqref{eqn: step6-gamma}, \quad \mathfrak{s} = \eqref{eqn: step6-s},  \quad  \text{and} \quad \widetilde{\bm{a}}_{n_2,m_2,k,l},\ \widetilde{b}_{n_2,m_2,k,l} = \eqref{eqn: approx payoff1} \ \text{ for } \ k=1,\dots K,\ l=1,\dots I_k.$$\\
	\nl Construct probability distribution quantum circuit $\mathcal{P} \equiv \mathcal{P}_{d,\eps}$ using Assumption \ref{assumption: distribution loading} (with $n \leftarrow n_1$, $m \leftarrow m_1$, $\eps \leftarrow \frac{\eps}{6C_2^2 d^2 2^{n_1+1}}$ in the notation of Assumption \ref{assumption: distribution loading}).\\
	\nl Construct continuous piecewise affine payoff with rotation quantum circuit $\mathcal{R}_{h}$ given by Proposition \ref{prop: loading payoff circuit} (with $s \leftarrow \mathfrak{s}$, $a_{k,l,j} \leftarrow \mathrm{E}_{n_2,m_2}(\widetilde{\bm{a}}_{n_2,m_2,k,l,j})$, $b_{k,l} \leftarrow \mathrm{E}_{n_2,m_2}(\widetilde{b}_{n_2,m_2,k,l})$ for $k=1,\ldots,K$, $l=1,\ldots,I_k$, $j=1,\ldots,d$ in the notation of Proposition \ref{prop: loading payoff circuit}).\\
	\nl Construct the quantum circuit $\mathcal{A} = \mathcal{R}_h(\mathcal{P}\otimes I_2^{\otimes (N-d(n_1+m_1))})$ using the quantum circuits $\mathcal{R}_h$ and $\mathcal{P}$.\\
	\nl Set $\widehat{a} = \tfrac{a_u+a_l}{2}$ using the output $[a_l,a_u]$ from the modified iterative quantum amplitude estimation algorithm \cite[Algorithm 1 Modified IQAE]{fukuzawa2022modified} (with $\eps \leftarrow \tfrac{\eps\mathfrak{s}}{12}$, $\alpha \leftarrow \alpha$, $N_{\text{shots}} \leftarrow 1$, and  $\mathcal{A} \leftarrow \mathcal{A}$ in the notation of \cite[Algorithm 1]{fukuzawa2022modified}).\\
	\nl Return $\widetilde{U}_{t,\bm{x}} := \mathfrak{s}^{-1}\gamma e^{-r(T-t)}(2\widehat{a} - 1)$.
	\caption{Quantum algorithm for  solving Black-Scholes PDEs
		with continuous piecewise affine payoff functions}
	\label{quantum algorithm}
\end{algorithm} 
\subsubsection{Outline of Algorithm~\ref{quantum algorithm}}\label{section: outline of algorithm}
\begin{figure}[t]
	\centering
	\includegraphics[width=0.9\linewidth]{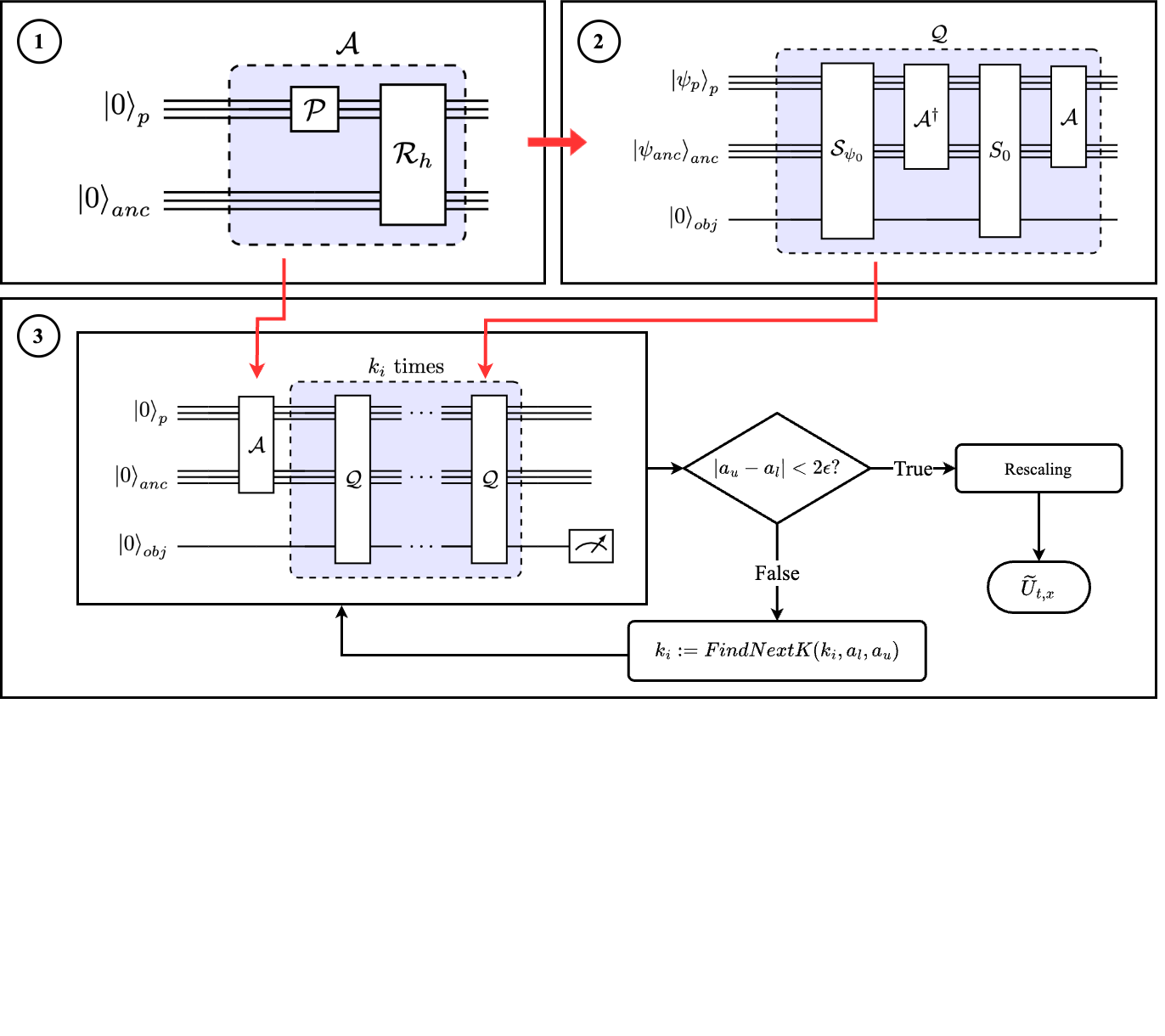}
	\vspace*{-4.2cm}
	\caption{Flowchart of Algorithm 1. (Top left) (1) Construction of the operator $\mathcal{A}$ using probability distribution operator $\mathcal{P}$ and continuous piecewise affine payoff operator $\mathcal{R}_h$. (Top right) (2) Construction of the Grover operator using the operator $\mathcal{A}$, oracle operator $\mathcal{S}_{\psi_0}$ and phase flip operator $\mathcal{S}_0$. (Bottom) (3) Illustration of the Modified Iterative Quantum Amplitude Estimation algorithm to produce the final estimate $\widetilde{U}_{t,\bm{x}}$.}
	\label{fig:diagram_of_algorithm_1}
\end{figure}

The steps of Algorithm~\ref{quantum algorithm} can be briefly described into four parts as follows:
\begin{enumerate}
    \item upload the transition probability function $p(\cdot,T;\bm{x},t)$ given in \eqref{eqn: density formula}, 
    \item upload the continuous piecewise affine payoff function $h:\R^d_+ \to \R$ given in \eqref{eqn: CPWA payoff}, 
    \item apply the Modified IQAE algorithm \cite[Algorithm 1]{fukuzawa2022modified} to obtain an estimated amplitude $\widehat{a} \in [0,1]$,
    \item rescale the estimated amplitude $\widehat{a}$ to output $\widetilde{U}_{t,\bm{x}}$  which approximates $u(t,\bm{x})$ defined in \eqref{eqn: option price formula}. 
\end{enumerate}
More precisely, in order to apply the Modified IQAE algorithm in step 3, we first need to upload both the probability distribution function and the payoff function as quantum circuits. For the first step, we need to truncate and discretize the distribution function on a grid. The corresponding parameters $n_1, m_1 \in \N$ encodes the grid $([-2^{n_1-1},2^{n_1-1}-2^{-m_1}]\cap 2^{-m_1}\Z)^d$, which is the support of the discretized distribution. Then, the truncated, discretized distribution can be uploaded approximately on the quantum computer using the quantum circuit $\mathcal{P}$, introduced in Section \ref{section: distribution quantum circuit}. The corresponding parameter $\gamma$ is needed to renormalize the amplitude coefficients in $\mathcal{P}\ket{0}_{d(n_1+m_1)}$ to $[0,1]$. We note that step 1 accounts for the truncation error, quadrature error, and distribution loading error in Section \ref{section: error estimates}, see Proposition \ref{prop: truncation error},  Proposition \ref{prop: quadrature error}, and Proposition \ref{prop: loading errors}, respectively. 

In step 2, we discretize the coefficients $(\bm{a}_{k,l},b_{k,l})$ of the continuous piecewise affine payoff function \eqref{eqn: CPWA payoff} onto a grid. The corresponding parameter $n_2\in \N$ encodes the bounds on the coefficients while the parameter $m_2 \in \N$ encodes the rounding off accuracy level for the coefficients.  We construct the quantum circuit $\mathcal{R}_h$ that encodes the approximated continuous piecewise affine payoff in a phase amplitude using controlled $Y$-rotation gates. The corresponding parameter $\mathfrak{s} \equiv s$ indicates the scaling parameter for the rotation circuit $\mathcal{R}_h$, see Proposition \ref{prop: loading payoff circuit}. We account for the errors for approximating the payoff function $h:\R^d_+ \to \R$ and the rotation errors in Proposition \ref{prop: quad-payoff} and Proposition \ref{prop: rotation error}. The quantum circuit $\mathcal{A}$ is then defined as the combination (i.e. composition) of the two quantum circuits $\mathcal{P}$ and $\mathcal{R}_h$ which enables us to apply quantum amplitude estimation (QAE) algorithms.

In step 3, we apply the Modified IQAE algorithm\footnote{We emphasize that the other QAE algorithms can also be used since the query complexities are essentially similar to \cite{fukuzawa2022modified}, c.f. Remark~\ref{remark: circuit Q^kA} Item 4.} \cite{fukuzawa2022modified} with $\mathcal{A}$ to get an output $\widehat{a} \in [0,1]$.  Then in the last step, since the estimated amplitude $\widehat{a}$ given by the Modified IQAE algorithm is between $0$ and $1$, we need to rescale this number to approximate the option price $u(t,\bm{x})$. We account for the QAE error in Proposition \ref{proposition: final error estimate}. 

\subsubsection{Main Theorem}

\begin{theorem} \label{main theorem}
    Let $\eps \in (0,1)$, $\alpha \in (0,1)$, $d \in \N$, $r,T \in (0,\infty)$, $(t,\bm{x}) \in [0,T) \times \R_+^d$, and covariance matrix $\bm{C}_d \in \R^{d \times d}$ be the input of Algorithm~\ref{quantum algorithm}. Let $u(t,\bm{x}) \in \R$ be the option price given by \eqref{eqn: option price formula} with continuous piecewise affine payoff $h:\R^d_+ \to \R$ given by \eqref{eqn: CPWA payoff}, let Assumption \ref{assumption: cov matrix}, Assumption \ref{assumption: CPWA}, and Assumption \ref{assumption: distribution loading}  hold with respective constants $C_1,C_2,C_3 \in [1,\infty)$, and let $\mathfrak{c},\mathfrak{C}_1,\mathfrak{C}_2,\mathfrak{C}_3 \in[2,\infty)$ be constants defined by
    \begin{align}
        &\mathfrak{c} := 2C_2^2 e^{4C_1^2T^2}e^{2rT}\max_{i=1,\ldots,d}\{1,x_i^2\},\\
        &\mathfrak{C}_1 := 648 C_2 \log_2(\mathfrak{c}),\label{eqn: constant frakC_1}\\
        &\mathfrak{C}_2 :=(1.6\times 10^8)C_2^4 C_3\big(27\log_2(\mathfrak{c})\big)^{\max\{3,2C_3\}},\label{eqn: constant frakC_2}\\
        &\mathfrak{C}_3 :=(6.1\times 10^5) C_2\mathfrak{c}^{\tfrac{3}{2}} \label{eqn: constant frakC_3}.                
    \end{align}
    Then, Algorithm~\ref{quantum algorithm} outputs $\widetilde{U}_{t,\bm{x}} \in \R$ which satisfies
    \begin{equation}\label{eqn: theorem estimate}
        \vert u(t,\bm{x}) - \widetilde{U}_{t,\bm{x}} \vert \leq \eps, \quad \text{with probability at least $1-\alpha$},
    \end{equation}
    where the number of qubits used in Algorithm~\ref{quantum algorithm} is at most
    \begin{equation}\label{eqn: bound on no. of qubits}
        \mathfrak{C}_1 d^2 (1+\log_2(d\eps^{-1})),
    \end{equation}
    the number of elementary gates used in Algorithm~\ref{quantum algorithm} is at most
    \begin{equation}\label{eqn: bound on no. of gates}
        \mathfrak{C}_2 d^{\max\{10.75,4.75+C_3\}}\eps^{-3} (1+\log_2(d\eps^{-1}))^{\max\{3,2C_3\}},
    \end{equation}
    and the number of applications\footnote{c.f. Remark \ref{remark: circuit Q^kA} Item 3. for the precise meaning of \textit{number of applications}.} of quantum circuit $\mathcal{A}$ in Algorithm~\ref{quantum algorithm} is at most 
    \begin{equation}\label{eqn: bound on queries of A in algorithm 1}
        \mathfrak{C}_3 d^{4.75}\eps^{-3}\ln(\tfrac{21}{\alpha}).
    \end{equation}

\end{theorem}

\begin{remark}\label{rem:complexity}
Let us remark the following on the complexity of Algorithm~\ref{quantum algorithm}.
\begin{enumerate}
    \item The bounds \eqref{eqn: bound on no. of qubits} and \eqref{eqn: bound on no. of gates} on the number of qubits and the number of elementary gates in Algorithm~\ref{quantum algorithm} specify the requirements on the quantum computer needed to run Algorithm~\ref{quantum algorithm}. The bound \eqref{eqn: bound on queries of A in algorithm 1} on the number of applications of quantum circuit $\mathcal{A}$ can be interpreted as the computational running time for Algorithm~\ref{quantum algorithm}.
    \item The $O(\eps^{-3})$ running time complexity in \eqref{eqn: bound on queries of A in algorithm 1} can be attributed as follows. 
    \begin{enumerate}[(i)]
        \item The truncation of the integral \eqref{eqn: option price formula} from $\R_+^d$ to the cube $[0,M]^d$ requires $M \sim 2^{n_1} \sim O(\eps^{-1})$, see Proposition \ref{prop: truncation error}. 
        \item Since the payoff function $h$ grows linearly,  $\Vert h \Vert_{L^\infty((0,M)^d)}$ grows of order $O(\eps^{-1})$. Hence we require the scaling parameter $\mathfrak{s}$ to satisfy $\mathfrak{s} \sim O\left(\left(\tfrac{\eps}{\Vert h \Vert_{L^\infty((0,M)^d)}^3}\right)^{1/2}\right) \sim O(\eps^{2})$, see Proposition \ref{prop: rotation error}.
        \item We use the Modified IQAE to output $\widehat{a}$ with accuracy $\eps\mathfrak{s} \sim O(\eps^3)$  to obtain the estimate \eqref{eqn: theorem estimate}, see Proposition \ref{proposition: final error estimate}. This implies the query complexity bound \eqref{eqn: bound on queries of A in algorithm 1}.
    \end{enumerate}
    In the case where the payoff function $h:\R^d_+ \to \R$ is bounded uniformly in $\bm{x} \in \R^d_+$, then the scaling parameter~$\mathfrak{s}$ requires only $O(\eps^{\tfrac{1}{2}})$, see (ii) above. This implies that the number of applications of $\mathcal{A}$ can be reduced to $O(\eps^{-\tfrac{3}{2}})$, which is a speed-up compared to  classical Monte Carlo methods and recovers the complexity observed in \cite{QC4_optionpricing}. We highlight that the fact that an unbounded payoff function can lead to a higher complexity has been already briefly outlined in \cite[Equation~(36)]{chakrabarti2021threshold}. Moreover, we highlight that one cannot expect to obtain $O(\eps^{-1})$, which would have meant to have a quadratic speed-up over classical Monte Carlo methods, since one cannot expect to have an oracle which can perfectly upload the distribution and payoff function in rotated form (see, e.g., \eqref{rotated_form} or \cite[Equation (16)]{chakrabarti2021threshold}), as already pointed out, e.g., in \cite{herbert2022quantum}.
    
    \item We note from \eqref{eqn: theorem estimate}--\eqref{eqn: bound on queries of A in algorithm 1} that the number of qubits and elementary gates used in Algorithm \ref{quantum algorithm} as well as the number of applications of the quantum circuit $\mathcal{A}$ grow only polynomially\footnote{under the additional assumption that $\max_{i=1,\ldots,d}\{\vert x_i \vert \}$ is uniformly bounded in $d \in \N$. This assumption is naturally fulfilled in practice, as $x_i^2$ corresponds to the (squared) spot price of the $i$-th asset.} in $d$ and $\eps^{-1}$.  
   This is in line with classical (i.e.\ non-quantum) Monte Carlo methods, where the number of Monte Carlo samples required in order to approximate the solution of the Black-Scholes PDE $u(t,\bm{x})$  only grows polynomially in the dimension $d$ and $\varepsilon^{-1}$, see, e.g., \cite{HutzenthalerJentzenVonWurstemberger}.
    
    Indeed, let $N \in \N$ be the number of classical Monte Carlo samples, let $\bm{Y}^{(n)}, n=1,\ldots,N$ be $N$ independent samples of the stock price process $\bm{S}_T$  given that $\bm{S}_t=\bm{x}$ (c.f.\ Lemma~\ref{lemma: density}), and let $ u_{\mathrm{MC}}^N(t,\bm{x}) := \frac{e^{-r(T-t)}}{N}\sum_{n=1}^N h(\bm{Y}^{(n)})$ be the classical Monte Carlo approximation of $u(t,\bm{x})$. By using Markov inequality, the linear growth property of the payoff function $h$  under Assumption~\ref{assumption: CPWA}  (see Lemma~\ref{lemma: lin growth}), and second moment estimates for the multi-variate lognormal distribution under Assumption~\ref{assumption: cov matrix}  (see \eqref{second-mom-multi-log-normal}), we see  for any prescribed accuracy $\eps \in (0,1)$ that
    \begin{equation}
    	\begin{split}
    		\mathbb{P}\left( |u(t,\bm{x}) - u_{\mathrm{MC}}^N(t,\bm{x})| \geq  \eps \right)
    		&= 	\mathbb{P}\left( \left| \E[h(\bm{S}_T)\,|\,\bm{S}_t = \bm{x}] - \frac{1}{N}\sum_{n=1}^N h(\bm{Y}^{(n)})\right| \geq  \eps e^{r(T-t)} \right)\\
    		& \leq \frac{\mathbb{E}\left[ \left|\E[h(\bm{S}_T)\,|\,\bm{S}_t = \bm{x}] - \frac{1}{N}\sum_{n=1}^N h(\bm{Y}^{(n)}) \right|^2\right]}{\eps^2 e^{2r(T-t)}}\\
    		&\leq \frac{2C_2^4 d^{3} \left( 1 + de^{2C_1^2T^2} e^{2rT}\max_{i=1,\ldots,d}\{1,x_i^2\} \right)}{N\eps^2 e^{2r(T-t)}}.
    	\end{split}
    \end{equation}
    Therefore, for any confidence level $\alpha \in (0,1)$,  by choosing
    \begin{equation}
    	N \geq \frac{2C_2^4 d^{3} \left( 1 + de^{2C_1^2T^2} e^{2rT}\max_{i=1,\ldots,d}\{1,x_i^2\} \right)}{\alpha\eps^2 e^{2r(T-t)}},
    \end{equation}
    we can ensure that 
    \begin{equation}
    	\mathbb{P}\left( |u(t,\bm{x}) - u_{\mathrm{MC}}^N(t,\bm{x})| \geq  \eps \right) \leq \alpha
    \end{equation}
    as required. Since each sample of $\bm{Y}^{(n)}$ requires $d$ copies of standard Gaussian random variables, we conclude that the total computational cost for a classical Monte Carlo method to approximately solve the multi-variate Black-Scholes PDE~\eqref{eqn: PDE} is about $Nd = \mathcal{O}(d^5\eps^{-2})$. 	
    \item The explicit constants \eqref{eqn: constant frakC_1}--\eqref{eqn: constant frakC_3} are not likely to be sharp since we did not optimize every inequality when bounding the number of elementary gates used to construct the quantum circuits in Section \ref{section: loading CPWA circuits}.
\end{enumerate}
\end{remark}

\section{Numerical Simulations}\label{sec: numerics}
In this section, we discuss the implementation of the proposed quantum algorithm, and illustrate its numerical performance on three concrete European options introduced in Example~\ref{example: call options} in dimension two. 
We have  developed  a package we named \texttt{qfinance} using the \texttt{Qiskit} framework to implement our proposed algorithm. The \texttt{OptionPricing} class within this package enables the user to input  the stock parameters 
(i.e.\ volatility $\sigma$, correlation $\rho$, interest rate $r$, current time $t$, spot price $\bm{S}_t$, and maturity $T$), select from any of the six classes of continuous piecewise affine payoff functions $h$ presented in Example~\ref{example: call options}, together with its corresponding parameters (i.e.\ weights and strikes), as well as specify the error tolerance $\varepsilon$ in order to obtain an approximated value for the solution $u(t,\bm{x})$ of the PDE \eqref{eqn: PDE} with payoff function $h$ which corresponds to the fair price of the option $h$ at time $t$ given spot price~$\bm{S}_t=\bm{x}$.

In our presented numerical simulations, we implement Algorithm~\ref{quantum algorithm} to approximate the expected option payoff $u(t,\bm{x})$ at time $t=0$, see \eqref{eqn: option price formula} for the analytic expression.
We used two assets with each asset having initial spot price of $x= S_0 =2.0$, volatility of $\sigma=0.4$, an annual risk-free interest rate of $r=0.04$, and a time to maturity of $T=40$ days,
and where the assets are correlated to each other with a correlation coefficient of $\rho=0.2$. These asset parameters are processed to produce an
	expected future price of $\mu_{S_T}:=\mathbb{E}[S_T] \approx 2.00879$ with a standard deviation $\sigma_{S_T}\approx 0.267168$. 

All numerical experiments presented in this section were implemented using IBM's quantum computing toolkit \texttt{Qiskit} \cite{Qiskit}. We utilized the \texttt{LogNormalDistribution} class from \texttt{Qiskit} to load the multivariate log-normal distribution on the quantum computer. We discretize each asset price $S_T$ using three qubits, 
hence into eight distinct values on equally spaced points in the interval $[\max\{\mu_{S_T}- 3\sigma_{S_T},0\}, \mu_{S_T}+3\sigma_{S_T}]$. Namely, $\ket{000}$ is mapped to $\max\{\mu_{S_T}-3\sigma_{S_T},0\} $ and $\ket{111}$ is mapped to $\mu_{S_T}+3\sigma_{S_T}$ in the natural binary order.

The integration of the distribution loading quantum circuit and the payoff quantum circuit forms the quantum circuit $\mathcal{A}$, which is used as the input for the Quantum Amplitude Estimation (QAE) algorithm. We employ the Modified Iterative Quantum Amplitude Estimation (Modified IQAE) in \cite[Algorithm 1]{fukuzawa2022modified}, with the parameters $\alpha = 0.005$ and $\epsilon = 0.001$; see Proposition~\ref{prop: m-IQAE}. For the numerical results, \textit{Estimated (mid)} corresponds to  $\widetilde{U}_{t,\bm{x}}$ (see line 8 of Algorithm~\ref{quantum algorithm} for the precise rescaling of $\widehat{a}$) whereas \textit{Estimated (high)} and \textit{Estimated (low)} correspond to the same scaling as line~8 of Algorithm~\ref{quantum algorithm} but with respect to $a_u$ and $a_\ell$, respectively.    All numerical experiments in this section were implemented in \texttt{Python} using \texttt{Qiskit} on an Ubuntu 22.04 machine with a AMD Ryzen 9 5950X @ 3.875GHz CPU with 64GB of RAM and Nvidia GeForce RTX3090 GPU. The source codes are available at \url{https://github.com/jianjun-dot/quantum-finance}.

\subsection{Basket Call Option}\label{sec: basket call example}
We consider  the basket call option with two assets 
and strike $\kappa$
where the weight of each asset is set to one, i.e.,
\begin{align}
	h(x_1,x_2) = \max\{x_1+x_2-\kappa, 0\}.
\end{align}
The results of the algorithm across a range of strike prices are shown in Figure~\ref{fig:basket_call_estimates}. We see that the estimated expected payoffs align closely with the reference values, which are computed based on the probabilities in \texttt{LogNormalDistribution}.
\begin{figure}[h!]
	\centering
	\includegraphics[width=0.5\linewidth]{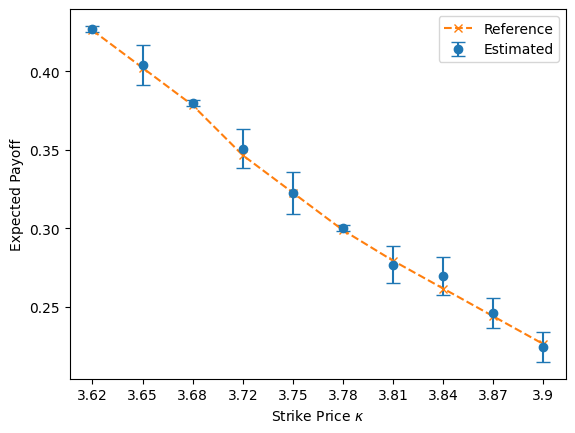}
	\caption{Expected payoff estimates from our proposed algorithm for the basket call option across a range of strike prices, compared with the 
		 reference 
		expected payoff. The tested strike prices are labeled on the horizontal axis.}
	\label{fig:basket_call_estimates}
\end{figure}
\begin{table}[h!]
	\centering
	\begin{tabular}{|c|c|c|c|c|}
		\hline
		Strike Price $\kappa $ & Reference & Estimated (mid) & Estimated (high) & Estimated (low)\\
		\hline
		3.62&0.426180&0.427173&0.428978&0.425367\\
		3.65&0.402294&0.404174&0.416717&0.391631\\
		3.68&0.378408&0.379875&0.381912&0.377837\\
		3.72&0.346561&0.350753&0.363412&0.338095\\
		3.75&0.322675&0.322466&0.335617&0.309314\\
		3.78&0.298789&0.300295&0.302123&0.298467\\
		3.81&0.279313&0.276816&0.288819&0.264814\\
		3.84&0.261602&0.269334&0.281401&0.257267\\
		3.87&0.243890&0.246027&0.255811&0.236242\\
		3.90&0.226178&0.224180&0.233560&0.214800\\
		\hline
	\end{tabular}
	\caption{Numerical results for basket call options}
	\label{tab:basket_call_options_table}
\end{table}
\FloatBarrier

\subsection{Call-on-min Option}\label{sec: callonmin example}
We consider the two-asset call-on-min option with the following payoff function 
\begin{align}
	h(x_1,x_2) = \max\{\min\{x_1, x_2\} - \kappa,0\}.
\end{align}
 The performance of our proposed algorithm is evaluated across a range of strike prices, with the results shown in Figure \ref{fig:call_on_min_estimate}. 
\begin{figure}[h!]
	\centering
	\includegraphics[width=0.5\linewidth]{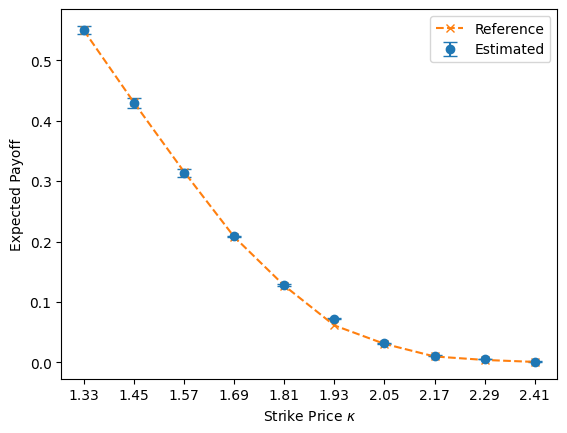}
	\caption{Expected payoff estimates from the algorithm for the call-on-min option across a range of strike prices, compared with the 
		reference 
		expected payoff. The tested strike prices are  labeled on the horizontal axis.}
	\label{fig:call_on_min_estimate}
\end{figure}
\begin{table}[h!]
	\centering
	\begin{tabular}{|c|c|c|c|c|}
		\hline
		Strike Price $\kappa $ & Reference & Estimated (mid) & Estimated (high) & Estimated (low)\\
		\hline
		1.33&0.549242&0.549779&0.556598&0.542961\\
		1.45&0.429959&0.429376&0.437922&0.420830\\
		1.57&0.315390&0.313495&0.319649&0.307340\\
		1.69&0.207789&0.208521&0.209452&0.207590\\
		1.81&0.127052&0.128095&0.129009&0.127182\\
		1.93&0.061264&0.072553&0.073112&0.071993\\
		2.05&0.030753&0.031998&0.032863&0.031132\\
		2.17&0.009879&0.011256&0.012077&0.010435\\
		2.29&0.004124&0.005356&0.005689&0.005022\\
		2.41&0.000863&0.001378&0.001902&0.000854\\
		\hline
	\end{tabular}
	\caption{Numerical results for call-on-min options}
	\label{tab:call_on_min_options_table}
\end{table}
\FloatBarrier
\subsection{Best-of-call Option}\label{sec: bestofcall example}
The best-of-call option is the most complex option among the examples discussed, as it involves multiple assets with multiple strike prices. 
Here, we consider a two assets, two strike prices best-of-call option, with the following payoff function
\begin{align}
	h(x_1,x_2) = \max\{x_1-\kappa_1, x_2-\kappa_2,0\}.
\end{align}
We tested our proposed algorithm across a range of strike prices, with the results shown in Figure \ref{fig:best_of_call_estimates}. 

\begin{figure}[h!]
	\centering
	\includegraphics[width=0.5\linewidth]{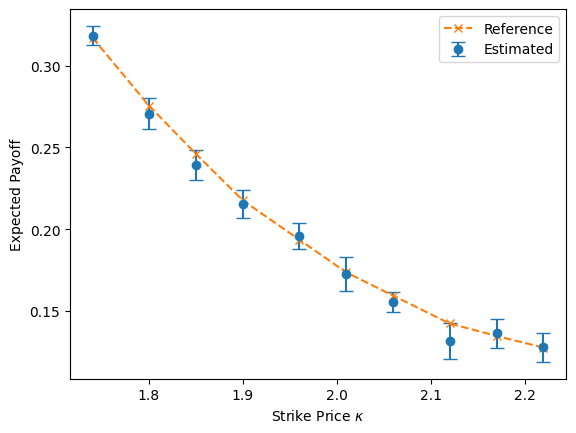}
	\caption{Expected payoff estimates from the algorithm for the best-of-call option across a range of strike prices, compared with the reference expected payoff. The strike price of the first asset is fixed, while the strike price of the second assets is varied. }
	\label{fig:best_of_call_estimates}
\end{figure}
\begin{table}[h!]
	\centering
	\begin{tabular}{|c|c|c|c|c|}
		\hline
		\multicolumn{5}{|c|}{Second strike price fixed at $\kappa_2=2.01$}\\
		\hline
		Strike Price $\kappa_1 $ & Reference & Estimated (mid) & Estimated (high) & Estimated (low)\\
		\hline
		1.74&0.317163&0.318514&0.324408&0.312619\\
		1.80&0.275635&0.270604&0.280077&0.261131\\
		1.85&0.246069&0.239256&0.248482&0.230031\\
		1.90&0.217595&0.215355&0.223858&0.206853\\
		1.96&0.193578&0.195811&0.203952&0.187671\\
		2.01&0.173564&0.172656&0.183125&0.162187\\
		2.06&0.159330&0.155583&0.161787&0.149379\\
		2.12&0.142249&0.131429&0.142335&0.120523\\
		2.17&0.134635&0.136207&0.145155&0.127260\\
		2.22&0.127487&0.127613&0.136500&0.118725\\
		\hline
	\end{tabular}
	\caption{Numerical results for best-of-call options}
	\label{tab:best_of_call_options_table}
\end{table}
\subsection{Discussion on the numerical simulations and their possible extension to higher dimensions}
We see from the numerical results that the algorithm meets our performance expectations. Compared to the results in \cite{QC4_optionpricing,QC6_risk_analysis}, our work extends the application to cover more complex call options, including 
multi-asset call options such as 
call-on-min and best-of-call options. New circuits are constructed to handle these options, incorporating components such as subtraction subroutines and comparison subroutines.

Currently, our implementation is limited to options involving two variables. Nevertheless, the framework can be easily extended to multiple variables through the integration of multiple $\mathcal{Q}_{(\mathrm{comp})}$ subroutines (later introduced in Lemma~\ref{lemma: quantum comp}) for variable comparisons. To discuss the involving steps more in detail, let us consider, e.g., the extension of the best-of-call option from two variables to three variables. The payoff function of the three variable case is defined as:
\begin{align}
	h(x_1,x_2,x_3)=\max\{x_1-\kappa_1, x_2-\kappa_2, x_3 - \kappa_3, 0\}
\end{align}
which can be decomposed into multiple nested two-variable comparisons:
\begin{align}
	h(x_1,x_2,x_3) = \max\{\max\{\max\{x_1-\kappa_1, x_2-\kappa_2\},  x_3-\kappa_3\}, 0\}.
\end{align}
This approach systematically extends the circuit from the two-variable case. For each variable, an ancilla register is added to load the corresponding strike price. Three different subtraction subroutines are applied on the corresponding variable and its ancilla, resulting in three different registers with that stores $\ket{x_1-\kappa_1}$, $\ket{x_2-\kappa_2}$, $\ket{x_3-\kappa_3}$. To determine the maximum of three registers, we can use two comparison operators. The first $\mathcal{Q}_{(\mathrm{comp})}$ operation compares registers $\ket{x_1-\kappa_1}$ and $\ket{x_2-\kappa_2}$, storing the comparison results in the ancilla register $\ket{c_1}$. By controlling on the results of the ancilla register $\ket{c_1}$, a second $\mathcal{Q}_{(\mathrm{comp})}(\max\{x_1-\kappa_1,x_2-\kappa_2\}, x_3-\kappa_3)$ operation then compares the result $\max\{x_1-\kappa_1,x_2-\kappa_2\}$ from the first comparison with $x_3-\kappa_3$, and stores the result in another ancilla register $\ket{c_2}$. Using the outcomes of the two comparison operators $\mathcal{Q}_{(\mathrm{comp})}(x_1-\kappa_1,x_2-\kappa_2)$ and $\mathcal{Q}_{(\mathrm{comp})}(\max\{x_1-\kappa_1,x_2-\kappa_2\}, x_3-\kappa_3)$, the largest value can be identified as outlined Table \ref{tab:ancilla_results_two_comparison}. Accordingly, by controlling the circuit based on the comparison results, the appropriate register can be selected to compute the expected payoff. By iteratively applying this nested $\mathcal{Q}_{(\mathrm{comp})}$ strategy, the algorithm can be generalized to calculate the payoff for any $d$ variables.

 For general continuous piecewise affine payoff functions, we expect a quadratic increase (ignoring logarithmic factors) with respect to the dimension $d$ in the number of qubits resource requirement due to the following two reasons -- the number of comparison operations required is $O(d)$ under Assumption~\ref{assumption: CPWA} and the computation for each affine sum (i.e. $\bm{x} \mapsto \bm{a}_{k,l}\cdot \bm{x}+b_{k,l}$ in \eqref{eqn: CPWA payoff}) requires $O(d)$ arithmetic operations and storage. We refer to equation~\eqref{eqn: bound on N} for the amount $N$ of qubits needed to upload the (discretized version of the) continuous piecewise affine payoff function, which under Assumption~\ref{assumption: CPWA} satisfies $N=O(d^2)$. This and line 6 of Algorithm~\ref{quantum algorithm} hence imply that the amount of qubits needed for Algorithm~\ref{quantum algorithm} scales quadratically in the dimension $d$ of the PDE. 

We highlight that the presented steps involving nested $\mathcal{Q}_{(\mathrm{comp})}$ can be similarly applied to other options presented in Example~\ref{example: call options}.

However, current quantum computing hardware are still too nascent to handle complex computations involving multiple variables. 
Thus, given the computational constraints, only two variables systems are implemented and tested in the package. As compute power increases, we can use the method described to extend the circuit to accommodate more variables.

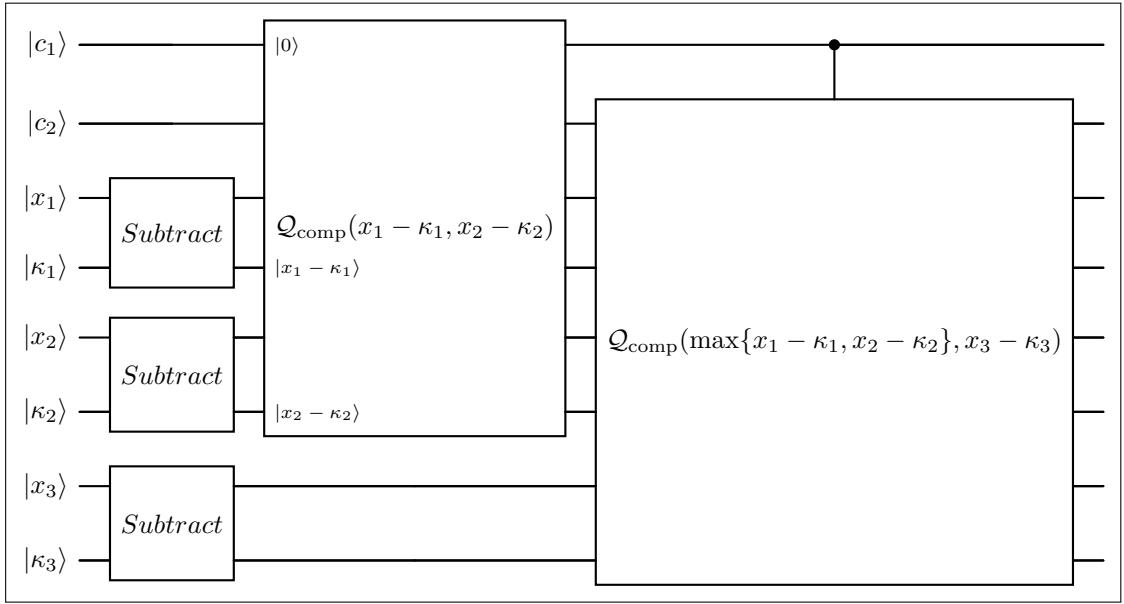
\begin{figure}[h]
	\centering
	\boxed{
		\begin{quantikz}[transparent, column sep = 0.4cm, row sep = 0.4cm]
			\lstick{$\ket{c_1}$}\qw&\qw& \gate[wires=6]{\mathcal{Q}_{(\mathrm{comp})}(x_1-\kappa_1,x_2-\kappa_2)}\gateinput{$\ket{0}$}& \ctrl{1} &\qw\\
			\lstick{$\ket{c_2}$}\qw&\qw& &\gate[wires=7]{\mathcal{Q}_{(\mathrm{comp})}(\max\{x_1-\kappa_1,x_2-\kappa_2\}, x_3-\kappa_3)}&\qw\\
			\lstick{$\ket{x_1}$} \qw& \gate[wires=2]{Subtract} & \qw&\qw&\qw\\
			\lstick{$\ket{\kappa_1}$} \qw& \qw& \qw\gateinput{$\ket{x_1-\kappa_1}$}&\qw&\qw\\
			\lstick{$\ket{x_2}$} \qw&\gate[wires=2]{Subtract} & \qw&\qw&\qw\\
			\lstick{$\ket{\kappa_2}$} \qw& \qw& \qw\gateinput{$\ket{x_2-\kappa_2}$}&\qw&\qw\\
			\lstick{$\ket{x_3}$} \qw&\gate[wires=2]{Subtract} & \qw&\qw&\qw \\
			\lstick{$\ket{\kappa_3}$} \qw& \qw& \qw&\qw&\qw
		\end{quantikz}
	}
	\caption{Description of the circuit to compare three variables. The first $\mathcal{Q}_{(\mathrm{comp})}(x_1-\kappa_1,x_2-\kappa_2)$ compares the subtraction result stored in the register $\ket{\kappa_1}$ and $\ket{\kappa_2}$, and the result of the comparison is store in register $\ket{c_1}$. The second comparison $\mathcal{Q}_{(\mathrm{comp})}(\max\{x_1-\kappa_1,x_2-\kappa_2\}, x_3-\kappa_3)$ is then controlled on the results from the first comparison, thereby selecting the correct register for the second comparison. The result of the second comparison is then stored in $\ket{c_2}$.}
\end{figure}

\begin{table}[h]
	\centering
	\begin{tabular}{|c|c|c|} 
		\hline
		Measurement results for $\ket{c_1}_3$ & Measurement results for $\ket{c_2}_3$ & Maximal value \\
		\hline
		100 & 100 & $x_3-\kappa_3$ \\ 
		100 & 010 & $x_2-\kappa_2$ \\ 
		100 & 001 & $x_2-\kappa_2$ \\
		010 & 100 & $x_3-\kappa_3$ \\
		010 & 010 & $x_1-\kappa_1$ \\
		010 & 001 & $x_1-\kappa_1$ \\
		001 & 100 & $x_3-\kappa_3$ \\
		001 & 010 & $x_1-\kappa_1$ \\
		001 & 001 & $x_1-\kappa_1$ \\
		\hline
	\end{tabular}
	\caption{The measurement results for the corresponding $\mathcal{Q}_{(\mathrm{comp})}$ operations, and the identified maximal value among the three variables. The ancilla register $\ket{c_1}$ holds the result for the first comparison $\mathcal{Q}_{(\mathrm{comp})}(x_1-\kappa_1,x_2-\kappa_2)$ between $x_1-\kappa_1$ and $x_2-\kappa_2$, whereas $\ket{c_2}$ holds the result of the second comparison $\mathcal{Q}_{(\mathrm{comp})}(\max\{x_1-\kappa_1,x_2-\kappa_2\}, x_3-\kappa_3)$. For cases where the two variables are equal, the first variable is selected as the maximal value. The ancilla results can locate the maximal value for computing the payoff in the subsequent parts of the circuit.}
	\label{tab:ancilla_results_two_comparison}
\end{table}
\FloatBarrier

\section{Quantum Circuits}\label{sec:quantum_circuits}
In Section \ref{section: two complement's}, we introduce the so-called \textit{two's complement method} for representing signed dyadic rational numbers on a bounded interval using binary strings of finite length. These binary strings correspond to the grid points on the truncated interval. The binary strings are represented by the qubits on the quantum computer, in the form of linear combinations of the computational basis states $\ket{i}_n = \ket{i_1}\ket{i_2}\ldots\ket{i_n}$, where $i = (i_1,\ldots,i_n) \in \{0,1\}^n$. In Section \ref{section: quantum arithmetic circuits}, we describe how quantum circuits are constructed to perform arithmetic operations on two complement's numbers. In Section \ref{section: distribution quantum circuit}, we assume that the discretized multivariate log-normal distribution can be loaded on the quantum computer with complexities comparable to the estimates in \cite{chakrabarti2021threshold}. In Section \ref{section: loading CPWA circuits},  the approximate option payoff function is loaded, by using quantum circuits that perform arithmetic operations on numbers represented by the two's complement method. 

\subsection{Representing signed dyadic rationals using the two's complement method}\label{section: two complement's}
The two's complement method is a way of representing signed integers on a computer using binary strings, see e.g. Chapter 2.2 \cite{bryant2003computer} for an introduction to the subject. We first describe the representation of signed integers using the two's complement method. For a given $n \in \N$ and for an integer $x \in [-2^{n-1},2^{n-1}-1] \cap \Z$, we encode $x$ in the two's complement method by a $n$-bit string denoted by $(x_{n-1},x_{n-2}, \cdots, x_0) \in \{0,1\}^n$. The value of the integer $x$ is converted from the bit string $(x_{n-1},\cdots, x_0)$ by the following formula
\begin{equation}
    x = - x_{n-1} 2^{n-1} + \sum_{k=0}^{n-2} x_k 2^k. \label{eqn: integer conversion}
\end{equation}
The most significant bit (MSB) is the bit $x_{n-1} \in \{0,1\}$ which determines the sign of $x$. There are classical computer algorithms for performing arithmetic operations (such as addition and multiplication) in the two's complement representation, such as the \textit{carry adder algorithm} and \textit{Booth's multiplication algorithm} \cite{bryant2003computer}. Numbers with a fractional part can also be represented using the two's complement method. This can be done by introducing the \textit{radix point} (commonly referred as the decimal point in decimal expansion) to separate the integer part and fractional part. The additional bits are mapped to the dyadics $2^{-m}, m \in \N$ to represent the fractional part of a number. 

\begin{definition}[Two's complement representation]\label{def: two complement}
Let $n\in \N$, $m \in \N_0$. For $m \geq 1$, we define the following set of $(n,m)$-bit strings by 
\begin{equation}
    \mathbb{F}_{n,m} := \{0,1\}^n \times \{0,1\}^m := \{((x_{n-1},x_{n-2},\ldots,x_0),(x_{-1},\ldots,x_{-m})) \in \{0,1\}^n \times \{0,1\}^m\},
\end{equation}
and the set of dyadic rational numbers on a closed interval by
\begin{equation}
    \mathbb{K}_{n,m} := [-2^{n-1},2^{n-1} - 2^{-m}]\cap 2^{-m}\Z := \{-2^{n-1},-2^{n-1}+2^{-m},-2^{n-1}+2\cdot 2^{-m},\ldots,2^{n-1}-2^{-m}\}.
\end{equation}
Further, we denote 
\begin{equation}
	\mathbb{K}_{n,m,+} := \mathbb{K}_{n,m} \cap [0,\infty), \quad \text{and} \quad \mathbb{K}_{n,m,-} := \mathbb{K}_{n,m} \cap (-\infty,0),
\end{equation}
and we denote 
\begin{equation}
	\begin{split}
		&\mathbb{F}_{n,m,+} := \{((0,x_{n-2},\ldots, x_0),(x_{-1},\ldots,x_{-m})): x_{n-2},\ldots,x_{-m} \in \{0,1\}\}, \quad \text{and}\\
		&\mathbb{F}_{n,m,-} := \{((1,x_{n-2},\ldots, x_0),(x_{-1},\ldots,x_{-m})): x_{n-2},\ldots,x_{-m} \in \{0,1\}\}.
	\end{split}
\end{equation}
If $m=0$, we then use the usual signed integers $\mathbb{K}_{n,0} := [-2^{n-1},2^{n-1}-1] \cap \Z$ and the set of $n$-bit strings $\mathbb{F}_{n,0} := \{(x_{n-1},\cdots,x_0) \in \{0,1\}^n\}$, and define $\mathbb{K}_{n,0,\pm}$ and $\mathbb{F}_{n,0,\pm}$ analogously. 
\end{definition}
\begin{definition}[Encoder and decoder maps]\label{def: encoder-decoder}
Let $n \in \N$, $m \in \N_0$. We define the encoder function which maps the rational numbers to bit strings by 
\begin{equation}
\begin{split}
    \mathrm{E}_{n,m}:\ &\mathbb{K}_{n,m} \longrightarrow \mathbb{F}_{n,m}\\
    &y \mapsto ((x_{n-1},x_{n-2},\ldots,x_0),(x_{-1},\ldots,x_{-m}))
\end{split}
\end{equation}
where we define $\mathrm{E}_{n,m}(y) = ((x_{n-1},x_{n-2},\ldots,x_0),(x_{-1},\ldots,x_{-m}))$ recursively by
\begin{equation}
\begin{split}
    &x_{n-1} = \begin{cases} 1,\quad \text{if $y < 0$,} \\ 0,\quad \text{if $y \geq 0$,}\end{cases}\\
    &x_{n-2} = \begin{cases}1,\quad \text{if $-x_{n-1}2^{n-1} + 2^{n-2} \leq y$},\\ 0, \quad \text{if $-x_{n-1}2^{n-1} + 2^{n-2} > y$}, \end{cases}\\
\end{split}
\end{equation}    
and for $k=n-3,n-4,\ldots,-m$,
\begin{equation}
    x_k = \begin{cases}1,\quad \text{if $-x_{n-1}2^{n-1} + \displaystyle\sum_{j=k+1}^{n-2} x_j 2^j + 2^k \leq y$},\\
        0,\quad \text{if $-x_{n-1}2^{n-1} + \displaystyle\sum_{j=k+1}^{n-2} x_j 2^j + 2^k > y$}.\end{cases} 
\end{equation}
We define the decoder function which maps the bit strings to the rational numbers by
\begin{equation}
\begin{split}
    \mathrm{D}_{n,m}:\ &\mathbb{F}_{n,m} \longrightarrow \mathbb{K}_{n,m} \\
    &((x_{n-1},x_{n-2},\ldots,x_0),(x_{-1},\ldots,x_{-m})) \mapsto -x_{n-1} 2^{n-1} + \sum_{k=-m}^{n-2} x_k 2^k.
\end{split}
\end{equation}
\end{definition}
The sets $\mathbb{F}_{n,m}$ and $\mathbb{K}_{n,m}$ are equivalent in the following sense.
\begin{proposition}[Bijection between $(n,m)$-bit strings and dyadics]\label{prop: isomorphism for TC}
Let $n \in \N$, $m \in \N_0$. The sets $\mathbb{F}_{n,m}$ and $\mathbb{K}_{n,m}$ (c.f. Definition \ref{def: two complement}) have the same finite cardinality, and the encoder and decoder functions $\mathrm{E}_{n,m}:\mathbb{K}_{n,m} \longrightarrow \mathbb{F}_{n,m}$ and $\mathrm{D}_{n,m}:\mathbb{F}_{n,m} \longrightarrow \mathbb{K}_{n,m}$ (c.f. Definition \ref{def: encoder-decoder}) are bijective and inverses of the other. 

\begin{proof}
For the first part of the statement, by observing that
\begin{equation}
    \mathbb{K}_{n,m} = \{j \cdot 2^{-m}: j=0,1,\ldots,2^{n+m-1}-1 \} \cup \{-j\cdot 2^{-m}: j = 1,\ldots, 2^{n+m-1}\},
\end{equation}
it follows that
\begin{align}
    \#\mathbb{K}_{n,m} =  2 \cdot 2^{n+m-1} = 2^n \cdot 2^{m} =\#(\{0,1\}^n \times \{0,1\}^m) = \#\mathbb{F}_{n,m},
\end{align}
where we denote by $\#A$ the cardinality of a set $A$. This shows that the two sets $\mathbb{F}_{n,m}$ and $\mathbb{K}_{n,m}$ have the same cardinality. Injectivity is clear from Definition \ref{def: encoder-decoder}, and bijectivity follows from injectivity since both sets have same finite cardinality.
\end{proof}

\end{proposition}
In the later sections, we will use the notation $\mathbb{F}_{n,m}^k := \underbrace{\mathbb{F}_{n,m}\times\cdots\times \mathbb{F}_{n,m}}_{\text{$k$-times}} $ for any $k \in \N$. 
The arithmetic algorithms in two's complement (TC) representation for signed rational numbers can be extended from the arithmetic algorithms on the two's complement signed integers. The modifying process is done by shifting the fractional bits to the integer bits, applying the integer arithmetic algorithms, then shifting the integer bits back to fractional bits. The proofs in the two following lemmas provide the extension procedure.

\begin{lemma}[Addition in two's complement]\label{lemma: addition in TC}
Let $n_1,n_2 \in \N$, and let $n := \max\{n_1,n_2\}$. Let $\boxplus:\mathbb{F}_{n_1,0} \times \mathbb{F}_{n_2,0} \to \mathbb{F}_{n+1,0}$ be the addition algorithm for integers represented in the two's complement method. Then, for any $m_1,m_2 \in \N_0$ with $m:= \max\{m_1,m_2\}$, there is a natural extension of the addition algorithm to the rational numbers represented in the two's complement method where ${\boxplus}:\mathbb{F}_{n_1,m_1} \times \mathbb{F}_{n_2,m_2} \to \mathbb{F}_{n+1,m}$, such that for any $x\in \mathbb{F}_{n_1,m_1}$, $y\in \mathbb{F}_{n_2,m_2}$, there is an unique element $x \ {\boxplus}\ y \in \mathbb{F}_{n+1,m}$ that satisfies 
\begin{equation}\label{boxplus}
    x\ {\boxplus}\ y = \mathrm{E}_{n+1,m}(\mathrm{D}_{n_1,m_1}(x) + \mathrm{D}_{n_2,m_2}(y)).
\end{equation}
\begin{proof}
First, consider the case where both $m_1$ and $m_2$ are positive. Let $x = ((x_{n_1 -1},\ldots,x_0),(x_{-1},\ldots,x_{-m_1})) \in \mathbb{F}_{n_1,m_1}$ and $y = ((y_{n_2 -1},\ldots,y_0),(y_{-1},\ldots,y_{-m_2})) \in \mathbb{F}_{n_2,m_2}$ be given. For every $p,q,r \in \N$ with $r\geq q$, define a left-shift operator $\tau_r: \mathbb{F}_{p,q} \to \mathbb{F}_{p+r,0}$ defined by
\begin{equation}\label{eqn: shift operator}
    \tau_r: ((z_{p-1},\ldots,z_0),(z_{-1},\ldots,z_{-q})) \mapsto (z_{p-1},\ldots,z_0,z_{-1},\ldots,z_{-q},\underbrace{0,\ldots,0}_{\text{$(r-q)$-times}}) .
\end{equation}
Then, it holds that
\begin{equation}
    \tau_m(x) = (x_{n_1-1},\ldots,x_0,x_{-1},\ldots,x_{-m_1},\ldots,x_{-m}) \in \mathbb{F}_{n_1+m,0},
\end{equation}
\begin{equation}
    \tau_m(y) = (y_{n_2-1},\ldots,y_0,y_{-1},\ldots,y_{-m_2},\ldots,y_{-m}) \in \mathbb{F}_{n_2+m,0},
\end{equation}
where $x_{-k} = 0$ for $k=m_1 +  1,\ldots,m$ if $m_1 < m$ and $y_{-l}= 0$ for $l=m_2 + 1,\ldots,m$ if $m_2 < m$. Hence, we may apply the integer addition algorithm and get an output $\tau_m(x) \boxplus \tau_m(y) \in \mathbb{F}_{n+1+m,0}$. Then, for every $p,r \in \N$ with $r \leq p$, we define a right-shift operator $\tau_{-r}: \mathbb{F}_{p,0} \to \mathbb{F}_{p-r,r}$ defined by
\begin{equation}
    \tau_{-r}:(z_{p-1},\ldots,z_0) \mapsto \left((z_{p-1},\ldots,z_r),(z_{r-1},\ldots,z_0)\right).
\end{equation}
Let $\tau_m(x) \boxplus \tau_m(y) = z = (z_{n+m},z_{n+m-1},\ldots,z_0)$, for some bit string $z \in \mathbb{F}_{n+m+1,0}$. Then, we have
\begin{equation}
    \tau_{-m}(\tau_m(x) \boxplus \tau_m(y)) = ((z_{n+m},z_{n+m-1},\ldots,z_m),(z_{m-1},\ldots,z_{0})) \in \mathbb{F}_{n+1,m}.
\end{equation}
Furthermore, it holds that $\mathrm{D}_{n_1+m,0}(\tau_m(x)) = 2^m\mathrm{D}_{n_1,m_1}(x)$ and $\mathrm{D}_{n_2+m,0}(\tau_m(x)) = 2^m\mathrm{D}_{n_2,m_2}(y)$. This implies that
\begin{equation}
    x\ {\boxplus}\ y  := \tau_{-m}(\tau_m(x) \boxplus \tau_m(y)) = \tau_{-m} \mathrm{E}_{n+m+1,0}(2^m(\mathrm{D}_{n_1,m_1}(x)+ \mathrm{D}_{n_2,m_2}(y))) = \mathrm{E}_{n+1,m}(\mathrm{D}_{n_1,m_1}(x) + \mathrm{D}_{n_2,m_2}(y)).
\end{equation}
Hence we have shown \eqref{boxplus}.
The cases where $m_1$ and/or $m_2$ equals to zero follow analogously.

\end{proof}
\end{lemma}

\begin{lemma}[Multiplication in two's complement]\label{lemma: multiplication in TC}
Let $n_1,n_2 \in \N$, and let $n := n_1+n_2$. Let ${\boxdot}:\mathbb{F}_{n_1,0} \times \mathbb{F}_{n_2,0} \to \mathbb{F}_{n,0}$ be the multiplication algorithm for integers represented in the two's complement method. Then, for any $m_1,m_2 \in \N_0$ with $m:=m_1+m_2$, there is a natural extension of the multiplication algorithm to the rational numbers represented in the two's complement method where ${\boxdot}:\mathbb{F}_{n_1,m_1} \times \mathbb{F}_{n_2,m_2} \to \mathbb{F}_{n,m}$ such that for any $x\in \mathbb{F}_{n_1,m_1}$, $y\in \mathbb{F}_{n_2,m_2}$, there is an unique element $x \ {\boxdot}\  y \in \mathbb{F}_{n,m}$, where
\begin{equation}\label{boxdot}
    x\ {\boxdot}\ y = \mathrm{E}_{n,m}(\mathrm{D}_{n_1,m_1}(x) \cdot \mathrm{D}_{n_2,m_2}(y)).
\end{equation}
\begin{proof}
	First, consider the case where both $m_1$ and $m_2$ are positive. 
Let $x = ((x_{n_1},\ldots,x_0),(x_{-1},\ldots,x_{-m_1})) \in \mathbb{F}_{n_1,m_1}$ and $y = ((y_{n_2},\ldots,y_0),(y_{-1},\ldots,y_{-m_2})) \in \mathbb{F}_{n_2,m_2}$ be given. Then, with the left-shift operator $\tau_m$ defined in \eqref{eqn: shift operator} in the proof of the previous lemma, it holds that $\tau_m(x) \in \mathbb{F}_{n_1+m,0}$ and $\tau_m(y) \in \mathbb{F}_{n_2+m,0}$. Hence, applying the multiplication algorithm on integers we have $\tau_m(x) \boxdot \tau_m(y) \in \mathbb{F}_{n+m,0}$. This implies that $\tau_{-m}(\tau_m(x) \boxdot \tau_m(y)) \in \mathbb{F}_{n,m}$. We verify that
\begin{equation}
    x\ {\boxdot}\ y := \tau_{-m}(\tau_m(x) \boxdot \tau_m(y))  = \tau_{-m}\left(\mathrm{E}_{n+m,0}(2^m \mathrm{D}_{n_1,m_1}(x) \cdot \mathrm{D}_{n_2,m_2}(y))\right) = \mathrm{E}_{n,m}(\mathrm{D}_{n_1,m_1}(x) \cdot \mathrm{D}_{n_2,m_2}(y)).
\end{equation}
Hence we have shown \eqref{boxdot}.
The cases where $m_1$ and/or $m_2$ equals to zero follow analogously.
\end{proof}
\end{lemma}

\subsection{Quantum circuits for elementary arithmetic operations}\label{section: quantum arithmetic circuits}
We now describe quantum circuits for arithmetic and elementary operations (such as addition, multiplication, comparison, absolute value), on qubit registers representing numbers in two's complement method. There are many quantum circuits for performing arithmetic operations with its quantum circuit complexities available in the literature, see, e.g., \cite{QC14_arithmetic4, QC10_adder,  QC15_arithmetic5, QC13_arithmetic3, QC12_arithmetic2, QC11_arithmetic}. We first introduce in Lemma \ref{lemma: permutation circuit} an important quantum circuit to perform permutations, which will be necessary for arithmetic computations in the later parts. We have included quantum circuit diagrams in this section to visualize the construction of the quantum circuits for their 
	better understanding. 
	For Lemma~\ref{lemma: quantum addition}, Lemma~\ref{lemma: quantum mult1}, and Lemma~\ref{lemma: quantum comp}, we refer the reader to \cite{QC12_arithmetic2} for its quantum circuit diagrams.

\begin{definition}[Cycle](\cite[Section 1.3, pg 29]{dummitalgebra})\label{def: cycle}
    Let $n,m \in \N$ satisfy $2 \leq m \leq n$, and let $\{a_1,\ldots,a_m\} \subset \{1,2,\ldots,n\}$ be distinct numbers. A cycle $C := (a_1 a_2 \cdots a_m)$ is a permutation $ \sigma:\{1,2,\ldots,n\} \to \{1,2,\ldots,n\}$ such that
    \begin{equation}
        \sigma(j) = \begin{cases}
            a_{i+1}, &\quad \text{if $j = a_i$ for $1 \leq i \leq m-1$}, \\
            a_{1}, &\quad \text{if $j = a_m$}, \\
            j, &\quad \text{if $j \not \in \{a_1,\ldots,a_m\}$}.
        \end{cases}
    \end{equation}
\end{definition}

\begin{proposition}[Cycle decomposition theorem](\cite[Section 4.1, pg 115]{dummitalgebra})\label{prop: cycle decomp}
    Let $n \in \N$ and $\pi:\{1,2,\ldots,n\} \to \{1,2,\ldots,n\}$ be a permutation. Then $\pi$ can be written as a composition of disjoint cycles\footnote{We adopt the convention that cycles of length 1 will not be written.}
            \begin{equation}
        \pi = C_1 C_2 \cdots C_k = (a_1 a_2 \cdots a_{m_1})(a_{m_1 + 1} a_{m_1 + 2} \cdots a_{m_2}) \cdots (a_{m_{k - 1}+1} a_{m_{k-1}+2} \cdots a_{m_{k}}),
    \end{equation}
    where $k$ is the number of cycles, and $\{a_j: j=1,\ldots,m_k\} \subset \{1,\ldots,n\}$ are distinct integers. Moreover, the  cycle decomposition above is unique up to a rearrangement of the cycles and up to a cyclic permutation of the integers within each cycle.     
\end{proposition}

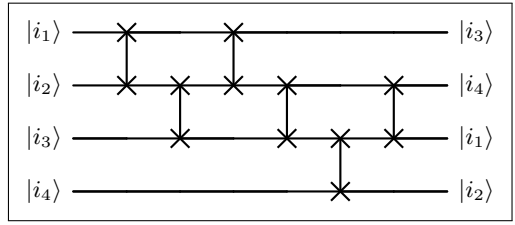
\begin{figure}
	\centering 
\boxed{
	\begin{quantikz}[row sep={20pt,between origins},column sep=20pt,font=\small]
		\lstick{$\ket{i_1}$} & \swap{1} &\qw & \swap{1} &\qw &\qw &\qw &\qw \rstick{$\ket{i_3}$}\\		
		\lstick{$\ket{i_2}$} 	& \targX{} & \swap{1} & \targX{} & \swap{1} &\qw  &\swap{1} &\qw \rstick{$\ket{i_4}$}\\
		\lstick{$\ket{i_3}$} &\qw   &\targX{} &\qw & \targX{} &\swap{1} & \targX{} &\qw \rstick{$\ket{i_1}$}\\
		\lstick{$\ket{i_4}$} 	&\qw  &\qw &\qw &\qw & \targX{} &\qw  &\qw \rstick{$\ket{i_2}$} 
	\end{quantikz}		
}
\caption{An example for the quantum circuit for permutation in Lemma~\ref{lemma: permutation circuit}, with $\pi=(13)(24)$.}
\label{fig: permutation circuit example}
\end{figure}

\begin{lemma}[Quantum circuit for permutation]\label{lemma: permutation circuit}
    \begin{enumerate}
        \item Let $n \in \N$, and let $\pi: \{1,\ldots,n\} \to \{1,\ldots,n\}$ be a permutation. Then, there is a quantum circuit $\mathcal{T}_{\pi}\in U(2^n)$ on $n$ qubits such that for every $\ket{i}_n = \ket{i_1}\cdots\ket{i_n} \in \{0,1\}^n$, it holds that 
        \begin{equation}\label{eqn: permutation circuit}
            \mathcal{T}_{\pi}\ket{i}_n = \mathcal{T}_{\pi}\ket{i_1}\cdots\ket{i_n} =  \ket{i_{\pi(1)}}\cdots\ket{i_{\pi(n)}},
        \end{equation}
        and that $\mathcal{T}_\pi$ uses at most $2n^2$ swap gates.
        
        See Figure~\ref{fig: permutation circuit example} for an example.
        \item Let $\mathcal{Q}\in  U(2^n)$ be a given quantum circuit such that for every $\ket{i}_n \in \{0,1\}^n$,
        \begin{equation}
            \mathcal{Q}\ket{i}_n = \sum_{j = (j_1,\ldots,j_n) \in \{0,1\}^n} \alpha_{i,j} \ket{j_{1}}\cdots\ket{j_n},
        \end{equation}
        where $\alpha_{i,j} \in \C$. Then, $\Tau_\pi \mathcal{Q}$ is also a quantum circuit such that for every
        $\ket{i}_n \in \{0,1\}^n$, 
        \begin{equation}
            \Tau_{\pi}\mathcal{Q}\ket{i}_n = \sum_{j \in \{0,1\}^n} \alpha_{i,j} \ket{j_{\pi(1)}}\cdots\ket{j_{\pi(n)}}.
        \end{equation}
    \end{enumerate}

    \begin{proof}
        Firstly, we show, for any $j,k \in \{1,2,\ldots,n\}$, $j < k$, and for any $n$-qubit $\ket{i}_n = \ket{i_1}\cdots\ket{i_n} \in \{0,1\}^n$, that there exists a quantum circuit $\mathcal{T}_{j \leftrightarrow k}$ consisting of $2(k-j) -1$ swap gates
          where it holds that 
        \begin{equation}\label{eqn: swap 2 position circuit}
            \begin{split}
                \mathcal{T}_{j \leftrightarrow k}:& \ket{i}_n = \ket{i_1}\cdots \ket{i_{j-1}}\ket{i_j}\ket{i_{j+1}}\cdots \ket{i_{k-1}}\ket{i_k}\ket{i_{k+1}}\cdots\ket{i_n}\\
                &\mapsto \ket{i_1}\cdots \ket{i_{j-1}}\ket{i_k}\ket{i_{j+1}}\cdots \ket{i_{k-1}}\ket{i_j}\ket{i_{k+1}}\cdots\ket{i_n}.
            \end{split}
        \end{equation}
        Denote by $\mathcal{S} := \SWAP \in U(2^2)$ the swap gate
         which satisfy for all $\ket{i_1}\ket{i_2} \in \{0,1\}^2$ that 
        \begin{equation}
            \mathcal{S}\ket{i_1}\ket{i_2} = \ket{i_2}\ket{i_1}.
        \end{equation}
        If $k - j = 1$, (i.e. $k = j+1$) then we simply set
        \begin{equation}\label{eqn: swap 2 position circuit pt2}
            T_{j \leftrightarrow j+1} = I_2^{\otimes j-1} \otimes \mathcal{S} \otimes I_2^{\otimes n- j - 1}.
        \end{equation}
        If $k-j \geq 2$, then proceeding inductively set 
        \begin{equation}\label{eqn: swap 2 position circuit pt3}
            \mathcal{T}_{j \leftrightarrow k} = \prod_{l=1}^{k-j-1} (I_2^{\otimes k - 2 - l} \otimes \mathcal{S}\otimes I_2^{\otimes n - k + l})  \prod_{l=0}^{k-j-1}( I_2^{\otimes j - 1 + l} \otimes \mathcal{S} \otimes I_2^{\otimes n - j - 1 - l }),
        \end{equation} 
    (c.f. Definition \ref{def: quantum circuit}), where we use the usual convention that $I_2^{\otimes 0} = 1 \in \C$ and $A \otimes 1 = A = 1 \otimes A$ for any $A \in U(2^{m})$, $m \in \N$. By direct verification, we note that $\mathcal{T}_{j \leftrightarrow k}$ satisfy \eqref{eqn: swap 2 position circuit} for all $\ket{i}_n \in \{0,1\}^n$ and that only $2(k-j)-1$ swap gates were required in its construction. 

        Secondly, by the cycle decomposition theorem (Proposition \ref{prop: cycle decomp}), the given permutation $\pi$ can be written as a composition of disjoint cycles (c.f. Definition \ref{def: cycle}, Proposition \ref{prop: cycle decomp}) as 
        \begin{equation}\label{eqn: cycle decomp}
            \pi = C_1 C_2 \cdots C_k = (a_1 a_2 \cdots a_{m_1})(a_{m_1 + 1} a_{m_1 + 2} \cdots a_{m_2}) \cdots (a_{m_{k - 1}+1} a_{m_{k-1}+2} \cdots a_{m_{k}}),
        \end{equation}
        where $k$ is the number of cycles, and $a_1,\ldots,a_{m_k} \in \{1,\ldots,n\}$ are distinct numbers. Note  by convention that each of these cycles has length $m_l \geq 2$. For each of these cycles $C_l = (a_{m_{l-1}+1}\cdots a_{m_{l}})$, $l=1\ldots,k$, with $m_0 := 0$ we construct the quantum circuits $\mathcal{T}_{C_l}$, $l=1,\ldots,k$, via 
        \begin{equation}\label{eqn: cycle circuit}
            \mathcal{T}_{C_l} = \prod_{i=m_{l-1}}^{m_{l}-1} \mathcal{T}_{a_i \leftrightarrow a_{i+1}} = \mathcal{T}_{a_{m_{l}-1} \leftrightarrow a_{m_{l}}} \cdots \mathcal{T}_{a_{m_{l-1}+1} \leftrightarrow a_{m_{l-1} + 2}}\mathcal{T}_{a_{m_{l-1}} \leftrightarrow a_{m_{l-1} + 1}},
        \end{equation}
        where the quantum circuits $\mathcal{T}_{a_i \leftrightarrow a_{i+1}}$ are constructed based on the first step of the proof (c.f. \eqref{eqn: swap 2 position circuit}, \eqref{eqn: swap 2 position circuit pt2}, \eqref{eqn: swap 2 position circuit pt3}). Finally, we construct the quantum circuit $\mathcal{T}_\pi$ via
        \begin{equation}\label{eqn: product cycle circuit}
            \mathcal{T}_\pi = \prod_{l=1}^k \mathcal{T}_{C_l} = \mathcal{T}_{C_k} \cdots \mathcal{T}_{C_2} \mathcal{T}_{C_1}.
        \end{equation} 
       Thus, \eqref{eqn: swap 2 position circuit}, \eqref{eqn: cycle decomp}, \eqref{eqn: cycle circuit}, and \eqref{eqn: product cycle circuit} imply that the quantum circuit $\mathcal{T}_\pi$ satisfies \eqref{eqn: permutation circuit}. Moreover, we note that the total number of quantum circuits of the form $\mathcal{T}_{j \leftrightarrow k}$ (c.f. \eqref{eqn: swap 2 position circuit}) is $m_k \leq n$, where each of these quantum circuits requires $2 \vert a_i - a_{i+1} \vert - 1 \leq (2n-1)$ swap gates. Hence, the total number of swap gates used to construct $\mathcal{T}_{\pi}$ is at most $n\cdot (2n-1) \leq 2n^2$.  Thus, we have proved the first statement of the lemma. The second statement of the lemma follows directly from the fact that $\mathcal{T}_\pi$ is a linear operator. 
    \end{proof}

\end{lemma}

\begin{lemma}[Quantum circuit for addition](\cite[Section 3.1, QNMAdd]{QC12_arithmetic2})\label{lemma: quantum addition}
Let $n_1,n_2 \in \N$, with $n_1 \geq n_2$. Then, there is a quantum circuit $\mathcal{Q}_{(+)}$ on $(n_1+n_2+1)$ qubits such that for any $a \in \mathbb{F}_{n_1,0}, b \in \mathbb{F}_{n_2,0}$,
\begin{equation}
    \mathcal{Q}_{(+)}: \ket{0} \ket{a}_{n_1} \ket{b}_{n_2}  \mapsto \ket{a \boxplus b}_{n_1+1}\ket{b}_{n_2}.
\end{equation}
The quantum circuit $\mathcal{Q}_{(+)}$ requires $n_1^2 + 3n_1 + 18 + \frac{1}{2}(n_2(2n_1-n_2+3))$ elementary gates.
\end{lemma}

\begin{figure}
	\centering 
	\boxed{
		\begin{quantikz}[transparent, row sep={20pt,between origins},column sep=20pt,font=\small]
\lstick[wires=3]{$\ket{a}_{n_1+m_1}$}  & \gate[7]{\mathcal{T}_{\pi_1}} & \gate[7][3cm]{\mathcal{Q}_{(+)}}\gateinput{$\ket{0}$}\gateoutput[4]{$\ket{a \boxplus b}$} & \gate[7]{\mathcal{T}_{\pi_2}} & \qw \rstick[3]{$\ket{b}_{n_2+m_2}$}\\
& &\gateinput[3]{$\ket{a}_{n_1+m_1}$}  & & \qw \\
& & & & \qw \\
	\lstick[3]{$\ket{b}_{n_2+m_2}$} &  & & &\qw \rstick[4]{$\ket{a \boxplus b}_{n_1+m_1+1}$} \\
& &\gateinput[3]{$\ket{b}_{n_2+m_2}$} \gateoutput[3]{$\ket{b}_{n_2+m_2}$} & & \qw \\
& & & &  \qw\\
\lstick{$\ket{0}$} & & & & \qw
		\end{quantikz}		
	} 
	\caption{Circuit diagram for $\widetilde{\mathcal{Q}}_{(+)}$ in Corollary~\ref{corollary: add circuit}.}
	\label{fig: add circuit diagram}
\end{figure}

\begin{corollary}[Quantum circuit for addition with fractional part]\label{corollary: add circuit}
    Let $n_1,n_2,m_1,m_2 \in \N$, with $n_1+m_1 \geq n_2 + m_2$. Let $n = n_1+n_2$, and $m = m_1+m_2$. Then, there is a quantum circuit $\widetilde{\mathcal{Q}}_{(+)}$ on  $(n+ m+1)$ qubits such that for any $a \in \mathbb{F}_{n_1,m_1}$, $b \in \mathbb{F}_{n_2,m_2}$,
        \begin{equation}\label{eqn: cor addition}
            \widetilde{\mathcal{Q}}_{(+)}:\ket{a}_{n_1+m_1}\ket{b}_{n_2+m_2}\ket{0} \mapsto \ket{b}_{n_2+m_2}\ket{a \boxplus b}_{n_1+m_1+1}.
        \end{equation}
        The quantum circuit $\widetilde{\mathcal{Q}}_{(+)}$ requires at most $29(n+m+1)^2$ elementary gates. See Figure~\ref{fig: add circuit diagram} for the circuit diagram.
\begin{proof}
    By Lemma \ref{lemma: permutation circuit}, there is a quantum circuit $\mathcal{T}_{\pi_1}$ with at most $2(n_1+m_1+n_2+m_2+1)^2 = 2(n+m+1)^2$ swap gates satisfying 
    \begin{equation}\label{eqn: pf addition circuit 1}
        \mathcal{T}_{\pi_1}: \ket{a}_{n_1+m_1}\ket{b}_{n_2+m_2}\ket{0}  \mapsto \ket{0} \ket{a}_{n_1+m_1}\ket{b}_{n_2+m_2}.
    \end{equation}
    Note that by Lemma \ref{lemma: addition in TC},  we can extend the addition operation $\boxplus: \mathbb{F}_{n_1+m_1,0} \times \mathbb{F}_{n_2+m_2,0} \to \mathbb{F}_{n+m+1,0}$ to $\boxplus: \mathbb{F}_{n_1,m_1} \times \mathbb{F}_{n_2,m_2} \to \mathbb{F}_{\tilde{n}+1,\tilde{m}}$ where $\tilde{n} = \max\{n_1,n_2\}$ and $\tilde{m} = \max\{m_1,m_2\}$. This, the hypothesis that $n_1+m_1 \geq n_2 +m_2$, and Lemma \ref{lemma: quantum addition} (with $n_1 \leftarrow n_1 + m_1$, $n_2 \leftarrow n_2+m_2$ in the notation of Lemma \ref{lemma: quantum addition}) imply that there exists a quantum circuit $\mathcal{Q}_{(+)}$ such that for any $a \in \mathbb{F}_{n_1,m_1}$, $b \in \mathbb{F}_{n_2,m_2}$ that
    \begin{equation}\label{eqn: pf addition circuit 2}
        \mathcal{Q}_{(+)}:\ket{0} \ket{a}_{n_1+m_1}\ket{b}_{n_2+m_2} \mapsto \ket{a \boxplus b}_{n_1+m_1+1}\ket{b}_{n_2+m_2},
    \end{equation}
    and that the number of elementary gates required to construct $\mathcal{Q}_{(+)}$ is
    \begin{equation} 
        (n_1+m_1)^2 + 3(n_1+m_1) + 18 + \tfrac{1}{2}(n_2+m_2)(2(n_1+m_1) - (n_2+m_2) + 3).
    \end{equation}
    Moreover, by Lemma \ref{lemma: permutation circuit}, there is a quantum circuit $\mathcal{T}_{\pi_2}$ such that
    \begin{equation}\label{eqn: pf addition circuit 3}
         \ket{a \boxplus b}_{n_1+m_1+1}\ket{b}_{n_2+m_2} \mapsto  \ket{b}_{n_2+m_2}\ket{a \boxplus b}_{n_1+m_1+1},
    \end{equation}
    which uses at most $2(n+m+1)^2$ swap gates.

    Define the quantum circuit $\widetilde{\mathcal{Q}}_{(+)} := \mathcal{T}_{\pi_2}\mathcal{Q}_{(+)}\mathcal{T}_{\pi_1}$. Observe that \eqref{eqn: pf addition circuit 1}, \eqref{eqn: pf addition circuit 2}, and \eqref{eqn: pf addition circuit 3} shows that  $\widetilde{\mathcal{Q}}_{(+)}$ satisfies \eqref{eqn: cor addition}, and that the total number of elementary gates required to construct $\widetilde{\mathcal{Q}}_{(+)}$ is at most
    \begin{equation}
        \begin{split}
            &2(n + m+ 1)^2 + (n_1+m_1)^2 + 3(n_1+m_1) + 18 + \tfrac{1}{2}(n_2+m_2)(2(n_1+m_1) - (n_2+m_2) + 3) +2(n + m+ 1)^2 \\
            &\leq 2(n + m+ 1)^2 + (n_1+m_1)^2 + 3(n_1+m_1) + 18 +  (n+m+1)^2 + \frac{3}{2}(n+m+1)^2 + 2(n + m+ 1)^2\\
            &\leq (2+1+3+18+1+2+2)(n+m+1)^2\\
            &= 29(n+m+1)^2.
        \end{split}
    \end{equation}
%
%
%
\end{proof}
\end{corollary}

\begin{lemma}[Quantum circuit for multiplication](\cite[Section 3.5, QNMMul]{QC12_arithmetic2})\label{lemma: quantum mult1}
    Let $n_1,n_2 \in \N$, with $n_1 \geq n_2$. Then, there is a quantum circuit $\mathcal{Q}_{(\times)}$ on $(2n_1+3n_2+3)$ qubits such that for any $a \in \mathbb{F}_{n_1,0}, b \in \mathbb{F}_{n_2,0}$, 
    \begin{equation}
        \mathcal{Q}_{(\times)}:\ket{0}\ket{0}_{n_1+n_2}\ket{a}_{n_1} \ket{b}_{n_2} \ket{0}_{n_2} \ket{0}_{2}\mapsto \ket{\anc} \ket{a \boxdot b}_{n_1+n_2}\ket{a}_{n_1}\ket{b}_{n_2} \ket{\anc}_{n_2+2}.
    \end{equation}
    The quantum circuit $\mathcal{Q}_{(\times)}$ requires $(\frac{1}{2}(5n_1^2 + n_1) + 4n_2^2 + 4n_1n_2 + 6n_2 + 7)$ elementary gates. 
\end{lemma}

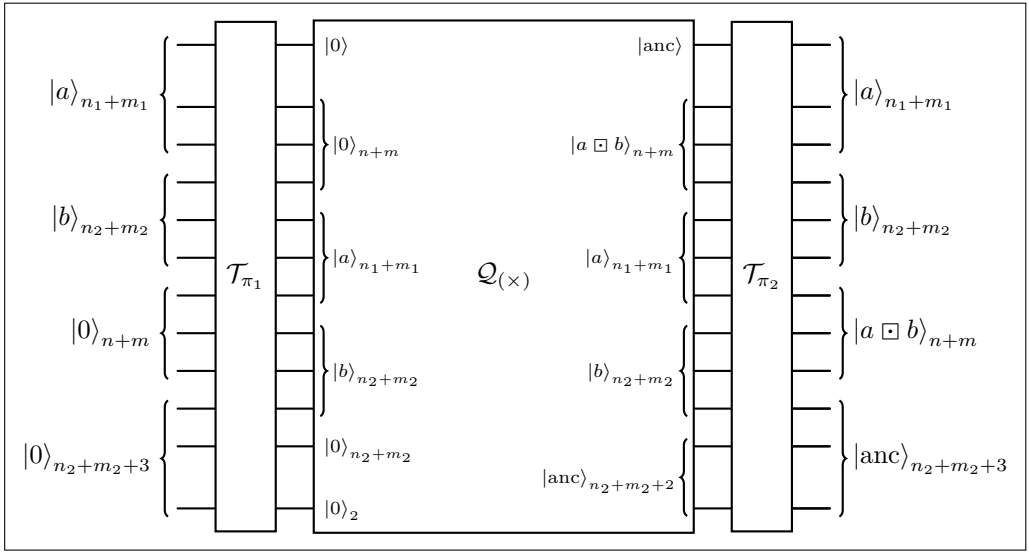
\begin{figure}
	\centering 
	\boxed{
	\begin{quantikz}
		\lstick[wires=3]{$\ket{a}_{n_1+m_1}$}  & \gate[12]{\mathcal{T}_{\pi_1}} & \gate[12][5cm]{\mathcal{Q}_{(\times)}}\gateinput{$\ket{0}$}\gateoutput{$\ket{\anc}$} &\gate[12]{\mathcal{T}_{\pi_2}} &\qw  \rstick[3]{$\ket{a}_{n_1+m_1}$}\\
		& &\gateinput[3]{$\ket{0}_{n+m}$}\gateoutput[3]{$\ket{a \boxdot b}_{n+m}$} & & \qw \\
		& & & & \qw \\
		\lstick[3]{$\ket{b}_{n_2+m_2}$} &  & & &\qw \rstick[3]{$\ket{b}_{n_2+m_2}$} \\
		& &\gateinput[3]{$\ket{a}_{n_1+m_1}$}\gateoutput[3]{$\ket{a }_{n_1+m_1}$} & & \qw \\
		& & & & \qw \\
		\lstick[3]{$\ket{0}_{n + m}$} & & & &\qw \rstick[3]{$\ket{a \boxdot b}_{n+m}$} \\
		& &\gateinput[3]{$\ket{b}_{n_2+m_2}$}\gateoutput[3]{$\ket{b}_{n_2+m_2}$} & & \qw \\
		& & & & \qw \\	
		\lstick[3]{$\ket{0}_{n_2 + m_2 + 3}$} & & & &\qw \rstick[3]{$\ket{\mathrm{anc}}_{n_2+m_2+3}$} \\
		& &\gateinput{$\ket{0}_{n_2+m_2}$}\gateoutput[2]{$\ket{\anc}_{n_2+m_2+2}$} & & \qw \\
		& &\gateinput{$\ket{0}_2$} & & \qw 
	\end{quantikz}	
	} 
	\caption{Circuit diagram for $\widetilde{\mathcal{Q}}_{(\times)}$ in Corollary~\ref{corollary: mult circuit}.} 
	\label{fig: mult circuit diagram}
\end{figure}
\begin{corollary}[Quantum circuit for multiplication with fractional part]\label{corollary: mult circuit}
    Let $n_1,n_2,m_1,m_2 \in \N$, with $n_1+m_1 \geq n_2 + m_2$. Let $n := n_1+n_2$ and $m := m_1 + m_2$. Then, there is a quantum circuit $\widetilde{\mathcal{Q}}_{(\times)}$ on $(2n+2m + n_2+m_2 +3)$ qubits such that for any $a \in \mathbb{F}_{n_1,m_1}$, $b \in \mathbb{F}_{n_2,m_2}$,
    \begin{equation}\label{eqn: cor mult}
        \widetilde{\mathcal{Q}}_{(\times)}: \ket{a}_{n_1+m_1}\ket{b}_{n_2+m_2}\ket{0}_{n+m}\ket{0}_{n_2+m_2+3} \mapsto \ket{a}_{n_1+m_1}\ket{b}_{n_2+m_2}\ket{a \boxdot b}_{n+m}\ket{\anc}_{n_2+m_2+3}.
    \end{equation}
        The quantum circuit $\widetilde{\mathcal{Q}}_{(\times)}$ requires at most $61 (n+m+1)^2$ elementary gates. See Figure~\ref{fig: mult circuit diagram} for the circuit diagram.
\begin{proof}
    By Lemma \ref{lemma: permutation circuit}, there is a quantum circuit $\mathcal{T}_{\pi}$ with at most $2(2n+2m + n_2+m_2 +3)^2$ swap gates satisfying for any $a \in \mathbb{F}_{n_1,m_1}$, $b \in \mathbb{F}_{n_2,m_2}$ that 
    \begin{equation}\label{eqn: cor mult eqn 1}
        \mathcal{T}_{\pi}: \ket{a}_{n_1+m_1}\ket{b}_{n_2+m_2}\ket{0}_{n+m}\ket{0}_{n_2+m_2+3}  \mapsto \ket{0}\ket{0}_{n+m} \ket{a}_{n_1+m_1}\ket{b}_{n_2+m_2}\ket{0}_{n_2+m_2}\ket{0}_2.
    \end{equation}
    Note that by Lemma \ref{lemma: multiplication in TC}, we can extend the multiplication operation $\boxdot: \mathbb{F}_{n_1+m_1,0} \times \mathbb{F}_{n_2+m_2,0} \to \mathbb{F}_{n+m,0}$ to $\boxdot: \mathbb{F}_{n_1,m_1} \times \mathbb{F}_{n_2,m_2} \to \mathbb{F}_{n,m}$. This, the condition that $n_1+m_1 \geq n_2+m_2$, and Lemma \ref{lemma: quantum mult1} (with $n_1 \leftarrow n_1 + m_1, n_2 \leftarrow n_2 + m_2$ in the notation of Lemma \ref{lemma: quantum mult1}) imply that there exists a quantum circuit $\mathcal{Q}_{(\times)}$ such that for any $a \in \mathbb{F}_{n_1,m_1}$, $b \in \mathbb{F}_{n_2,m_2}$ 
    \begin{equation}\label{eqn: cor mult eqn 2}
        \mathcal{Q}_{(\times)}:\ket{0}\ket{0}_{n+m}\ket{a}_{n_1+m_1} \ket{b}_{n_2+m_2} \ket{0}_{n_2+m_2} \ket{0}_{2}\mapsto \ket{\anc} \ket{a \boxdot b}_{n+m}\ket{a}_{n_1+m_1} \ket{b}_{n_2+m_2} \ket{\anc}_{n_2+m_2+2},
    \end{equation}
    and that the number of elementary gates required to construct $\mathcal{Q}_{(\times)}$ is at most
    \begin{equation}
        (\frac{1}{2}(5(n_1+m_1)^2 + (n_1+m_1)) + 4(n_2+m_2)^2 + 4(n_1+m_1)(n_2+m_2) + 6(n_2+m_2) + 7).
    \end{equation}
    By another application of Lemma \ref{lemma: permutation circuit}, there is a quantum circuit $\mathcal{T}_{\pi'}$ with at most 
    $2(2n+2m + n_2+m_2 +3)^2$ swap gates satisfying
    \begin{equation}\label{eqn: cor mult eqn 3}
        \mathcal{T}_{\pi'}:\ket{\anc} \ket{a \boxdot b}_{n+m}\ket{a}_{n_1+m_1} \ket{b}_{n_2+m_2} \ket{\anc}_{n_2+m_2+2} \mapsto \ket{a \boxdot b}_{n+m}\ket{a}_{n_1+m_1} \ket{b}_{n_2+m_2} \ket{\anc}_{n_2+m_2+3}.
    \end{equation}
    We define the quantum circuit $\widetilde{\mathcal{Q}}_{(\times)}=\mathcal{T}_{\pi'}\mathcal{Q}_{(\times)}\mathcal{T}_{\pi}$. Hence, \eqref{eqn: cor mult eqn 1}, \eqref{eqn: cor mult eqn 2}, and  \eqref{eqn: cor mult eqn 3} imply that the quantum circuit $\widetilde{\mathcal{Q}}_{(\times)}$ satisfy \eqref{eqn: cor mult} for all $a \in \mathbb{F}_{n_1,m_1}$, $b \in \mathbb{F}_{n_2,m_2}$. Moreover, the number of elementary gates required to construct $\widetilde{\mathcal{Q}}_{(\times)}$ is at most
    \begin{equation}
        \begin{split}
    &  2(2n+2m + n_2+m_2 +3)^2 + \frac{1}{2}(5(n_1+m_1)^2 + (n_1+m_1)) + 4(n_2+m_2)^2 + 4(n_1+m_1)(n_2+m_2) \\
    &\qquad + 6(n_2+m_2) + 7 + 2(2n+2m + n_2+m_2 +3)^2\\
    &\leq 2\cdot 3^2 (n+m+1)^2 + \frac{5}{2}(n+m+1)^2 + \frac{1}{2}(n+m+1) + 4(n+m+1)^2 + 4(n+m+1)^2\\
    &\qquad + 6(n+m+1) + 7(n+m+1) + 2\cdot3^2 (n+m+1)^2\\
    &\leq (18 + 3 + 1 + 4 + 4 + 6 + 7 + 18)(n+m+1)^2\\
    &= 61 (n+m+1)^2.
        \end{split}
    \end{equation}
    
\end{proof}
\end{corollary}

\begin{lemma}[Quantum circuit for integer comparison](\cite[Section 3.4, QComp]{QC12_arithmetic2})\label{lemma: quantum comp}
Let $n_1,n_2 \in \N$, with $n_1 \geq n_2$. Let $n = n_1+n_2$. Then, there is a quantum circuit $\mathcal{Q}_{(\mathrm{comp})}$ on $(n_1+n_2+4)$ qubits such that for any $a \in \mathbb{F}_{n_1,0}, b \in \mathbb{F}_{n_2,0}$,
\begin{equation}
    \mathcal{Q}_{(\mathrm{comp})}:\ket{0}\ket{a}_{n_1}\ket{b}_{n_2}\ket{0}\ket{0}\ket{0} \mapsto \ket{a\boxminus b}_{n_1+1}\ket{b}_{n_2}\ket{c_1}\ket{c_2}\ket{c_3},
\end{equation}
where\footnote{The notation $a \boxminus b $ here refers to subtraction for a pair of two complement numbers, i.e., for any $a,b \in \mathbb{F}_{n,m}$, $a \boxminus b  \in \mathbb{F}_{n+1,m}$ is defined by $\mathrm{E}_{n+1,m}(\mathrm{D}_{n,m}(a)-\mathrm{D}_{n,m}(b))$.}
\begin{equation}
    \ket{c_1}\ket{c_2}\ket{c_3} = \begin{cases}
        \ket{1}\ket{0}\ket{0}, \quad \text{if $\mathrm{D}_{n_1,0}(a)>\mathrm{D}_{n_2,0}(b)$}, \\ 
        \ket{0}\ket{1}\ket{0}, \quad \text{if $\mathrm{D}_{n_1,0}(a)<\mathrm{D}_{n_2,0}(b)$},\\
        \ket{0}\ket{0}\ket{1}, \quad \text{if $\mathrm{D}_{n_1,0}(a)=\mathrm{D}_{n_2,0}(b)$}.
    \end{cases}
\end{equation}
The quantum circuit $\mathcal{Q}_{(\mathrm{comp})}$ uses $(n_1^2 + 3n_1 + 41 + n_2(2n_1-n_2+3)/2)$ elementary gates.
\end{lemma}

\begin{figure}
	\centering 
	\boxed{
		\begin{quantikz}[row sep={20pt,between origins},column sep=20pt,font=\small]
\lstick[3]{$\ket{a}_{n_1+m_1}$} & \gate[10]{\mathcal{T}_{\pi}} & \qw & \gate[7][3cm]{\mathcal{Q}_{(+)}}\gateinput{$\ket{0}$}\gateoutput[4]{$\ket{(a \boxminus b) \boxplus b}$} & \gate[10]{\mathcal{T}_{\pi}} &\qw \rstick[3]{$\ket{a}_{n_1+m_1}$}\\
&\qw &\gate[9][5cm]{\mathcal{Q}_{(\mathrm{comp})}}\gateinput{$\ket{0}$}\gateoutput[3]{$\ket{a \boxminus b}_{n_1+m_1+1}$} &\qw\gateinput[3]{$\ket{a \boxminus b}$} &\qw &\qw \\
&\qw &\qw\gateinput[3]{$\ket{a}_{n_1+m_1}$} &\qw &\qw &\qw \\
\lstick[3]{$\ket{b}_{n_2+m_2}$} &\qw &\qw &\qw &\qw &\qw  \rstick[3]{$\ket{b}_{n_2+m_2}$}\\
&\qw &\qw\gateoutput[3]{$\ket{b}_{n_2+m_2}$} &\qw\gateinput[3]{$\ket{b}$}\gateoutput[3]{$\ket{b}$} &\qw &\qw \\
&\qw &\qw\gateinput[3]{$\ket{b}_{n_2+m_2}$} &\qw &\qw &\qw \\
\lstick{$\ket{0}$} &\qw &\qw &\qw &\qw &\qw \rstick{$\ket{c_1}$}\\
\lstick{$\ket{0}$} &\qw &\qw\gateoutput{$\ket{c_1}$} &\qw &\qw &\qw \rstick{$\ket{c_2}$}\\
\lstick{$\ket{0}$} &\qw &\qw\gateinput{$\ket{0}$}\gateoutput{$\ket{c_2}$} &\qw &\qw &\qw \rstick{$\ket{c_3}$}\\
\lstick{$\ket{0}_2$} &\qw &\qw\gateinput{$\ket{0}$}\gateoutput{$\ket{c_3}$} &\qw &\qw &\qw \rstick{$\ket{\mathrm{anc}}_2$}
		\end{quantikz}		
	} 
	\caption{Circuit diagram for ${\widetilde{\mathcal{Q}}}_{(\mathrm{comp})}$ in Corollary~\ref{corollary: comp circuit}.} 
	\label{fig: comp circuit diagram}
\end{figure}
\begin{corollary}[Quantum circuit for fractional comparison]\label{corollary: comp circuit}
    
        Let $n_1,n_2,m_1,m_2 \in \N$, with $n_1+m_1 \geq n_2 + m_2$. Let $n = n_1+n_2$, and $m=m_1+m_2$. Then, there is a quantum circuit $\widetilde{\mathcal{Q}}_{(\mathrm{comp})}$ on $(n+m+5)$ qubits such that for any $a \in \mathbb{F}_{n_1,m_1}$, $b \in \mathbb{F}_{n_2,m_2}$,
        \begin{equation}\label{eqn: comp1}
            \widetilde{\mathcal{Q}}_{(\mathrm{comp})}: \ket{a}_{n_1+m_1}\ket{b}_{n_2+m_2}\ket{0}_{5} \mapsto \ket{a}_{n_1+m_1}\ket{b}_{n_2+m_2}\ket{c_1}\ket{c_2}\ket{c_3}\ket{\anc}_2,
        \end{equation}
    where
        \begin{equation}\label{eqn: comp2}
        \ket{c_1}\ket{c_2}\ket{c_3} = \begin{cases}
            \ket{1}\ket{0}\ket{0}, \quad \text{if $\mathrm{D}_{n_1,m_1}(a)>\mathrm{D}_{n_2,m_2}(b)$}, \\ 
            \ket{0}\ket{1}\ket{0}, \quad \text{if $\mathrm{D}_{n_1,m_1}(a)<\mathrm{D}_{n_2,m_2}(b)$},\\
            \ket{0}\ket{0}\ket{1}, \quad \text{if $\mathrm{D}_{n_1,m_1}(a)=\mathrm{D}_{n_2,m_2}(b)$}.
        \end{cases}
    \end{equation}
        The quantum circuit $\widetilde{\mathcal{Q}}_{(\mathrm{comp})}$ requires at most $175(n+m+1)^2$ elementary gates. See Figure~\ref{fig: comp circuit diagram} for the circuit diagram.
        \begin{proof}
            The construction of the quantum circuit $\widetilde{\mathcal{Q}}_{(\mathrm{comp})}$ consists of the following steps.
            
            \begin{enumerate}
                \item We first employ the permutation circuit $\Tau_\pi$ from  Lemma \ref{lemma: permutation circuit} so that 
                \begin{equation}\label{eqn: proof comp3}
                    \Tau_\pi: \ket{a}_{n_1+m_1}\ket{b}_{n_2+m_2}\ket{0}_{5} \mapsto \ket{0}\ket{0} \ket{a}_{n_1+m_1}\ket{b}_{n_2+m_2}\ket{0} \ket{0}\ket{0}.
                \end{equation}
                In this step, the number of elementary gates used is at most $2(n+m+5)^2$.
                
                \item Next, we use the comparison quantum circuit $\mathcal{Q}_{(\mathrm{comp})}$ from Lemma \ref{lemma: quantum comp} (with $n_1 \leftarrow n_1+m_1$, $n_2 \leftarrow n_2 + m_2$ in the notation of Lemma \ref{lemma: quantum comp}) to obtain 
                \begin{equation}\label{eqn: proof comp4}
                    I_2 \otimes \mathcal{Q}_{(\mathrm{comp})}:\ket{0}\ket{0}\ket{a}_{n_1+m_1}\ket{b}_{n_2+m_2}\ket{0}\ket{0}\ket{0} \mapsto \ket{0}\ket{a \boxminus b}_{n_1+m_1+1}\ket{b}_{n_2+m_2}\ket{c_1}\ket{c_2}\ket{c_3},
                \end{equation}
                where $\ket{c_1}\ket{c_2}\ket{c_3}$ satisfy \eqref{eqn: comp2}. In this step, the number of elementary gates used is 
                \begin{equation}
                    (n_1+m_1)^2 + 3(n_1+m_1) + 41 + (n_2+m_2)\big(2(n_1+m_1)-(n_2+m_2)+3\big)/2.
                \end{equation}
            
                \item We use the adder quantum circuit $Q_{(+)}$ from Lemma \ref{lemma: quantum addition} (with $n_1 \leftarrow n_1 + m_1+1$, $n_2 \leftarrow n_2+m_2$ in the notation of Lemma \ref{lemma: quantum addition}) to obtain 
                \begin{equation}
                    Q_{(+)}\otimes I_2^{\otimes 3}: \ket{0}\ket{a \boxminus b}_{n_1+m_1+1}\ket{b}_{n_2+m_2}\ket{c_1}\ket{c_2}\ket{c_3} \mapsto \ket{(a \boxminus b)\boxplus b}_{n_1+m_1+2}\ket{b}_{n_2+m_2}\ket{c_1}\ket{c_2}\ket{c_3}.
                \end{equation}            
                In this step, the number of elementary gates used is 
                \begin{equation}
                    (n_1+m_1+1)^2+3(n_1+m_1+1)+18+\frac{1}{2}\big((n_2+m_2)(2(n_1+m_1+1)-(n_2+m_2)+3)\big).
                \end{equation}
                \item As elements of $\mathbb{F}_{n_1+2,m_2}$, it can be directly checked that 
                \begin{equation}
                    (a \boxminus b)\boxplus b = \begin{cases}
                        ((1,1,a_{n_1-1},\ldots,a_0),(a_{-1},\ldots,a_{-m_1})) & \text{if $\mathrm{D}_{n_1,m_1}(a)<0$},\\
                       ((0,0,a_{n_1-1},\ldots,a_0),(a_{-1},\ldots,a_{-m_1})) & \text{if $\mathrm{D}_{n_1,m_1}(a)\geq0$}.
                    \end{cases}
                \end{equation}
                In the above expression, we treat the leftmost two bits as ancilla qubits and rewrite $\ket{(a \boxminus b)\boxplus b}_{n_1+m_1+2} = \ket{\anc}_2\ket{a}_{n_1+m_1}$. We apply another permutation circuit $\Tau_\pi$ from  Lemma \ref{lemma: permutation circuit} to obtain 
                \begin{equation}
                    \Tau_\pi: \ket{(a \boxminus b)\boxplus b}_{n_1+m_1+2}\ket{b}_{n_2+m_2}\ket{c_1}\ket{c_2}\ket{c_3} \mapsto \ket{a}_{n_1+m_1}\ket{b}_{n_2+m_2}\ket{c_1}\ket{c_2}\ket{c_3}\ket{\anc}_2.
                \end{equation}
            In this step, the number of elementary gates used is at most $2(n+m+5)^2$.
            \end{enumerate}
            The desired quantum circuit $\widetilde{Q}_{(\mathrm{comp})}$ is constructed from the above steps.  We note that the number of elementary gates used to construct $\widetilde{Q}_{(\mathrm{comp})}$ is at most
            \begin{equation}
                \begin{split}
                    &2(n+m+5)^2 + [(n_1+m_1)^2 + 3(n_1+m_1) + 41 + (n_2+m_2)(2(n_1+m_1)-(n_2+m_2)+3)/2]\\
                    &\quad +[(n_1+m_1+1)^2+3(n_1+m_1+1)+18+\frac{1}{2}\big((n_2+m_2)(2(n_1+m_1+1)-(n_2+m_2)+3)\big)] + 2(n+m+5)^2\\
                    &\leq 2\cdot5^2(n+m+1)^2+ [(n+m+1)^2+ 3(n+m+1) + 41 + (n+m+1)^2 + \frac{3}{2}(n+m+1)]\\
                    &\quad + [(n+m+1)^2 + 3(n+m+1) + 18 + (n+m+1)^2 + \frac{3}{2}(n+m+1)^2 + 2(n+m+1)^2] + 2\cdot5^2(n+m+1)^2\\
                    &\leq (50 + 48 + 27+ 50)(n+m+1)^2\\
                    &= 175(n+m+1)^2.
                \end{split}
            \end{equation}
            
        \end{proof}
\end{corollary}

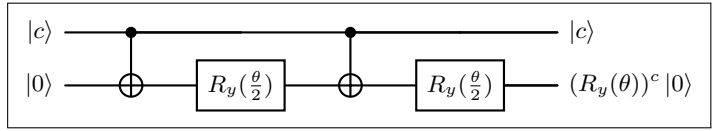
\begin{figure}
	\centering 
	\boxed{
		\begin{quantikz}[row sep={20pt,between origins},column sep=20pt,font=\small]
	\lstick{$\ket{c}$} & \ctrl{1} &\qw & \ctrl{1} &\qw &\qw \rstick{$\ket{c}$}\\
	\lstick{$\ket{0}$} & \targ{} &\gate{R_y(\tfrac{\theta}{2})}	& \targ{}  &\gate{R_y(\tfrac{\theta}{2})}	&\qw \rstick{$(R_y(\theta))^c\ket{0}$}
		\end{quantikz}		
	} 
	\caption{Circuit diagram for $CR_y(\theta)$ in Lemma~\ref{lemma: controlled y-rotation}.} 
	\label{fig: controlled y-rotation diagram}
\end{figure}
\begin{lemma} [Controlled $Y$-rotations]\label{lemma: controlled y-rotation}
For any $\theta \in (0,4\pi)$, there is a controlled $Y$-rotation gate acting on two qubits that performs the following operation
\begin{equation}
    CR_y(\theta):\ket{c}\ket{0} \mapsto \ket{c} (R_y(\theta))^{c}\ket{0} = \begin{cases} \ket{c}\ket{0}, \quad \text{if $c = 0$},\\ \ket{c}(\cos(\theta/2)\ket{0} + \sin(\theta/2)\ket{1}), \quad \text{if $c = 1$}. \end{cases}
\end{equation}
The quantum circuit to construct $CR_y(\theta)$ requires two $R_y(\theta/2)$ gates 
and two $\CNOT$ gates.
See Figure~\ref{fig: controlled y-rotation diagram} for the circuit diagram.
\begin{proof}
    The quantum circuit can be constructed by the following definition
    \begin{equation}
        CR_y(\theta) = (I_2\otimes R_y(\theta/2))(\CNOT)(I_2\otimes R_y(\theta/2))(\CNOT).
    \end{equation}
\end{proof}
\end{lemma}

\subsection{Distribution loading}\label{section: distribution quantum circuit}
The task of loading an arbitrary $n$-qubit state on a quantum computer is known generally to be a hard problem, as highlighted, e.g., in \cite{QC17_knill}. However, in some cases, the problem of loading states  representing certain probability distributions on a quantum computer have been shown to be polynomially tractable. Grover and Rudolph have shown an efficient method to load a discrete approximation of any log-concave probability distributions \cite{QC18_Grover_distribution}. 
%
Zoufal et al.\ have employed the so-called \textit{quantum Generative Adversarial Networks} (qGANs) for learning and loading of probability distributions such as the uniform, normal, or log-normal distributions, including their multivariate versions \cite{QC16_GAN}. For these distributions, it has been shown empirically that the qGANs can well approximate the truncated and discretized distributions, and the gate complexity of the qGANs circuits scale only polynomially in the number of input qubits. Moreover, recently, \cite{fuchs2023hybrid} introduced \textit{Wasserstein}  qGAN (qWGAN) employing the Wasserstein loss, and numerically demonstrated the advantage over the standard qGAN methodology.
 In \cite{chakrabarti2021threshold}, Chakrabarti et al.\ constructed a quantum circuit for uploading the discretized multivariate log-normal distributions, where they used the Variational Quantum Eigensolvers (VQE) approach \cite{peruzzo2014variational} to upload quantum circuits for approximating the cumulative log-return process $R_t^{i}$, defined in \eqref{eqn: log-transform process}. It was estimated that loading the discretized multivariate log-normal distribution requires $O(Ld^2n\log_2(\eps^{-1}))$ gates, where $L\in \N$ is the depth of each variational quantum circuit for approximating the Gaussian distribution 
and $n$ is the number of qubits used in each quantum circuit for approximating the Gaussian distribution, see \cite[Appendix E]{chakrabarti2021threshold}. Justified by the above examples in the literature, we make the following assumption. 
\begin{assumption}[Loading of discretized multivariate log-normal distribution]\label{assumption: distribution loading}
Let $n,m \in \N$, and let $T > 0$. For every $d\in \N$ and $(t,\bm{x}) \in [0,T) \times \R_+^d$ let $p_d(\cdot,T;\bm{x},t):\R_+^d \to \R_+$ be the log-normal transition density given by \eqref{eqn: density formula}. Then, we assume that there exists a constant $C_3 \in [1,\infty)$ such that for every $d \in \N$ and $\varepsilon>0$ there exists a quantum circuit $\mathcal{P}_{d,\varepsilon}$ on $d(n+m)$ qubits such that the number of elementary gates used to construct $\mathcal{P}_{d,\varepsilon}$ is at most  
\begin{equation}\label{eqn: bound on P_d,eps}
C_3 d^{C_3}(n+m)^{C_3}(\log_2(\eps^{-1}))^{C_3}
\end{equation}
and that $\mathcal{P}_{d,\varepsilon}$ satisfies
\begin{equation}
\mathcal{P}_{d,\varepsilon}\ket{0}_{d(n+m)} = \sum_{\bm{i}=(i_1,\ldots,i_d) \in \mathbb{F}_{n,m,+}^d} \sqrt{\widetilde{p}_{\bm{i}}}\ket{i_1}_{n+m}\cdots\ket{i_d}_{n+m},
\end{equation}
with coefficients 
$\widetilde{p}_{\bm{i}}\in [0,1]$ satisfying
\begin{equation}
	\sum_{\bm{i} \in \mathbb{F}_{n,m,+}^d} \widetilde{p}_{\bm{i}}=1
\end{equation}
and
\begin{equation}
\sum_{\bm{i} \in \mathbb{F}_{n,m,+}^d}\big\vert \widetilde{p}_{\bm{i}} - \gamma^{-1}  p_{\bm{i},m}\big\vert\leq \varepsilon,
\end{equation}
where 
\begin{equation}
	p_{\bm{i},m} := \int_{Q_{\bm{i},m}} p_d(\bm{y},T;\bm{x},t)\,d\bm{y}, \quad Q_{\bm{i},m} := [\mathrm{D}_{n,m}(i_1),\mathrm{D}_{n,m}(i_1) + 2^{-m})\times \cdots \times [\mathrm{D}_{n,m}(i_d),\mathrm{D}_{n,m}(i_d)+2^{-m}),
\end{equation}
and
\begin{equation}
	\gamma := \sum_{\bm{i} \in \mathbb{F}_{n,m,+}^d} p_{\bm{i},m} \in (0,1)
\end{equation}
is a normalization constant.
\end{assumption} 
\begin{remark}\label{rem:exampleP}
	In case one uses the quantum circuit constructed in \cite[Appendix~E]{chakrabarti2021threshold} to upload the discretized multivariate log-normal distribution, the corresponding constant $C_3$ defined in Assumption~\ref{assumption: distribution loading} can be chosen to be $C_3:=\max\{2,L\}$, where $L\in \N$ is the depth of each variational quantum circuit involved in \cite[Appendix~E]{chakrabarti2021threshold} for approximating the involved Gaussian distributions; we refer to  \cite[Appendix~E]{chakrabarti2021threshold} for the precise construction of their quantum circuit. 
\end{remark}

\begin{remark}
For general probability distributions, loading of its discretized probability density function (PDF) remains one of the main problems in quantum computing. In the quantum computing literature, this step is also referred to as quantum state preparation, and it is an important initialization step for many quantum algorithms for pricing options. Recently, there has been new approaches to the quantum state preparation problem in the literature, that are not related to the qGAN or VQE methods reviewed above. In \cite{iaconis2023quantum}, the authors considered the quantum state preparation problem for probability distribution with smooth differentiable density functions, such as the normal distribution, where they proposed an algorithm based on the matrix product state (MPS) approximation method, and provide an error analysis and numerical convergence for the single-variate normal distribution. In \cite{pracht2023pricing}, the author proposed a quantum binomial tree algorithm to approximate the option prices in a discrete time setting. We refer the reader to \cite{chang2023novel} for a similar random walk based algorithm, and to  \cite{de2023quantum} for a hybrid classical quantum approach based on deconvolution methods for the quantum state preparation problem. However, to the best of our knowledge, there seems not to be any result in the literature that provides rigorous upper bounds on the quantum circuit complexities and as well as convergence for general multi-variate distributions.
\end{remark}

\subsection{Loading continuous piecewise affine payoff functions}\label{section: loading CPWA circuits}
The goal of this section is to upload (an approximation of) the payoff function $h:\R^d \to \R$ given in \eqref{eqn: CPWA payoff} to a quantum circuit. To that end, let $K \in \N$ be the number of component functions of the payoff function $h$ given in \eqref{eqn: CPWA payoff}, and for $k=1,\ldots,K$, let $h_k: [0,M]^d \to \R$ be (up to the sign) the corresponding $k$-th component of $h$ given by 
\begin{equation}
    h_k(\bm{x}) = \max\{\bm{a}_{k,l} \cdot \bm{x} + b_{k,l}: l = 1,\ldots,I_k\}, \label{eqn: normalized payoff}
\end{equation}
where $\bm{a}_{k,l}\in \R^d$, $b_{k,l} \in \R$ for $l =  1,\ldots,I_k$. The parameters $(\bm{a}_{k,l}, b_{k,l})$ are approximated by the two's complement method with binary strings of a suitable length. These binary strings are loaded on a qubit register using quantum circuits with $X$-gates, see Lemma \ref{lemma: quantum affine sum}. Using the arithmetic quantum circuits that we have constructed in the Section \ref{section: quantum arithmetic circuits}, we  construct a quantum circuit which computes the two's complement-discretized version of the payoff function $h_k(x)$. The discrete payoff function is then loaded by a controlled $Y$-rotation circuit, see Lemma \ref{lemma: quantum linear rotation} and Proposition \ref{prop: loading payoff circuit}. 
	We have also included quantum circuit diagrams in this section for the ease of understanding of the involved quantum circuits.

\begin{figure}[h!]
	\centering 
	\boxed{
\begin{quantikz} [row sep={20pt,between origins},column sep=20pt,font=\small]
\lstick{$\ket{i_1}_{n_1+m_1}$}	&\qw &\gate[6,nwires=3]{\mathcal{T}_\pi}  &\gate[6,nwires=3]{\bigotimes_{k=1}^d \mathcal{Q}_{(\times)}^{(k)}} &\gate[8,nwires=3]{\mathcal{T}_\pi} &\qw &\qw  &\qw\rstick{$\ket{i_1}_{n_1+m_1}$} \\
\lstick{$\ket{i_2}_{n_1+m_1}$} &\qw &\qw &\qw &\qw &\qw&\qw&\qw\rstick{$\ket{i_2}_{n_1+m_1}$}\\
\lstick{$\vdots$\qquad }  \\
\lstick{$\ket{i_d}_{n_1+m_1}$} &\qw &\qw & \qw &\qw &\qw&\qw&\qw \rstick{$\ket{i_d}_{n_1+m_1}$}\\
\lstick{$\ket{0}_{(d+1)(n_2+m_2)}$} &\gate[wires=3]{{\mathcal{X}}_{\bm{a},b}} &\qw &\qw &\qw &\gate[4]{\mathcal{Q}_{(+)} } &\gate[4]{\mathcal{T}_{\pi} }&\qw \rstick[3]{$\ket{\big(\boxplus_{l=1}^d a_l \boxdot i_l\big)\boxplus b}_{n+m+d} $}\\
\lstick{$\ket{0}_{d(n+m)}$} &\linethrough &\qw &\qw&\qw&\qw&\qw&\qw\\
\lstick{$\ket{0}_{d(n_2+m_2+3)}$} &\qw &\qw &\qw&\qw&\qw&\qw&\qw\\
\lstick{$\ket{0}_{d}$} &\qw &\qw &\qw&\qw&\qw&\qw&\qw\rstick{$\ket{\mathrm{anc}}$}
\end{quantikz}
} 
	\caption{Circuit diagram for $\mathcal{Q}_{+}^{d,n,m}$ in Lemma~\ref{lemma: quantum affine sum}.} 
	\label{fig: affine sum circuit diagram}
\end{figure}
\begin{lemma}[Quantum circuit for affine sums]\label{lemma: quantum affine sum}
Let $d,n_1,n_2 \in \N$, $m_1,m_2 \in \N_0$. Let $n := n_1+n_2$, and $m := m_1 + m_2$. Let $a_1,\ldots,a_d,b \in \mathbb{F}_{n_2,m_2}$. Then, there is a quantum circuit $\mathcal{Q}_{+}^{d,n,m}$ on $N$ qubits, where
\begin{equation}\label{eqn: lemma affine def N}
    N := d(n_1+m_1) + (d+1)(n_2+m_2) + d(n+m) + d(n_2+m_2+3) + d 
\end{equation}
such that for any $i_1,\ldots,i_d \in \mathbb{F}_{n_1,m_1}$,
\begin{equation}
\begin{split}
    \mathcal{Q}_{+}^{d,n,m}:&\ket{i_1}_{n_1+m_1}\cdots\ket{i_d}_{n_1+m_1}\ket{0}_{(d+1)(n_2+m_2)} \ket{0}_{ d(n+m)}\ket{0}_{d(n_2+m_2+3)}\ket{0}_{d} \\
    &\mapsto \ket{i_1}_{n_1+m_1}\cdots\ket{i_d}_{n_1+m_1}\ket{(\bigboxplus_{k=1}^d( a_k \boxdot i_k)) \boxplus b}_{n+m+d}\ket{\anc}_{p}, \label{eqn: affine sum}
\end{split}
\end{equation}
with $p := d(2n_2+2m_2+3) + (n_2+m_2)+ (d-1)(n+m)$, and where $(\displaystyle\boxplus_{k=1}^d( a_k \boxdot i_k)) \boxplus b \in \mathbb{F}_{d+n,m}$ is the two's complement binary string representing the affine sum 
\begin{equation}
\left(    \sum_{k=1}^d (\mathrm{D}_{n_2,m_2} (a_k) \cdot \mathrm{D}_{n_1,m_1}(i_k)) \right)+ \mathrm{D}_{n_2,m_2}(b) \in \mathbb{K}_{n+d,m},
\end{equation}
(c.f. Definition \ref{def: encoder-decoder}). The quantum circuit $\mathcal{Q}_{+}^{d,n,m}$ uses at most $563d^3(n+m+1)^2$ elementary gates. See Figure~\ref{fig: affine sum circuit diagram} for the circuit diagram.
\end{lemma}
\begin{proof}
The construction of this circuit involves the following steps:
\begin{enumerate}
    \item We first load the given two's complement binary strings $a_1,\ldots,a_d,b \in \mathbb{F}_{n_2,m_2}$ on the qubit register $\ket{0}_{(d+1)(n_2+m_2)} = \ket{0}_{n_2+m_2}\cdots\ket{0}_{n_2+m_2}$. To that end, we use the Pauli $X$ gate 
    to flip the bit $0$ to $1$ according the binary strings $a_1,\ldots,a_d,b$ if necessary, to obtain the state
\begin{equation}
    \mathcal{X}_{\bm{a},b}: \ket{0}_{n_2+m_2}\cdots\ket{0}_{n_2+m_2} \mapsto \ket{a_1}_{n_2+m_2}\cdots\ket{a_d}_{n_2+m_2}\ket{b}_{n_2+m_2},
\end{equation}
where we define 
\begin{equation}
    \mathcal{X}_{\bm{a},b} := \left (\bigotimes_{k=1}^d \bigotimes_{l=-m_2}^{n_2-1} X^{a_k(l)}\right ) \otimes \left(\bigotimes_{l=-m_2}^{n_2-1} X^{b(l)}\right),
\end{equation}
given the binary strings 
$a_k = ((a_k(l))_{l=0}^{n_2-1},(a_k(l))_{l=-m_2}^{-1}) 
=((a_k(n_2-1),\dots,a_k(0)),(a_k(-1),\dots,a_k(-m_2)))$
$\in \mathbb{F}_{n_2,m_2}$
 and  $b=((b(n_2-1),\dots,b(0)),(b(-1),\dots,b(-m_2)))\in \mathbb{F}_{n_2,m_2}$. Note that we use the convention of $X^0 = I_2$ for any unitary matrix $X$. We define the quantum circuit
\begin{equation}
    \widetilde{\mathcal{X}}_{\bm{a},b} := I^{\otimes d(n_1+m_1)}_2 \otimes \mathcal{X}_{\bm{a},b} \otimes I^{\otimes (d(n+m)+d(n_2+m_2+3)+d)}_2.
\end{equation} 
Hence, for any $i_1,\ldots,i_d \in \mathbb{F}_{n_1,m_1}$, we have 
\begin{equation}
    \begin{split}
        \widetilde{\mathcal{X}}_{\bm{a},b}:&\ket{i_1}_{n_1+m_1}\cdots\ket{i_d}_{n_1+m_1}\ket{0}_{(d+1)(n_2+m_2)} \ket{0}_{ d(n+m)}\ket{0}_{d(n_2+m_2+3)}\ket{0}_{d}\\    &\mapsto\ket{i_1}_{n_1+m_1}\cdots\ket{i_d}_{n_1+m_1}\ket{a_1}_{n_2+m_2}\cdots\ket{a_d}_{n_2+m_2}\ket{b}_{n_2+m_2}\ket{0}_{ d(n+m)}\ket{0}_{d(n_2+m_2+3)}\ket{0}_{d},
    \end{split}
\end{equation}
and the number of Pauli $X$ gates used to construct the quantum circuit $\widetilde{\mathcal{X}}_{\bm{a},b}$ is at most
\begin{equation}
    (d+1)(n_2+m_2).
\end{equation}

    \item Next, we apply the permutation quantum circuit $\Tau_{\pi}$ from Lemma \ref{lemma: permutation circuit} to prepare for the upcoming $d$ multiplications so that
    \begin{equation}\begin{split}
            \Tau_{\pi}:&\ket{i_1}_{n_1+m_1}\cdots\ket{i_d}_{n_1+m_1}\ket{a_1}_{n_2+m_2}\cdots\ket{a_d}_{n_2+m_2}\ket{b}_{n_2+m_2}\ket{0}_{ d(n+m)}\ket{0}_{d(n_2+m_2+3)}\ket{0}_{d} \\
            &\mapsto \bigotimes_{k=1}^d \big(\ket{i_k}_{n_1+m_1}\ket{a_k}_{n_2+m_2}\ket{0}_{n+m}\ket{0}_{n_2+m_2+3}\big) \otimes \ket{b}_{n_2+m_2}\ket{0}_d\\
            &=  \ket{i_1}_{n_1+m_1}\ket{a_1}_{n_2+m_2}\ket{0}_{n+m}\ket{0}_{n_2+m_2+3}\cdots \ket{i_d}_{n_1+m_1}\ket{a_d}_{n_2+m_2}\ket{0}_{n+m}\ket{0}_{n_2+m_2+3}\ket{b}_{n_2+m_2}\ket{0}_{d}.
        \end{split} 
    \end{equation}
    The number of swap gates used to construct $\Tau_{\pi}$ in this step is at most $2N^2$.
    
    \item Next, for each $k=1,\ldots,d$, we apply the multiplication quantum circuit $\mathcal{Q}_{(\times)}^{(k)} := \widetilde{\mathcal{Q}}_{(\times)}$ from Corollary \ref{corollary: mult circuit} (with $ n_1 \leftarrow n_1, n_2 \leftarrow n_2, m_1 \leftarrow m_1, m_2 \leftarrow m_2, a \leftarrow i_k, b \leftarrow a_k$ in the notation of Corollary \ref{corollary: mult circuit}) on each component $(\ket{i_k}_{n_1+m_1}\ket{a_k}_{n_2+m_2}\ket{0}_{n+m}\ket{0}_{n_2+m_2+3})$ such that  
\begin{equation}\begin{split}
    \bigotimes_{k=1}^d \mathcal{Q}_{(\times)}^{(k)}:&\bigotimes_{k=1}^d \big(\ket{i_k}_{n_1+m_1}\ket{a_k}_{n_2+m_2}\ket{0}_{n+m}\ket{0}_{n_2+m_2+3}\big) \otimes \ket{b}_{n_2+m_2}\ket{0}_d\\
    &\mapsto \bigotimes_{k=1}^d \big(\ket{i_k}_{n_1+m_1}\ket{a_k}_{n_2+m_2}\ket{a_k \boxdot i_k}_{n+m}\ket{\anc}_{n_2+m_2+3}\big) \otimes \ket{b}_{n_2+m_2}\ket{0}_d\\
    &= \ket{i_1}_{n_1+m_1}\ket{a_1}_{n_2+m_2}\ket{a_1 \boxdot i_1}_{n+m}\ket{\anc}_{n_2+m_2+3}\cdots\\
    &\quad \cdots\ket{i_d}_{n_1+m_1}\ket{a_d}_{n_2+m_2}\ket{a_d \boxdot i_d}_{n+m}\ket{\anc}_{n_2+m_2+3} \ket{b}_{n_2+m_2}\ket{0}_{d}.
    \end{split}
\end{equation}
    The number of elementary gates used in this step is at most 
    \begin{equation}
        d \cdot 61(n + m + 1)^2.
    \end{equation}
    
    \item We apply the permutation quantum circuit $\Tau_{\pi}$ from Lemma \ref{lemma: permutation circuit} to prepare for the upcoming $d$ additions so that 
    \begin{equation}
        \begin{split}
            \Tau_{\pi}:&\bigotimes_{k=1}^d \big(\ket{i_k}_{n_1+m_1}\ket{a_k}_{n_2+m_2}\ket{a_k \boxdot i_k}_{n+m}\ket{\anc}_{n_2+m_2+3}\big) \otimes \ket{b}_{n_2+m_2}\ket{0}_d\\
            &\mapsto \ket{i_1}_{n_1+m_1}\cdots \ket{i_d}_{n_1+m_1}\ket{a_1 \boxdot i_1}_{n+m}\ket{a_2\boxdot i_2}_{n+m}\ket{0}\ket{a_3\boxdot i_3}_{n+m}\ket{0}\cdots \\
            &\quad \cdots \ket{a_d\boxdot i_d}_{n+m}\ket{0} \ket{b}_{n_2+m_2}\ket{0} \ket{\anc}_{d(2n_2+2m_2+3)}.
        \end{split}
    \end{equation}
Here, we consolidate the qubits $\ket{a_1}_{n_2+m_2},\ldots,\ket{a_d}_{n_2+m_2}$ in the ancilla qubit placeholder $\ket{\anc}_{d(2n_2+2m_2+3)}$ as we do not need them in the later computations. The number of elementary gates used for this step is at most $2N^2$.

    \item We perform the following addition inductively on the sums for $k=1,\ldots,d-1$
    \begin{equation}
        \boxplus: \mathbb{F}_{n+k-1,m} \times \mathbb{F}_{n,m} \to \mathbb{F}_{n+k,m}, \quad 
        \left( \bigboxplus_{l=1}^k (a_l \boxdot i_l),  (a_{k+1} \boxdot i_{k+1}) \right) \mapsto \bigboxplus_{l=1}^{k+1} (a_l \boxdot i_l),
    \end{equation}
    and the addition 
    \begin{equation}
        \boxplus: \mathbb{F}_{n+d-1,m} \times \mathbb{F}_{n_2,m_2} \to \mathbb{F}_{n+d,m}, \quad 
        \left(\bigboxplus_{l=1}^{d} (a_l \boxdot i_l), b \right)\mapsto (\bigboxplus_{l=1}^{d} (a_l \boxdot i_l)) \boxplus b,
    \end{equation}
    (c.f. Lemma \ref{lemma: addition in TC} for definition of $\boxplus$). That is, we apply the quantum circuit $\mathcal{Q}_{(+)}^{(k)} := \mathcal{Q}_{(+)}$ from Corollary \ref{corollary: add circuit} inductively for $k=1,\ldots,d-1$ (with $n_1 \leftarrow n+k-1$, $n_2 \leftarrow n$, $m_1 \leftarrow m$, $m_2 \leftarrow m$, $a \leftarrow \displaystyle \bigboxplus_{l=1}^k (a_l \boxdot i_l)$, $b \leftarrow a_{k+1} \boxdot i_{k+1}$ in the notation of Corollary \ref{corollary: add circuit}), and we apply the quantum circuit $\mathcal{Q}_{(+)}^{(b)}  := \mathcal{Q}_{(+)}$ (with $n_1 \leftarrow n+d-1$, $n_2 \leftarrow n_2$, $m_1 \leftarrow m$, $m_2 \leftarrow m_2$, $a \leftarrow \displaystyle \bigboxplus_{l=1}^{d} (a_l \boxdot i_l)$, $b \leftarrow b$ in the notation of Corollary \ref{corollary: add circuit})  so that 
    \begin{equation}
        \begin{split}
            &\ket{i_1}_{n_1+m_1}\cdots \ket{i_d}_{n_1+m_1}\ket{a_1 \boxdot i_1}_{n+m}\ket{a_2\boxdot i_2}_{n+m}\ket{0}\ket{a_3\boxdot i_3}_{n+m}\ket{0}\cdots \\
            &\quad \cdots \ket{a_d\boxdot i_d}_{n+m}\ket{0} \ket{b}_{n_2+m_2}\ket{0} \ket{\anc}_{d(2n_2+2m_2+3)}\\
            &\xmapsto{\mathcal{Q}_{(+)}^{(1)}}\ket{i_1}_{n_1+m_1}\cdots \ket{i_d}_{n_1+m_1}\ket{a_2 \boxdot i_2}_{n+m}\ket{\bigboxplus_{l=1}^2 a_l\boxdot i_l}_{n+m+1}\ket{a_3\boxdot i_3}_{n+m}\ket{0}\cdots \\
            &\quad \cdots \ket{a_d\boxdot i_d}_{n+m}\ket{0} \ket{b}_{n_2+m_2}\ket{0} \ket{\anc}_{d(2n_2+2m_2+3)}\\
            &\xmapsto{\mathcal{Q}_{(+)}^{(2)}}\ket{i_1}_{n_1+m_1}\cdots \ket{i_d}_{n_1+m_1}\ket{a_2 \boxdot i_2}_{n+m}\ket{a_3\boxdot i_3}_{n+m}\ket{\bigboxplus_{l=1}^3 a_l\boxdot i_l}_{n+m+2}\cdots \\
            &\quad \cdots \ket{a_d\boxdot i_d}_{n+m}\ket{0} \ket{b}_{n_2+m_2}\ket{0} \ket{\anc}_{d(2n_2+2m_2+3)}\\
            &\qquad\qquad\qquad\qquad\vdots\qquad\qquad\qquad\qquad\vdots \\
            &\xmapsto{\mathcal{Q}_{(+)}^{(d-1)}}\ket{i_1}_{n_1+m_1}\cdots \ket{i_d}_{n_1+m_1}\ket{a_2 \boxdot i_2}_{n+m}\ket{a_3 \boxdot i_3}_{n+m} \cdots  \\
            &\quad \cdots \ket{a_{d} \boxdot i_{d}}_{n+m}\ket{\bigboxplus_{l=1}^{d} a_l\boxdot i_l}_{n+m+d-1} \ket{b}_{n_2+m_2}\ket{0} \ket{\anc}_{d(2n_2+2m_2+3)}\\
            &\xmapsto{\mathcal{Q}_{(+)}^{(b)}}\ket{i_1}_{n_1+m_1}\cdots \ket{i_d}_{n_1+m_1}\ket{a_2 \boxdot i_2}_{n+m}\ket{a_3 \boxdot i_3}_{n+m} \cdots \\
            &\quad \cdots\ket{a_{d} \boxdot i_{d}}_{n+m}\ket{b}_{n_2+m_2} \ket{(\bigboxplus_{l=1}^{d} a_l\boxdot i_l)\boxplus b}_{n+m+d} \ket{\anc}_{d(2n_2+2m_2+3)}.
        \end{split}
    \end{equation}
    The number of elementary gates used for this step is at most
    \begin{equation}
        \begin{split}
            &\sum_{k=1}^{d-1} 29[(n+m+k-1) + (n+m)+1]^2 + 29[(n+m+d-1)+(n_2+m_2)+1]^2\\
            &\qquad \leq 29d \cdot 4d^2(n+m+1)^2\\
            &\qquad = 116d^3 (n+m+1)^2,
        \end{split}
    \end{equation}
        where we use the fact that  $[(n+m+k-1) + (n+m)+1]^2\leq (2n+2m+d)^2\leq 4d^2 (n+m+1)^2$ when $k \leq d$. 
    \item We consolidate the ancillary qubits by combining the qubits (labeled $\ket{a_2\boxdot i_2},\ldots,\ket{a_d \boxdot i_d},\ket{b})$ under ancilla qubits $\ket{\anc}_\star$. The permutation circuit $\Tau_\pi$ from Lemma \ref{lemma: permutation circuit} performs the following operation
    \begin{equation}
        \begin{split}
           \Tau_{\pi}:&\ket{i_1}_{n_1+m_1}\cdots \ket{i_d}_{n_1+m_1}\ket{a_2 \boxdot i_2}_{n+m} \cdots \ket{a_{d} \boxdot i_{d}}_{n+m}\ket{b}_{n_2+m_2} \ket{(\bigboxplus_{l=1}^{d} a_l\boxdot i_l)\boxplus b}_{n+m+d} \ket{\anc}_{d(2n_2+2m_2+3)}\\
           & \mapsto \ket{i_1}_{n_1+m_1}\cdots\ket{i_d}_{n_1+m_1}\ket{(\bigboxplus_{l=1}^{d} a_l\boxdot i_l)\boxplus b}_{n+m+d} \ket{\anc}_{d(2n_2+2m_2+3) + (n_2+m_2)+ (d-1)(n+m)}.
        \end{split}
    \end{equation}
The number of elementary gates used for this step is at most $2N^2$.
\end{enumerate}
The resulting quantum circuit $\mathcal{Q}_{+}^{d,n,m}$ is a composition of the quantum circuits from each of the above steps. To deduce its gate complexity, we sum up the number of elementary gates used in each step. We note from the definition of $n,m \in \N$ and the definition of $N$ in \eqref{eqn: lemma affine def N} that
\begin{equation}
    \begin{split}
         N &:= d(n_1+m_1) + (d+1)(n_2+m_2) + d(n+m) + d(n_2+m_2+3) + d\\
         &\leq [d + 2d + d + 3d + d](n+m+1)\\
         &= 8d(n+m+1).
    \end{split}
\end{equation}
Summing the number of elementary gates used in each step, we find that the number of elementary gates used in total is at most
\begin{equation}
    \begin{split}
        &(d+1)(n_2+m_2) + 2N^2 + 61d(n+m+1)^2 + 2N^2 + 116d^3(n+m+1)^2+ 2N^2\\
        &\leq 2d(n+m+1) + 3\cdot 2(8d(n+m+1))^2 + 61d(n+m+1)^2 + 116d^3(n+m+1)^2\\
        &\leq (2+ 3\cdot2\cdot8^2 + 61 + 116) d^3(n+m+1)^2\\
        &= 563d^3(n+m+1)^2.
    \end{split}
\end{equation}
\end{proof}

\begin{figure}[t]
	\centering 
	\boxed{
		\begin{quantikz}[row sep={20pt,between origins},column sep=20pt,font=\small]
	\lstick{$\ket{a}_{n+m}$} & \gate[3]{\widetilde{\mathcal{Q}}_{(\mathrm{comp})}} & \gate[4]{\mathcal{X}_{c_1,a}} & \gate[4]{\mathcal{X}_{c_2,b}} & \gate[4]{\mathcal{X}_{c_3,b}}&\qw \rstick{$\ket{a}_{n+m}$} \\
	\lstick{$\ket{b}_{n+m}$} &\qw&\qw&\qw&\qw&\qw  \rstick{$\ket{b}_{n+m}$}\\
	\lstick{$\ket{0}_4$} &\qw&\qw&\qw&\qw&\qw  \rstick{$\ket{\mathrm{anc}}_{4}$}\\
	\lstick{$\ket{0}_{n+m}$} &\qw&\qw&\qw&\qw&\qw \rstick{$\ket{\mathrm{M}_{n,m}(a,b)}_{n+m}$} 
		\end{quantikz}		
	} 
	\caption{Circuit diagram for $\mathcal{Q}_{(\max)}^{n,m}$ in Lemma~\ref{lemma: circuit maximum 2}.} 
	\label{fig: max circuit diagram}
\end{figure}
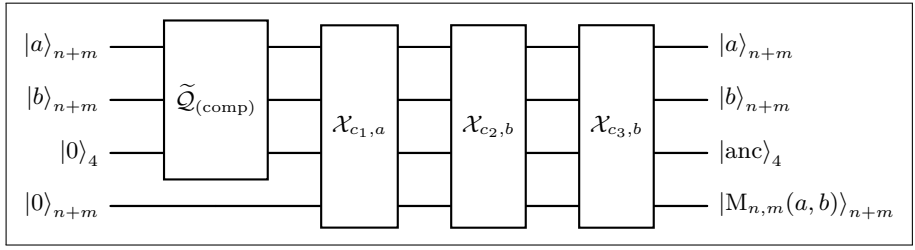

\begin{lemma}[Quantum circuit for maximum of two numbers]\label{lemma: circuit maximum 2}
    Let $n,m \in \N$, and let $\mathrm{M}_{n,m}:\mathbb{F}_{n,m} \times \mathbb{F}_{n,m}  \to \mathbb{F}_{n,m}$ be a function defined by 
    \begin{equation}
        \mathrm{M}_{n,m}(a,b) := \mathrm{E}_{n,m}(\max\{\mathrm{D}_{n,m}(a),\mathrm{D}_{n,m}(b)\}), \quad \forall a,b \in \mathbb{F}_{n,m}.
    \end{equation}
(c.f. Definition \ref{def: encoder-decoder}). Then, there is a quantum circuit $\mathcal{Q}_{(\max)}^{n,m}$ on $3(n+m) + 5$ qubits such that for any $a,b \in \mathbb{F}_{n,m}$,
\begin{equation}\label{eqn: max2 circuit}
    \begin{split}
        \mathcal{Q}_{(\max)}^{n,m}:&\ket{a}_{n+m}\ket{b}_{n+m}\ket{0}_{5}\ket{0}_{n+m}  \mapsto \ket{a}_{n+m}\ket{b}_{n+m}\ket{\anc}_{5}\ket{\mathrm{M}_{n,m}(a,b)}_{n+m},
    \end{split}
\end{equation}
which uses at most $1045(n+m+1)^3$ elementary gates. See Figure~\ref{fig: max circuit diagram} for the circuit diagram.
\begin{proof} The construction of this quantum circuit involves the following steps:
    \begin{enumerate}
        \item We use the comparison quantum circuit $\widetilde{\mathcal{Q}}_{(\mathrm{comp})}$ in Corollary \ref{corollary: comp circuit} (with $n_1 \leftarrow n ,n_2 \leftarrow n$, $m_1\leftarrow m,m_2 \leftarrow m$, $a \leftarrow a$ and $b \leftarrow b$ in the notation of Corollary \ref{corollary: comp circuit}) to obtain 
        \begin{equation}\label{eqn: max2 circuit1}
          \widetilde{\mathcal{Q}}_{(\mathrm{comp})} \otimes I^{\otimes n+m}_2: \ket{a}_{n+m}\ket{b}_{n+m}\ket{0}_5\ket{0}_{n+m} \mapsto \ket{a}_{n+m}\ket{b}_{n+m}\ket{c_1}\ket{c_2}\ket{c_3}\ket{\anc}_2\ket{0}_{n+m},
        \end{equation}
    where
    \begin{equation}\label{eqn: cases of c1-3}
        \ket{c_1}\ket{c_2}\ket{c_3} = \begin{cases}
            \ket{1}\ket{0}\ket{0}, \quad \text{if $\mathrm{D}_{n,m}(a)>\mathrm{D}_{n,m}(b)$}, \\ 
            \ket{0}\ket{1}\ket{0}, \quad \text{if $\mathrm{D}_{n,m}(a)<\mathrm{D}_{n,m}(b)$},\\
            \ket{0}\ket{0}\ket{1}, \quad \text{if $\mathrm{D}_{n,m}(a)=\mathrm{D}_{n,m}(b)$}.
        \end{cases}
    \end{equation}
    The number of gates used in this step is at most $175(2n+2m+1)^2$ elementary gates.
    \item Recall the Toffoli gate (also referred as the $\mathrm{CCNOT}$ gate) acting on three qubits such that for $a,b,c \in \{0,1\}$ 
    \begin{equation}\label{CCNOT-def}
    	\mathrm{CCNOT}: \ket{a}\ket{b}\ket{c} \mapsto \ket{a}\ket{b}\ket{c + a b \mod{2}}.
    \end{equation}
    Here one refers $\ket{a}$ as the control qubit, $\ket{b}$ the target qubit, and $\ket{c}$ the output qubit. 
    
    We apply $(n+m)$ Toffoli gates
     on the control qubit $\ket{c_1}$,  target qubits $\ket{a}_{n+m}$, and output qubits $\ket{0}_{n+m}$. More precisely, for $a = ((a_{n-1},\ldots,a_0),(a_{-1},\ldots,a_{-m})) \in \mathbb{F}_{n,m}$, we apply for each $j=-m,\ldots,n-1$ a Toffoli gate $\mathcal{C}_{c_1,a_j} =  \mathrm{CCNOT}$ (with $a \leftarrow c_1$, $b\leftarrow a_j$, and $c \leftarrow 0$ in the notation of \eqref{CCNOT-def})
    \begin{equation}
        \mathcal{C}_{c_1,a_j}: \ket{c_1}\ket{a_j}\ket{0} \mapsto \ket{c_1}\ket{a_j}X^{c_1 a_j}\ket{0} = \begin{cases}
            \ket{c_1}\ket{a_j}\ket{1}, &\quad \text{if $c_1=1$ and $a_j=1$},\\
            \ket{c_1}\ket{a_j}\ket{0}, &\quad \text{otherwise}.
        \end{cases}
    \end{equation}
    For each $j=-m,\ldots,n-1$, we apply the permutation quantum circuit $\mathcal{\Tau}_j$ before and after each Toffoli gate $\mathcal{C}_{c_1,a_j}$. We define  the quantum circuit $\mathcal{X}_{c_1,j}$  by
    \begin{equation}
        \mathcal{X}_{c_1,j} := \Tau_j \left(I^{\otimes n+m-1}_2\otimes I^{\otimes n+m}_2\otimes I^{\otimes 4}_2\otimes I^{\otimes n - j - 1}_2\otimes \mathcal{C}_{c_1,a_j} \otimes I^{\otimes m+j}_2 \right)\Tau_j,
    \end{equation}
    which computes from \eqref{eqn: max2 circuit1} the following
    \begin{equation}
        \begin{split}
            &\ket{a}_{n+m}\ket{b}_{n+m}\ket{c_1}\ket{c_2}\ket{c_3}\ket{\anc}_2\ket{0}_{n+m}\\
            &\xmapsto{\Tau_{j}} \ket{\hat{a}^{j}}_{n+m-1}\ket{b}_{n+m}\ket{c_2}\ket{c_3}\ket{\anc}_2\ket{0}_{n-j-1}\ket{c_1}\ket{a_{j}}\ket{0}\ket{0}_{m+j}\\
            &\xmapsto{\mathcal{C}_{c_1,a_j}}\ket{\hat{a}^{j}}_{n+m-1}\ket{b}_{n+m}\ket{c_2}\ket{c_3}\ket{\anc}_2\ket{0}_{n-j-1}\ket{c_1}\ket{a_{j}}X^{c_1 a_{j}}\ket{0}\ket{0}_{m+j}\\
            &\xmapsto{\Tau_{j}} \ket{a}_{n+m}\ket{b}_{n+m}\ket{c_1}\ket{c_2}\ket{c_3}\ket{\anc}_2\ket{0}_{n-j-1}X^{c_1 a_{j}}\ket{0}\ket{0}_{m+j},
        \end{split}
    \end{equation}
    where $\ket{\hat{a}^j}_{n+m-1} := \ket{a_{n-1}}\cdots\ket{a_{j+1}}\ket{a_{j-1}}\cdots\ket{a_{-m}}$. Finally, we define the quantum circuit $\mathcal{X}_{c_1,a}$ by 
    \begin{equation}
        \mathcal{X}_{c_1,a} := \prod_{j=-m}^{n-1} \mathcal{X}_{c_1,j},
    \end{equation}
    and we compute that 
    \begin{equation}
        \begin{split}
            &\mathcal{X}_{c_1,a}:\ket{a}_{n+m}\ket{b}_{n+m}\ket{c_1}\ket{c_2}\ket{c_3}\ket{\anc}_2\ket{0}_{n+m}\\
            &\qquad \mapsto \ket{a}_{n+m}\ket{b}_{n+m}\ket{c_1}\ket{c_2}\ket{c_3}\ket{\anc}_2\left(\bigotimes_{j=-m}^{n-1}X^{c_1a_j}\ket{0}_{n+m}\right).
        \end{split}
    \end{equation}
    Note that the number of elementary gates used in this step is at most $(n+m)[15 + 2\cdot2(3(n+m) + 5)^2]$, since each Toffoli gate circuit requires elementary $15$ gates (see, e.g., \cite[Section~6.3.5]{jacquier2022quantum}), and each permutation circuit uses $2(3(n+m) + 5)^2$ elementary gates. 
    \item We repeat step 2 but by instead using the control qubit $\ket{c_2}$ with target qubits $\ket{b}_{n+m}$. We define the quantum circuits $\mathcal{X}_{c_2,b}$ similarly, and we compute that 
\begin{equation}
    \begin{split}
        &\mathcal{X}_{c_2,b}:\ket{a}_{n+m}\ket{b}_{n+m}\ket{c_1}\ket{c_2}\ket{c_3}\ket{\anc}_2\left(\bigotimes_{j=-m}^{n-1}X^{c_1a_j}\ket{0}_{n+m}\right)\\
        &\qquad \mapsto \ket{a}_{n+m}\ket{b}_{n+m}\ket{c_1}\ket{c_2}\ket{c_3}\ket{\anc}_2\left(\bigotimes_{j=-m}^{n-1}X^{c_2b_j}X^{c_1a_j}\ket{0}_{n+m}\right).
    \end{split}
\end{equation}
    \item We repeat step 2 but by instead using the control qubit $\ket{c_3}$ with target qubits $\ket{b}_{n+m}$. We define the quantum circuits $\mathcal{X}_{c_3,b}$ similarly, and we compute that 
    \begin{equation}
        \begin{split}
            &\mathcal{X}_{c_3,b}:\ket{a}_{n+m}\ket{b}_{n+m}\ket{c_1}\ket{c_2}\ket{c_3}\ket{\anc}_2\left(\bigotimes_{j=-m}^{n-1}X^{c_2b_j}X^{c_1a_j}\ket{0}_{n+m}\right)\\
            &\qquad \mapsto \ket{a}_{n+m}\ket{b}_{n+m}\ket{c_1}\ket{c_2}\ket{c_3}\ket{\anc}_2\left(\bigotimes_{j=-m}^{n-1}X^{c_3b_j}X^{c_2b_j}X^{c_1a_j}\ket{0}_{n+m}\right).
        \end{split}
    \end{equation}
    \end{enumerate}
    Since the qubits $\ket{c_1}\ket{c_2}\ket{c_3}$ may only take one of the three possible values as in \eqref{eqn: cases of c1-3}, it holds
     for each $j=-m,\ldots,n-1$ that
     \begin{equation}
         X^{c_3b_j}X^{c_2b_j}X^{c_1a_j} = \begin{cases} X^{a_j}, &\quad \text{if $\mathrm{D}_{n,m}(a)>\mathrm{D}_{n,m}(b)$},\\
             X^{b_j}, &\quad \text{if $\mathrm{D}_{n,m}(a)\leq\mathrm{D}_{n,m}(b)$}.
         \end{cases}
     \end{equation}
    Thus, we have
    \begin{equation}
        \begin{split}
            &\ket{a}_{n+m}\ket{b}_{n+m}\ket{c_1}\ket{c_2}\ket{c_3}\ket{\anc}_2\left(\bigotimes_{j=-m}^{n-1}X^{c_3b_j}X^{c_2b_j}X^{c_1a_j}\ket{0}_{n+m}\right)\\
            &= \begin{cases}
                \ket{a}_{n+m}\ket{b}_{n+m}\ket{c_1}\ket{c_2}\ket{c_3}\ket{\anc}_2\left(\bigotimes_{j=-m}^{n-1}X^{a_j}\ket{0}_{n+m}\right),&\quad \text{if $\mathrm{D}_{n,m}(a)>\mathrm{D}_{n,m}(b)$},\\
                \ket{a}_{n+m}\ket{b}_{n+m}\ket{c_1}\ket{c_2}\ket{c_3}\ket{\anc}_2\left(\bigotimes_{j=-m}^{n-1}X^{b_j}\ket{0}_{n+m}\right),&\quad \text{if $\mathrm{D}_{n,m}(a)\leq\mathrm{D}_{n,m}(b)$},\\
            \end{cases}\\
            &= \begin{cases}
                \ket{a}_{n+m}\ket{b}_{n+m}\ket{c_1}\ket{c_2}\ket{c_3}\ket{\anc}_2\ket{a}_{n+m}, \quad \mathrm{D}_{n,m}(a)>\mathrm{D}_{n,m}(b),  \\
                \ket{a}_{n+m}\ket{b}_{n+m}\ket{c_1}\ket{c_2}\ket{c_3}\ket{\anc}_2\ket{b}_{n+m}, \quad \mathrm{D}_{n,m}(a)\leq\mathrm{D}_{n,m}(b),
            \end{cases}\\
            &= \ket{a}_{n+m}\ket{b}_{n+m}\ket{c_1}\ket{c_2}\ket{c_3}\ket{\anc}_2\ket{\mathrm{M}_{n,m}(a,b)}_{n+m}.
        \end{split}
    \end{equation}
    The above output is equivalent \eqref{eqn: max2 circuit} where we treat $\ket{c_1}\ket{c_2}\ket{c_3}$ as ancilla qubits. We find that the total number of elementary gates used is at most 
    \begin{equation}
        \begin{split}
            &175(2n+2m+1)^2 + 3\cdot (n+m)[15 + 2\cdot2(3(n+m) + 5)^2] \\
            &\leq 175\cdot 2^2 (n+m+1)^2 + 3\cdot 15 (n+m) + 3(n+m) \cdot 2 \cdot 2 \cdot 5^2(n+m+1)^2\\
            &\leq (700 + 45 + 300)(n+m+1)^3\\
            &= 1045(n+m+1)^3.
        \end{split}
    \end{equation}
\end{proof}

\end{lemma}

\begin{figure}
	\centering 
	\boxed{
	\begin{quantikz} [row sep={20pt,between origins},column sep=10pt,font=\small]
\lstick{$\ket{i_1}_{n+m}$} & \gate[11]{\mathcal{T}_{\pi}}\\
\lstick{$\ket{i_2}_{n+m}$} &\qw\\
&\qw\\
& \\
& \\
\lstick{$\vdots$\qquad }& \\
& \\
& \\
& \\
\lstick{$\ket{i_I}_{n+m}$}&\qw\\
\lstick{$\ket{0}_{(I-1)(n+m+4)}$}&\qw
\end{quantikz}
\begin{quantikz} [row sep={20pt,between origins},column sep=10pt,font=\small]
\lstick{$\ket{i_1}_{n+m}$} &\gate[4]{\mathcal{Q}_{(\mathrm{max})}^{n,m}}  &\qw&\qw&\qw&\gate[11,nwires=7]{\mathcal{T}_{\pi}}&\qw \rstick{$\ket{i_1}_{n+m}$} \\
\lstick{$\ket{i_2}_{n+m}$}  &\qw&\qw&\qw&\qw&\qw&\qw \rstick{$\ket{i_2}_{n+m}$}  \\
\lstick{$\ket{0}_{4}$} &\qw&\qw&\qw&\qw&\qw&\qw \quad \vdots\\
\lstick{$\ket{0}_{n+m}$} &\qw&\gate[4]{\mathcal{Q}_{(\mathrm{max})}^{n,m}}&\qw&\qw&\qw&\qw \quad \vdots \\
\lstick{$\ket{i_3}_{n+m}$} &\qw&\qw&\qw&\qw&\qw&\qw\rstick{$\ket{i_I}_{n+m}$} \\
\lstick{$\ket{0}_{4}$} &\qw&\qw&\qw&\qw&\qw&\qw\rstick[5]{$\ket{\mathrm{anc}}_{(I-2)(n+m)+4(I-1)}$} \\
\lstick{$\ket{0}_{n+m}$} &\qw&\qw&\qw&\qw&\qw&\qw  \\
\lstick{$\vdots$\quad } &\qw&\qw&\qw \vdots\ & \gate[4,nwires=1]{\mathcal{Q}_{(\mathrm{max})}^{n,m}}\qw &\qw &\qw\\
\lstick{$\ket{i_I}_{n+m}$} &\qw &\qw&\qw&\qw&\qw&\qw\\
\lstick{$\ket{0}_{4}$} &\qw &\qw&\qw&\qw&\qw&\qw\\
\lstick{$\ket{0}_{n+m}$} &\qw &\qw&\qw&\qw&\qw&\qw\rstick{$\ket{\mathrm{M}_{I,n,m}(i_1,\ldots,i_I)}_{n+m}$}
\end{quantikz}
	} 
	\caption{Circuit diagram for $\mathcal{Q}_{(\max)}^{I,n,m}$ in Corollary~\ref{lemma: quantum maximum}.} 
	\label{fig: I max circuit diagram}
\end{figure}
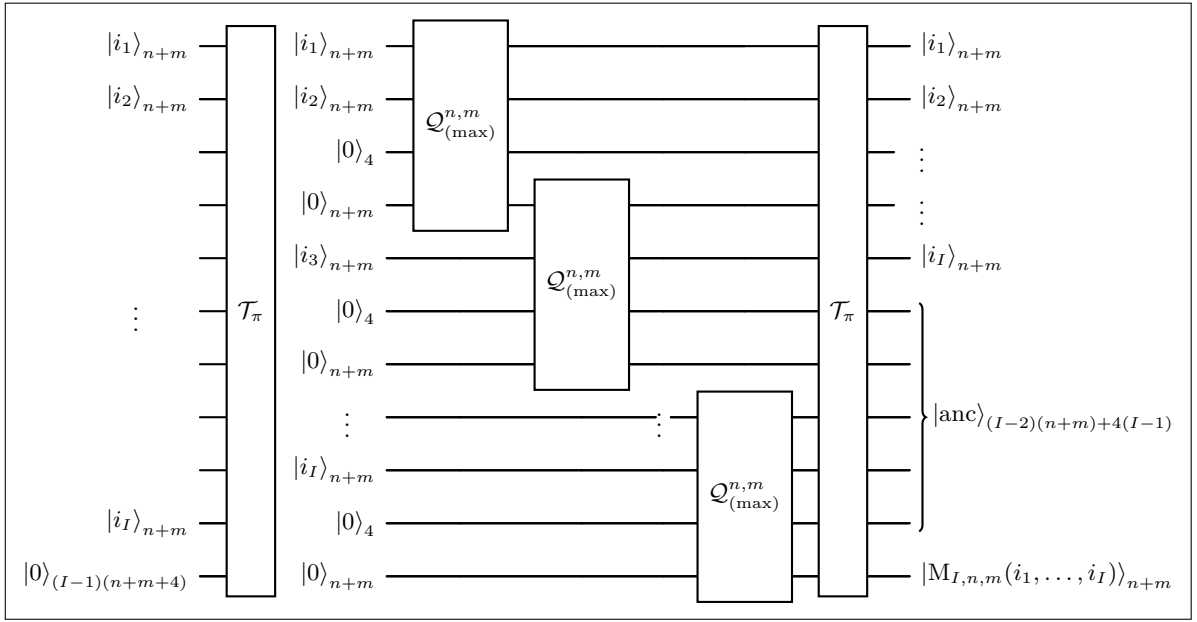
\begin{corollary}[Quantum circuit for maximum of $I$ numbers] \label{lemma: quantum maximum}
Let $I, n, m \in \N$. Let $\mathrm{M}_{I,n,m}:\mathbb{F}_{n,m}^I \to \mathbb{F}_{n,m}$ be a function defined by
\begin{equation}
    \mathrm{M}_{I,n,m}(i_1,\ldots,i_I) = \mathrm{E}_{n,m}(\max\{\mathrm{D}_{n,m}(i_1),\ldots,\mathrm{D}_{n,m}(i_I)\}), \quad \forall (i_1,\ldots,i_I) \in \mathbb{F}_{n,m}^I.
\end{equation}
(c.f. Definition \ref{def: encoder-decoder}). Then, there is a quantum circuit $\mathcal{Q}_{(\max)}^{I,n,m}$ on $N$ qubits, where
\begin{equation}
    N := I(n+m) + (I-1)(n+m+5)
\end{equation}
such that for any $i_1,\ldots,i_I \in \mathbb{F}_{n,m}$,
\begin{equation}\begin{split}
        \mathcal{Q}_{(\max)}^{I,n,m}:&\ket{i_1}_{n+m}\cdots\ket{i_I}_{n+m}\ket{0}_{(I-1)(n+m+5)}\\
        &\mapsto \ket{i_1}_{n+m}\cdots\ket{i_I}_{n+m}\ket{\anc}_{(I-2)(n+m) + 5(I-1)}\ket{\mathrm{M}_{I,n,m}(i_1,\ldots,i_I)}_{n+m}, \label{eqn: max circuit}
    \end{split}
\end{equation}
which uses at most $ 1189 I^2(n+m+1)^3$ elementary gates. See Figure~\ref{fig: I max circuit diagram} for the circuit diagram.

\begin{proof}For any $i_1,\ldots,i_I \in \mathbb{F}_{n,m}$, observe that the function $\mathrm{M}_{I,n,m}$ can be written recursively by setting
    \begin{equation}
        \begin{split}
            &\mathfrak{m}_1 := i_1,\\
            &\mathfrak{m}_2 := \mathrm{M}_{2,n,m}(i_1,i_2),\\
            &\mathfrak{m}_3 := \mathrm{M}_{2,n,m}(\mathfrak{m}_2,i_3) = \mathrm{M}_{3,n,m}(i_1,i_2,i_3),\\ 
            &\quad \vdots \\
            & \mathfrak{m}_I := \mathrm{M}_{2,n,m}(\mathfrak{m}_{I-1},i_I) = \cdots = \mathrm{M}_{I,n,m}(i_1,\ldots,i_I).\\           
        \end{split}
    \end{equation} With this setup in mind, we construct the quantum circuit $\mathcal{Q}_{(\max)}^{I,n,m}$ as follows.
\begin{enumerate}
    \item We first apply the permutation quantum circuit $\Tau_\pi$ from Lemma \ref{lemma: permutation circuit} to obtain
    \begin{equation}
        \begin{split}
            \Tau_\pi: 
            &\ket{i_1}_{n+m}\cdots\ket{i_I}_{n+m}\ket{0}_{(I-1)(n+m+5)}  \\
            &\mapsto 
            \ket{i_1}_{n+m}\ket{i_2}_{n+m} \ket{0}_5\ket{0}_{n+m} \ket{i_3}_{n+m}\ket{0}_5\ket{0}_{n+m}\ket{i_4}_{n+m}\cdots\ket{i_I}_{n+m}\ket{0}_5\ket{0}_{n+m}.
        \end{split}
    \end{equation}
    The number of elementary gates used is at most $2N^2$. 
   \item We apply inductively the maximum circuit for two TC numbers  $\mathcal{Q}_{(\max)}^{(k)}$ from Lemma \ref{lemma: circuit maximum 2} for $k\!=1,\ldots,I-1$ (with $n \leftarrow n$, $m \leftarrow m$, $a \leftarrow \mathfrak{m}_k$, $b \leftarrow i_{k+1}$ in the notation of Lemma \ref{lemma: circuit maximum 2}) so that
   
        \begin{equation}\begin{split}
        & \ket{i_1}_{n+m}\ket{i_2}_{n+m} \ket{0}_5\ket{0}_{n+m} \ket{i_3}_{n+m}\ket{0}_5\ket{0}_{n+m}\ket{i_4}_{n+m}\cdots\ket{i_I}_{n+m}\ket{0}_5\ket{0}_{n+m} \\
        &\xmapsto{\mathcal{Q}_{(\max)}^{(1)}} \ket{i_1}_{n+m}\ket{i_2}_{n+m} \ket{\anc}_5\ket{\mathrm{M}_{2,n,m}(i_1,i_2)}_{n+m} \ket{i_3}_{n+m}\ket{0}_5\ket{0}_{n+m}\cdots\ket{i_I}_{n+m}\ket{0}_{5}\ket{0}_{n+m}\\
        &\quad = \ket{i_1}_{n+m}\ket{i_2}_{n+m} \ket{\anc}_5\ket{\mathfrak{m}_2}_{n+m} \ket{i_3}_{n+m}\ket{0}_5\ket{0}_{n+m}\cdots\ket{i_I}_{n+m}\ket{0}_{5}\ket{0}_{n+m}\\
        &\xmapsto{\mathcal{Q}_{(\max)}^{(2)}} \ket{i_1}_{n+m}\ket{i_2}_{n+m} \ket{\anc}_5\ket{\mathfrak{m}_2}_{n+m} \ket{i_3}_{n+m}\ket{\anc}_5\ket{\mathrm{M}_{2,n,m}(\mathfrak{m}_2,i_3)}_{n+m} \ket{i_4}_{n+m}\ket{0}_5\ket{0}_{n+m}\\
        &\quad \quad \cdots\ket{i_I}_{n+m}\ket{0}_{5}\ket{0}_{n+m}\\
        &\quad =: \ket{i_1}_{n+m}\ket{i_2}_{n+m} \ket{\anc}_5\ket{\mathfrak{m}_2}_{n+m} \ket{i_3}_{n+m}\ket{\anc}_5\ket{\mathfrak{m}_3}_{n+m} \ket{i_4}_{n+m}\ket{0}_5\ket{0}_{n+m}\\
        &\quad \quad \cdots\ket{i_I}_{n+m}\ket{0}_{5}\ket{0}_{n+m}\\
        &\qquad\qquad\qquad\qquad\vdots\qquad\qquad\qquad\qquad\vdots \\
        &\xmapsto{\mathcal{Q}_{(\max)}^{(I)}} \ket{i_1}_{n+m}\ket{i_2}_{n+m} \ket{\anc}_5\ket{\mathfrak{m}_2}_{n+m} \ket{i_3}_{n+m}\cdots\ket{\mathfrak{m}_{I-1}}_{n+m} \ket{i_I}_{n+m} \ket{\anc}_5\ket{\mathrm{M}_{2,n,m}(\mathfrak{m}_{I-1},i_I)}_{n+m}\\
        &\quad =: \ket{i_1}_{n+m}\ket{i_2}_{n+m} \ket{\anc}_5\ket{\mathfrak{m}_2}_{n+m} \ket{i_3}_{n+m}\cdots\ket{\mathfrak{m}_{I-1}}_{n+m} \ket{i_I}_{n+m} \ket{\anc}_5\ket{\mathfrak{m}_{I}}_{n+m}.
        \end{split}
    \end{equation}
    
    The number of elementary gates used in this step is at most 
    \begin{equation}
        I \cdot 1045(n+m+1)^3.
    \end{equation}
    
    \item We consolidate the ancillary qubits by combining the following qubits: $(I-1)$ times of $\ket{\anc}_4$, and $(I-1)$ times of $\ket{\mathfrak{m}_2}_{n+m}$, $\cdots$,  $\ket{\mathfrak{m}_{I-1}}_{n+m}$ under the placeholder qubit $\ket{\anc}_{(I-2)(n+m) + 4(I-1)}$. The permutation quantum circuit $\Tau_\pi$ from Lemma \ref{lemma: permutation circuit} performs the following operation
    \begin{equation}
        \begin{split}
        \Tau_\pi:& \ket{i_1}_{n+m}\ket{i_2}_{n+m} \ket{\anc}_5\ket{\mathfrak{m}_2}_{n+m} \ket{i_3}_{n+m}\cdots\ket{\mathfrak{m}_{I-1}}_{n+m} \ket{i_I}_{n+m} \ket{\anc}_5\ket{\mathfrak{m}_{I}}\\
        & \mapsto \ket{i_1}_{n+m}\ket{i_2}_{n+m}\cdots\ket{i_I}_{n+m}\ket{\anc}_{(I-2)(n+m) + 5(I-1)}\ket{\mathfrak{m}_I}_{n+m}.
        \end{split}
    \end{equation}
The number of elementary gates used in this step is at most $2N^2$.
\end{enumerate}
The resulting quantum circuit $\mathcal{Q}_{(\max)}^{I,n,m}$ is a composition of the quantum circuits from each of the above steps. The number of elementary gates used in total is the sum of the number of gates used in each step which is at most
\begin{equation}
    \begin{split}
        &2N^2 + 1045I(n+m+1)^3 + 2N^2 \\
        &= 1045I(n+m+1)^3 + 4[I(n+m)+(I-1)(n+m+5)]^2 \\
        &\leq 1045I(n+m+1)^3 + 4[I(n+m+1) + 5I(n+m+1)]^2 \\
        &\leq 1045I(n+m+1)^3 + 4 \cdot 6^2 (I(n+m+1))^2 \\
        &\leq 1189 I^2(n+m+1)^3.
    \end{split}
\end{equation}
\end{proof}
\end{corollary}

\begin{figure}
	\centering 
\boxed{
\begin{quantikz}[row sep={20pt,between origins},column sep=20pt,font=\small]
\lstick{$\ket{i_1}_{n_1+m_1}$}&\gate[8][7cm]{\Tau_{\pi}^I\circ\mathcal{Q}_{+,I}^{d,n,m}\circ\cdots\circ\Tau_{\pi}^1\circ\mathcal{Q}_{+,1}^{d,n,m}}\gateoutput{$\ket{i_1}$}&\qw&\gate[8]{\Tau_\pi}&\qw\rstick{$\ket{i_1}_{n_1+m_1}$} \\
\lstick{$\vdots\ $}\qw&\gateoutput{$\vdots$}&\qw&\qw&\qw\rstick{$\vdots$}\\
\lstick{$\ket{i_d}_{n_1+m_1}$}&\gateoutput{$\ket{i_d}$}&\qw&\qw&\qw\rstick{$\ket{i_d}_{n_1+m_1}$}\\
\lstick[3]{$\ket{0}_\star$}&\gateoutput{$\ket{h_I(\bm{i})} $} &\gate[5][3cm]{\mathcal{Q}_{(\max)}^{I,n+d,m}}\gateoutput{$\ket{h_I(\bm{i})}$}&\qw&\qw\rstick[4]{$\ket{\anc}_{\star}$}\\
&\gateoutput{$\vdots$}&\qw\gateoutput{$\vdots$}&\qw&\qw\\
&\gateoutput{$\ket{h_1(\bm{i})} $}&\qw\gateoutput{$\ket{h_1(\bm{i})}$}&\qw&\qw\\
\lstick[2]{$\ket{0}_{n+m+d}$}&\gateoutput{$\ket{0}$}&\qw\gateoutput{$\ket{\anc}$}&\qw&\qw\\
&\gateoutput{$\ket{\anc}$}&\qw\gateoutput{$\ket{h(\bm{i})}$}&\qw&\qw\rstick{$\ket{h(\bm{i})}_{n+m+d}$}
\end{quantikz}		
} 
	\caption{Circuit diagram for $\mathcal{Q}_h$ in Proposition~\ref{prop: CPWA component circuit}.} 
	\label{fig: CPWA component diagram}
\end{figure}
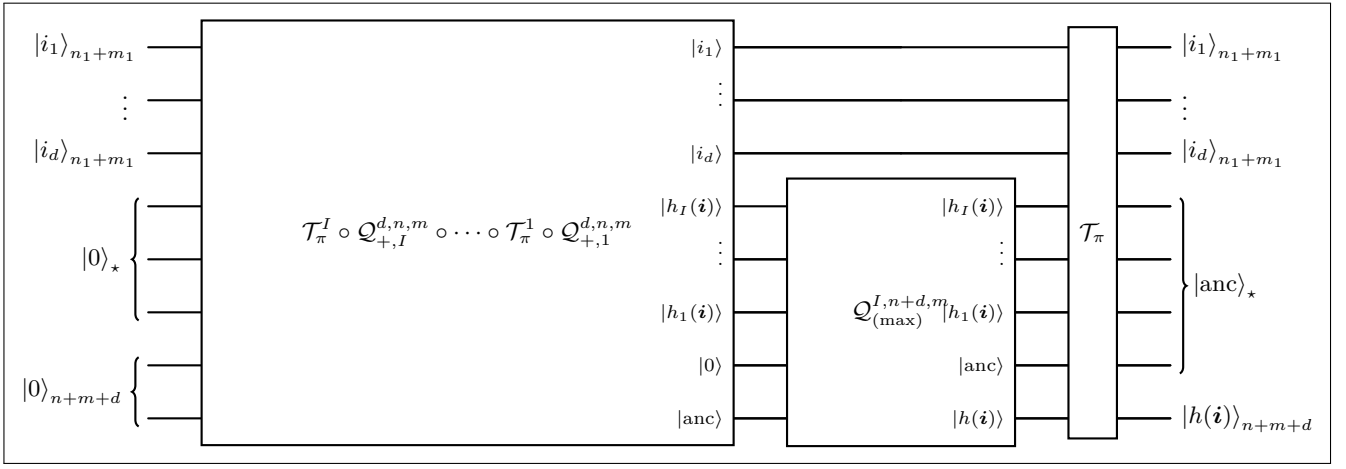
\begin{proposition}[Quantum circuit for loading continuous piecewise affine component functions]\label{prop: CPWA component circuit}
    Let $I,d,n_1,n_2,m_1,m_2 \in \N$. Define  $n:= n_1 + n_2$, $m := m_1 + m_2$, and $p := d(2n_2+2m_2+3)+(n_2+m_2)+(d-1)(n+m)$. Let $\{a_{l,j}\}_{l=1,\ldots,I;j=1,\ldots,d}, \{b_{l}\}_{l=1,\ldots,I} \subset \mathbb{F}_{n_2,m_2}$. Let $h_l: \mathbb{F}_{n_1,m_1}^d \to \mathbb{F}_{n+d,m}$, $l=1,\ldots,I$ be functions defined by 
    \begin{equation}\label{eqn: prop affine sum}
        h_l(i_1,\ldots,i_d) = \bigboxplus_{j=1}^d (a_{l,j} \boxdot i_j) \boxplus b_{l}, \quad \forall  (i_1,\ldots,i_d) \in \mathbb{F}_{n_1,m_1}^d,
    \end{equation}
    let $M_{I,n+d,m}:\mathbb{F}_{n+d,m}^I \to \mathbb{F}_{n+d,m}$ be a function defined by
    \begin{equation}
        M_{I,n+d,m}(i_1,\ldots,i_I) = \mathrm{E}_{n+d,m}(\max\{\mathrm{D}_{n+d,m}(i_1),\ldots,\mathrm{D}_{n+d,m}(i_I)\}), \quad \forall (i_1,\ldots,i_I) \in \mathbb{F}_{n+d,m}^I.
    \end{equation}
    Define $h: \mathbb{F}_{n_1,m_1}^d \to \mathbb{F}_{n+d,m}$ by
    \begin{equation}\label{eqn: prop payoff function}
        h(\bm{i}) :=M_{I,n+d,m}\big(h_1(\bm{i}),\ldots,h_I(\bm{i})\big), \quad \forall \bm{i}= (i_1,\ldots,i_d) \in \mathbb{F}_{n_1,m_1}^d.
    \end{equation}
    Then, there is a quantum circuit $\mathcal{Q}_h$ on $N$ qubits, where
    \begin{equation}
        N := d(n_1+m_1)  + I(n+m+d+p)+ (I-1)(n+m+d+5),
    \end{equation}
    such that for any $\bm{i} = (i_1,\ldots,i_d) \in \mathbb{F}_{n_1,m_1}^d$,
\begin{equation}\label{eqn: component payoff state}
    \begin{split}
        \mathcal{Q}_h :& \ket{i_1}_{n_1+m_1}\cdots\ket{i_d}_{n_1+m_1}\ket{0}_{I(n+m+d+p)+ (I-2)(n+m+d) + 5(I-1)}\ket{0}_{n+m+d}\\ &\mapsto \ket{i_1}_{n_1+m_1}\cdots\ket{i_d}_{n_1+m_1}\ket{\anc}_{I(n+m+d+p)+ (I-2)(n+m+d) + 5(I-1)}\ket{h(\bm{i})}_{n+m+d} ,
    \end{split}
\end{equation}    
which uses at most $10651 I^3d^3(n+m+1)^3$ elementary gates. See Figure~\ref{fig: CPWA component diagram} for the circuit diagram.
\begin{proof}
The construction of the quantum circuit $\mathcal{Q}_h$ involves the following steps:
\begin{enumerate}
    \item We first prepare the $I$ affine sums $h_l(i_1,\ldots,i_d)$ from \eqref{eqn: prop affine sum}, using Lemma \ref{lemma: quantum affine sum}. \\
    For $l=1,\ldots,I$, we apply the quantum circuits $(\mathcal{Q}_{+,l}^{d,n,m})_{l=1,\ldots,I}$ of Lemma \ref{lemma: quantum affine sum} \\
    (with $(d,n_1,n_2,m_1,m_2,a_1,\ldots,a_d,b)\leftarrow (d,n_1,n_2,m_1,m_2,a_{l,1},\ldots,a_{l,d},b_l)$  in the notation of Lemma \ref{lemma: quantum affine sum}) followed by an application of the permutation circuit $\Tau_\pi$ from Lemma \ref{lemma: permutation circuit}, where we compute 
    \begin{equation}
        \begin{split}
            &\ket{i_1}_{n_1+m_1}\cdots\ket{i_d}_{n_1+m_1}\ket{0}_{I(n+m+d+p)}\ket{0}_{(I-1)(n+m+d+5)}\\
            &\xmapsto{\mathcal{Q}_{+,1}^{d,n,m}}\ket{i_1}_{n_1+m_1}\cdots\ket{i_d}_{n_1+m_1}\ket{h_1(\bm{i})}_{n+m +d}\ket{\anc}_p\ket{0}_{(I-1)(n+m+d+p)} \ket{0}_{(I-1)(n+m+d+5)}\\
            &\xmapsto{\Tau_\pi^1} \ket{i_1}_{n_1+m_1}\cdots\ket{i_d}_{n_1+m_1}\ket{0}_{n+m+d+p}\ket{h_1(\bm{i})}_{n+m +d}\ket{\anc}_p \ket{0}_{(I-2)(n+m+d+p)} \ket{0}_{(I-1)(n+m+d+5)}\\
            &\xmapsto{\mathcal{Q}_{+,2}^{d,n,m}}\ket{i_1}_{n_1+m_1}\cdots\ket{i_d}_{n_1+m_1}\ket{h_2(\bm{i})}_{n+m +d}\ket{\anc}_p\ket{h_1(\bm{i})}_{n+m +d}\ket{\anc}_p \\
            &\hspace{100pt}\cdot \ket{0}_{(I-2)(n+m+d+p)} \ket{0}_{(I-1)(n+m+d+5)}\\
            &\xmapsto{\Tau_\pi^2}\ket{i_1}_{n_1+m_1}\cdots\ket{i_d}_{n_1+m_1}\ket{0}_{n+m+d+p}\ket{h_2(\bm{i})}_{n+m +d}\ket{h_1(\bm{i})}_{n+m +d}\ket{\anc}_{2p} \\
            &\hspace{100pt}\cdot \ket{0}_{(I-3)(n+m+d+p)} \ket{0}_{(I-1)(n+m+d+5)}\\
            &\hspace{100pt}\vdots \hspace{100pt} \vdots\\
            &\xmapsto{\mathcal{Q}_{+,I}^{d,n,m}}\ket{i_1}_{n_1+m_1}\cdots\ket{i_d}_{n_1+m_1}\ket{h_I(\bm{i})}_{n+m +d}\ket{\anc}_p\\
            &\hspace{100pt} \cdots \ket{h_2(\bm{i})}_{n+m +d}\ket{h_1(\bm{i})}_{n+m +d}\ket{\anc}_{(I-1)p} \ket{0}_{(I-1)(n+m+d+5)}\\
            &\xmapsto{\Tau_\pi^I}\ket{i_1}_{n_1+m_1}\cdots\ket{i_d}_{n_1+m_1}\ket{h_I(\bm{i})}_{n+m +d}\cdots \ket{h_2(\bm{i})}_{n+m +d}\ket{h_1(\bm{i})}_{n+m +d}\\
            &\hspace{100pt}\cdot \ket{0}_{(I-1)(n+m+d+5)}\ket{\anc}_{Ip}.
        \end{split}
    \end{equation}
    In this step, the number of elementary gates used is at most
    \begin{equation}
        I[2N^2 + 563d^3(n+m+1)^2].
    \end{equation}

    \item Next, we compute the maximum value amongst the affine sums $h_1,\ldots,h_I$. 
    We apply the quantum circuit $\mathcal{Q}_{(\max)}^{I,n+d,m}$ of Corollary \ref{lemma: quantum maximum} (with $I \leftarrow I$, $n \leftarrow d+n$, $m \leftarrow m$, $i_1,\ldots,i_I \leftarrow h_1(\bm{i}),\ldots,h_I(\bm{i})$ in the notation of Corollary \ref{lemma: quantum maximum}) where we have
    \begin{equation}
        \begin{split}
            &\ket{i_1}_{n_1+m_1}\cdots\ket{i_d}_{n_1+m_1}\ket{h_I(\bm{i})}_{n+m +d}\cdots \ket{h_2(\bm{i})}_{n+m +d}\ket{h_1(\bm{i})}_{n+m +d}\\
            &\hspace{100pt}\cdot \ket{0}_{(I-1)(n+m+d+5)}\ket{\anc}_{Ip}\\
            &\xmapsto{\mathcal{Q}_{(\max)}^{I,n+d,m}}\ket{i_1}_{n_1+m_1}\cdots\ket{i_d}_{n_1+m_1}\ket{h_I(\bm{i})}_{n+m +d}\cdots \ket{h_2(\bm{i})}_{n+m +d}\ket{h_1(\bm{i})}_{n+m +d}\\
            &\hspace{100pt} \ket{\anc}_{(I-2)(n+m+d)+5(I-1)} \ket{h(\bm{i})}_{n+m+d}\ket{\anc}_{Ip}
        \end{split}
    \end{equation}
    
    In this step, the number of elementary gates used is at most
    \begin{equation}
        1189 I^2(n+m+d+1)^3.
    \end{equation}
    
    \item We use the permutation quantum circuit $\Tau_\pi$ from Lemma \ref{lemma: permutation circuit} and we put the qubits $\ket{h_I(\bm{i})}_{n+m +d}\cdots \ket{h_1(\bm{i})}_{n+m +d}$ under $\ket{\anc}$, so that we have 
    \begin{equation}
        \begin{split}
            &\ket{i_1}_{n_1+m_1}\cdots\ket{i_d}_{n_1+m_1}\ket{h_I(\bm{i})}_{n+m +d}\cdots \ket{h_1(\bm{i})}_{n+m +d} \\
            &\hspace{120pt} \cdot \ket{\anc}_{(I-2)(n+m+d)+5(I-1)} \ket{h(\bm{i})}_{n+m+d}\ket{\anc}_{Ip} \\ &\xmapsto{\Tau_\pi} \ket{i_1}_{n_1+m_1}\cdots\ket{i_d}_{n_1+m_1} \ket{\anc}_{I(n+m+d+p) + (I-2)(n+m+d)+5(I-1)} \ket{h(\bm{i})}_{n+m+d}. 
        \end{split}
    \end{equation}
    We hence reach the desired state \eqref{eqn: component payoff state}. In this step, the number of elementary gates used is at most \begin{equation}
        2N^2.
    \end{equation}
\end{enumerate}
     We note that
\begin{equation}\label{eqn: bound on p}
    p = d(2n_2+2m_2+3)+(n_2+m_2)+(d-1)(n+m) \leq 4d(n+m+1).
\end{equation}
Hence, we obtain that
\begin{equation}
    \begin{split}
        N &=  d(n_1+m_1) + I(n+m+d+p)+ (I-1)(n+m+d+5) \\
        &\leq d(n+m+1) + Id(n+m+1) + Ip + (d+5)I(n+m+1)\\
        &\leq (1+1+4+1+5)Id(n+m+1)\\
        &= 12Id(n+m+1).
    \end{split}
\end{equation}
Thus, the total number of elementary gates used is at most
\begin{equation}
    \begin{split}
        &I[2N^2+563d^3(n+m+1)^2]+ 1189 I^2(n+m+d+1)^3 + 2N^2\\
        &= (I+1)2N^2 + 563Id^3(n+m+1)^2 + 1189 I^2(n+m+d+1)^3\\
        &\leq 4IN^2 + 563Id^3(n+m+1)^2 + 1189 I^2(2d)^3(n+m+1)^3\\
        &\leq 4I(12 Id(n+m+1))^2 + 563Id^3(n+m+1)^2 + 1189 I^2(2d)^3(n+m+1)^3\\
        &\leq (4\cdot12^2 + 563 + 1189 \cdot2^3)I^3d^3(n+m+1)^3\\
        &= 10651 I^3d^3(n+m+1)^3.
    \end{split}
\end{equation}

    \end{proof}
\end{proposition}

\begin{figure}
	\centering 
	\boxed{
\begin{quantikz}[row sep={20pt,between origins},column sep=10pt,font=\footnotesize]
\lstick{$\ket{i_{-m}}$}&\qw&\qw&\qw&\qw&\qw\ \ldots \ &\qw&\gate[6,nwires={2,4}]{\Tau_\pi}&\rstick{$\ket{i_{-m}}$}\\
\lstick{$\vdots\ $}&&&&&&&&\rstick{$\vdots$}\\
\lstick{$\ket{i_{0}}$}&\qw&\qw&\qw&\qw&\qw\ \ldots \ &\qw&\qw&\rstick{$\ket{i_{0}}$}\\
\lstick{$\vdots\ $}&&&&&&&&\rstick{$\vdots$}\\
\lstick{$\ket{i_{n-1}}$}&\qw &\ctrl{1}&\gate[2]{\Tau_{\pi}}&\ctrl{1}&\qw\ \ldots \ &\ctrl{1}&\qw&\rstick{$\ket{i_{n-1}}$}\\
\lstick{$\ket{0}$}&\gate{R_y(a_0)} &\gate{CR_y(-a_1  2^{n-1})}&\qw&\gate{CR_y(-a_1  2^{n-2})}&\qw \ \ldots \ &\gate{CR_y(-a_1  2^{-m})}&\qw&\rstick{$\cos(\tfrac{\bar{f}(i)}{2})\ket{0} + \sin(\tfrac{\bar{f}(i)}{2})\ket{1}$}
\end{quantikz}		
 	} 
	\caption{Circuit diagram for $\mathcal{R}_{f}$ in Lemma~\ref{lemma: quantum linear rotation}.}
	\label{fig: Y rotation circuit diagram}
\end{figure}
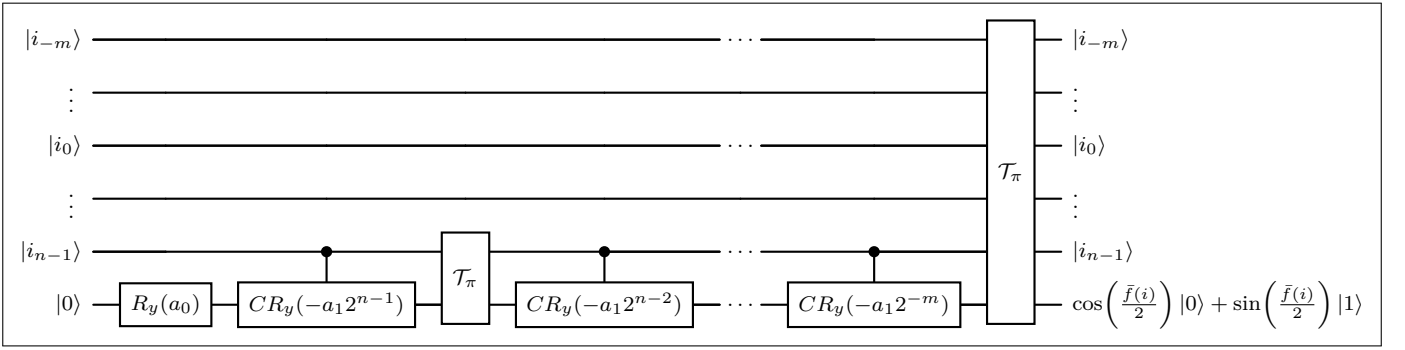

\begin{lemma}[Quantum circuit for $Y$-rotation] \label{lemma: quantum linear rotation}
	Let $n\in \N$, $m \in \N_0$, $a_0,a_1 \in \R$, and let $f(x) = a_1 x + a_0$ for $x \in \R$. Define
	\begin{equation}
		\bar{f}(i) := f\circ \mathrm{D}_{n,m}(i), \quad \forall i \in \mathbb{F}_{n,m}.
	\end{equation}
	Then, there is a quantum circuit $\mathcal{R}_f$ on $(n+m+1)$ qubits such that for any $i \in \mathbb{F}_{n,m}$,
	\begin{equation}
		\mathcal{R}_f:\ket{i}_{n+m}\ket{0} \mapsto \ket{i}_{n+m}[\cos(\bar{f}(i)/2) \ket{0} + \sin(\bar{f}(i)/2) \ket{1}],
	\end{equation}
	which uses $13(n+m+1)^3$ elementary gates. See Figure~\ref{fig: Y rotation circuit diagram} for the circuit diagram.
	
	\begin{proof}
		The quantum circuit $\mathcal{R}_f$ is constructed from the following steps:
		\begin{enumerate}
			\item We apply the $Y$-rotation gate $R_y(\theta)$ 
			with parameter $\theta \leftarrow a_0$
			  to obtain the state
			\begin{equation}
				I^{\otimes n+m}_2\otimes R_y(a_0):\ket{i}_{n+m}\ket{0} \mapsto \ket{i}_{n+m}(\cos(a_0/2)\ket{0} + \sin(a_0/2)\ket{1}).
			\end{equation}
			\item We apply the controlled $Y$-rotation gate $CR_y(\theta)$ of Lemma \ref{lemma: controlled y-rotation} on the qubit $\ket{i_{n-1}}$ (with parameter $\theta \leftarrow -a_1\cdot 2^{n-1}$, control qubit $c \leftarrow i_{n-1}$ and target qubit $0 \leftarrow (\cos(a_0/2)\ket{0} + \sin(a_0/2)\ket{1})$ in the notation of Lemma \ref{lemma: controlled y-rotation}) to obtain
			\begin{equation}
				\begin{split}
					&I^{\otimes n+m-1}_2\otimes CR_y(-a_1\cdot 2^{n-1}): \ket{i_{-m}}\cdots\ket{i_0}\cdots\ket{i_{n-1}} (\cos(a_0/2)\ket{0} + \sin(a_0/2)\ket{1}) \\
					&\mapsto  \ket{i_{-m}}\cdots\ket{i_0}\cdots\ket{i_{n-1}}(\cos((- a_1 2^{n-1}i_{n-1} + a_0)/2)\ket{0} + \sin((-a_1 2^{n-1}i_{n-1} + a_0)/2)\ket{1}) \\
					&\quad = \begin{cases}
						\ket{i}_{n+m}(\cos(a_0/2)\ket{0} + \sin(a_0/2)\ket{1}), &\quad \text{if $i_{n-1}=0$},\\
						\ket{i}_{n+m}(\cos((-a_1 2^{n-1} + a_0)/2)\ket{0} + \sin((-a_1 2^{n-1} + a_0)/2)\ket{1}), &\quad \text{if $i_{n-1}=1$}.
					\end{cases}
				\end{split}
			\end{equation}
			\item Similarly, we inductively apply for $k=n-2,\ldots,0,\ldots,-m$ the permutation quantum circuit $\Tau_{k \leftrightarrow n-1}$ from Lemma \ref{lemma: permutation circuit}, the controlled $Y$-rotation gate $CR_y(\theta)$ (with parameter $\theta \leftarrow a_1 2^{k}$ and control qubit $c \leftarrow i_{k}$ in the notation of Lemma \ref{lemma: controlled y-rotation}), and another permutation circuit $\Tau_{k\leftrightarrow n-1}$ to obtain
			\begin{equation}
				\begin{split}
					&\ket{i_{-m}}\cdots\ket{i_0}\cdots\ket{i_{n-1}}(\cos((- a_1 2^{n-1}i_{n-1} + a_0)/2)\ket{0} + \sin((-a_1 2^{n-1}i_{n-1} + a_0)/2)\ket{1})\\
					&\xmapsto{\Tau_{n-2 \leftrightarrow n-1}}\ket{i_{-m}}\cdots\ket{i_0}\cdots\ket{i_{n-1}}\ket{i_{n-2}}(\cos((- a_1 2^{n-1}i_{n-1} + a_0)/2)\ket{0} + \sin((-a_1 2^{n-1}i_{n-1} + a_0)/2)\ket{1})\\
					&\xmapsto{ CR_y(a_12^{n-2}) }\ket{i_{-m}}\cdots\ket{i_0}\cdots\ket{i_{n-1}}\ket{i_{n-2}} \bigg(\cos((a_1(2^{n-2}i_{n-2} -2^{n-1}i_{n-1}) + a_0)/2)\ket{0} \\
					&\hspace{100pt}+ \sin((a_1(2^{n-2}i_{n-2} -2^{n-1}i_{n-1}) + a_0)/2)\ket{1}\bigg)\\
					&\xmapsto{\Tau_{n-2 \leftrightarrow n-1}}\ket{i_{-m}}\cdots\ket{i_0}\cdots\ket{i_{n-2}}\ket{i_{n-1}} \bigg(\cos((a_1(2^{n-2}i_{n-2} -2^{n-1}i_{n-1}) + a_0)/2)\ket{0} \\
					&\hspace{100pt}+ \sin((a_1(2^{n-2}i_{n-2} -2^{n-1}i_{n-1}) + a_0)/2)\ket{1}\bigg)\\
					&\qquad\qquad\qquad\qquad\vdots\qquad\qquad\qquad\qquad\vdots \\ 
					&\mapsto \ket{i}_{n+m}(\cos((a_1 (-2^{n-1}i_{n-1} + \sum_{k=-m}^{n-2}2^ki_k) + a_0)/2)\ket{0} + \sin((a_1 (-2^{n-1}i_{n-1} + \sum_{k=-m}^{n-2}2^ki_k) + a_0)/2)\ket{1})\\
					&= \ket{i}_{n+m}[\cos(\bar{f}(i)/2) \ket{0} + \sin(\bar{f}(i)/2) \ket{1}]
				\end{split}
			\end{equation}
		\end{enumerate}
		The permutation circuits $\Tau_\pi$ from Lemma \ref{lemma: permutation circuit} requires $2(n+m+1)^2$ gates and the controlled $Y$-rotation gate $CR_y(\theta)$ requires 4 elementary gates (c.f. Lemma \ref{lemma: controlled y-rotation}). Hence, the total number of gates required for circuit $\mathcal{R}_f$ is 
		\begin{equation}
			1 + 4 + (n+m-1)[2\cdot 2(n+m+1)^2 + 4]  \leq 13(n+m+1)^3.
		\end{equation}
		Thus, we conclude the proof of the lemma.
	\end{proof}
\end{lemma}

\begin{proposition}[Quantum circuit for continuous piecewise affine payoff function with $Y$-rotation]\label{prop: loading payoff circuit}
Let  $d\in \N$, $n_1,n_2,m_1$, $m_2\in \N$, $K \in \N$, $I_1,\cdots,I_K \in \N$, $\xi_1,\ldots,\xi_K \in \{-1,1\}$, and $s \in (0,1)$. Define $n:= n_1 + n_2$, $m = m_1 + m_2$, and $p := d(2n_2+2m_2+3)+(n_2+m_2)+(d-1)(n+m)$. Let $\{a_{k,l,j}\}_{k=1,\ldots,K;l=1,\ldots,I_k;j=1,\ldots,d}$, $\{b_{k,l}\}_{k=1,\ldots,K;l=1,\ldots,I_k} \subset \mathbb{F}_{n_2,m_2}$. For $k=1,\ldots,K$, $l=1,\ldots,I_k$, let $h_{k,l}: \mathbb{F}_{n_1,m_1}^d \to \mathbb{F}_{n+d,m}$ be functions defined by 
\begin{equation}\label{eqn: prop affine sum2}
    h_{k,l}(\bm{i}) := \bigboxplus_{j=1}^d (a_{k,l,j} \boxdot i_j) \boxplus b_{k,l}, \quad \forall \bm{i}= (i_1,\ldots,i_d) \in \mathbb{F}_{n_1,m_1}^d.
\end{equation}
For $k=1,\ldots,K$, let $\overline{\xi}_k \in \mathbb{F}_{2,0}$ be defined by 
\begin{equation}
    \overline{\xi}_k := \mathrm{E}_{2,0}(\xi_k),
\end{equation}
let $M_{I_k,n+d,m}:\mathbb{F}_{n+d,m}^{I_k} \to \mathbb{F}_{n+d,m}$ be defined by
\begin{equation}
    M_{{I_k},n+d,m}(i_1,\ldots,i_{I_k}) := \mathrm{E}_{n+d,m}(\max\{\mathrm{D}_{n+d,m}(i_1),\ldots,\mathrm{D}_{n+d,m}(i_{I_k})\}), \quad \forall (i_1,\ldots,i_{I_k}) \in \mathbb{F}_{n+d,m}^{I_k},
\end{equation}
and let $\bm{h}_k: \mathbb{F}_{n_1,m_1}^d \to \mathbb{F}_{n+d,m}$ be defined by
\begin{equation}\label{eqn: prop payoff function2}
    \bm{h}_k(\bm{i}) :=M_{I_k,n+d,m}\big(h_{k,1}(\bm{i}),\ldots,h_{k,I_k}(\bm{i})\big), \quad \forall \bm{i}= (i_1,\ldots,i_d) \in \mathbb{F}_{n_1,m_1}^d.
\end{equation}
Let $\bm{h}:\mathbb{F}_{n_1,m_1}^d \to \mathbb{F}_{n+d+K+1,m}$ be defined by 
\begin{equation}
    \bm{h}(\bm{i}) := \bigboxplus_{k=1}^K (\overline{\xi}_k \boxdot \bm{h}_k(\bm{i})), \quad \forall \bm{i} = (i_1,\ldots,i_d) \in \mathbb{F}_{n_1,m_1}^d. \label{eqn: encoded payoff}
\end{equation}
Let $f:\R \to \R$ be a function defined by $f(x) = sx + \frac{\pi}{2}$, and define $\bar{f}:\mathbb{F}_{n +d +K + 1,m} \to \R$ by 
\begin{equation}
    \bar{f}(i) := f \circ \mathrm{D}_{n+d+K + 1,m}(i), \quad \forall i\in \mathbb{F}_{n +d+K + 1,m}.
\end{equation}
Then, there is a quantum circuit $\mathcal{R}_{\bm{h}} $ on $N$ qubits, where
\begin{equation}\label{eqn: def N, q_k}
    \begin{split}
        &N := d(n_1+m_1) + \sum_{k=1}^K q_k +  2K(n+m+d+5),\\
        &q_k := I_k(n+m+d+p) + (I_k-2)(n+m+d) + 5(I_k-1), \quad k=1,\ldots,K,
    \end{split}
\end{equation} 
such that for any $\bm{i} = (i_1,\ldots,i_d) \in \mathbb{F}_{n_1,m_1}^d$,
\begin{equation}
    \begin{split}
        \mathcal{R}_{\bm{h}}:&\ket{i_1}_{n_1+m_1}\cdots\ket{i_d}_{n_1+m_1}\ket{0}_{q_1+\cdots+q_K}\ket{0}_{2K(n+m+d+5)} \\
        &\mapsto \ket{i_1}_{n_1+m_1}\cdots\ket{i_d}_{n_1+m_1}\ket{\anc}_{q_1+\cdots+q_K + 2K(n+m+d+5)-1} \left[ \cos(\bar{f}(\bm{h}(\bm{i}))/2)\ket{0} + \sin(\bar{f}(\bm{h}(\bm{i}))/2)\ket{1} \right], \label{eqn: grover state}
    \end{split}
\end{equation}
which uses at most $16186 K^3\left(\max_{k=1,\ldots,K}\{I_k\}\right)^3d^3(n+m+1)^3$ elementary gates. See Figure~\ref{fig: payoff circuit diagram1} and Figure~\ref{fig: payoff circuit diagram2} for the circuit diagram.
\end{proposition}

\begin{figure}[t]
	\centering 
	\boxed{
		\begin{quantikz}[row sep={20pt,between origins},column sep=10pt,font=\small]
			\lstick{$\ket{i_1}_{n_1+m_1}$}&\gate[5]{\mathcal{Q}_{\bm{h}_1}}&\gate[10]{\mathcal{T}_{\pi_1}}&\gate[5]{\mathcal{Q}_{\bm{h}_2}}&\gate[12,nwires=13]{\mathcal{T}_{\pi_2}}&\qw \ \ldots\ &\gate[5]{\mathcal{Q}_{\bm{h}_K}}&\gate[16,nwires=13]{\mathcal{T}_{\pi_K}}&\qw&\qw&\qw\rstick{$\ket{i_1}_{n_1+m_1}$}\\
			\lstick{$\vdots$}&\qw&\qw&\qw&\qw&\qw \ \ldots\ &\qw&\qw&\qw&\qw&\qw\rstick{$\vdots$}\\
			\lstick{$\ket{i_d}_{n_1+m_1}$}&\qw&\qw&\qw&\qw&\qw \ \ldots\ &\qw&\qw&\qw&\qw&\qw\rstick{$\ket{i_d}_{n_1+m_1}$}\\ 
			\lstick[15]{$\ket{0}_{\star}$}&\qw&\qw&\qw&\qw&\qw \ \ldots\ &\qw&\qw&\qw&\qw&\qw\rstick{$\ket{\anc}_{q_1+\cdots+q_K}$}\\ 
			&\qw&\qw&\qw&\qw&\qw \ \ldots\ &\qw&\qw&\qw&\gate[4]{\widetilde{\mathcal{Q}_{(\times)}^{(1)}}}&\qw\rstick{$\ket{\bm{h}_1(\bm{i})}_{n+m+d}$}\\ 
			&\qw&\qw&\qw&\qw&\qw \ \ldots\ &\qw&\qw&\gate{X^{\overline{\xi}_1}}&\qw&\qw\rstick{$\ket{{\overline{\xi}_1}}_2$}\\ 
			&\qw&\qw&\qw&\qw&\qw \ \ldots\ &\qw&\qw&\qw&\qw&\qw\rstick{$\ket{\bm{h}_1(\bm{i}) \boxdot \overline{\xi}_1}_{n+m+d+2}$}\\ 
			&\qw&\qw&\qw&\qw&\qw \ \ldots\ &\qw&\qw&\qw&\qw&\qw\rstick{$\ket{\anc}_{5}$}\\ 
			&\qw&\qw&\qw&\qw&\qw \ \ldots\ &\qw&\qw&\qw&\gate[4]{\widetilde{\mathcal{Q}_{(\times)}^{(2)}}}&\qw\rstick{$\ket{\bm{h}_{2}(\bm{i})}_{n+m+d}$}\\ 
			&\qw&\qw&\qw&\qw&\qw \ \ldots\ &\qw&\qw&\gate{X^{\overline{\xi}_2}}&\qw&\qw\rstick{$\ket{\overline{\xi}_2}_2$}\\ 
			&\qw&\qw&\qw&\qw&\qw \ \ldots\ &\qw&\qw&\qw&\qw&\qw\rstick{${\ket{\bm{h}_2(\bm{i}) \boxdot \overline{\xi}_2}_{n+m+d+2}}$}\\ 
			&\qw&\qw&\qw&\qw&\qw \ \ldots\ &\qw&\qw&\qw&\qw&\qw\rstick{$\ket{\anc}_5$}\\ 
			&&&&& \ \ldots\ &&&&\vdots&\rstick{$\vdots$}\\ 
			&\qw&\qw&\qw&\qw&\qw \ \ldots\ &\qw&\qw&\qw&\gate[4]{\widetilde{\mathcal{Q}_{(\times)}^{(K)}}}&\qw\rstick{$\ket{\bm{h}_{K}(\bm{i})}_{n+m+d}$}\\ 
			&\qw&\qw&\qw&\qw&\qw \ \ldots\ &\qw&\qw&\gate{X^{\overline{\xi}_K}}&\qw&\qw\rstick{$\ket{\overline{\xi}_K}_2$} \\
			&\qw&\qw&\qw&\qw&\qw \ \ldots\ &\qw&\qw&\qw&\qw&\qw\rstick{${\ket{\bm{h}_K(\bm{i}) \boxdot \overline{\xi}_K}_{n+m+d+2}}$}\\ 
			&\qw&\qw&\qw&\qw&\qw \ \ldots\ &\qw&\qw&\qw&\qw&\qw\rstick{$\ket{\anc}_5$} \\
			&\qw&\qw&\qw&\qw&\qw \ \ldots\ &\qw&\qw&\qw&\qw&\qw\rstick{$\ket{0}_K$} 
		\end{quantikz}		
	} 
	\caption{Circuit diagram for $\mathcal{R}_{\bm{h}}$ in Proposition~\ref{prop: loading payoff circuit} (Steps 1--3).}
	\label{fig: payoff circuit diagram1}
\end{figure}

\begin{figure}[t]
	\centering 
	\boxed{
		\begin{quantikz}[row sep={20pt,between origins},column sep=10pt,font=\small]
			\lstick{$\ket{i_1}$}&\gate[18]{\mathcal{T}_{\pi}}\\    
			\lstick{$\vdots$}&\\
			\lstick{$\ket{i_d}$}&\\
			\lstick{$\ket{\anc}$}&\\
			\lstick{$\ket{\bm{h}_1(\bm{i})}$}&\\
			\lstick{$\ket{\overline{\xi}_1}$}&\\
			\lstick{${\ket{\bm{h}_1(\bm{i}) \boxdot \overline{\xi}_1}}$}&\\
			\lstick{$\ket{\anc}$}&\\
			\lstick{$\ket{\bm{h}_2(\bm{i})}$}&\\
			\lstick{$\ket{\overline{\xi}_2}$}&\\
			\lstick{${\ket{\bm{h}_2(\bm{i}) \boxdot \overline{\xi}_2}}$}&\\
			\lstick{$\ket{\anc}$}&\\
			\lstick{$\vdots$}&\\
			\lstick{$\ket{\bm{h}_K(\bm{i})}$}&\\
			\lstick{$\ket{\overline{\xi}_K}$}&\\
			\lstick{${\ket{\bm{h}_K(\bm{i}) \boxdot \overline{\xi}_K}}$}&\\
			\lstick{$\ket{\anc}$}&\\
			\lstick{$\ket{0}_K$}&
		\end{quantikz}
		\text{$\Longrightarrow$}\hspace{2mm}
			\begin{quantikz}[row sep={20pt,between origins},column sep=10pt,font=\small]
				\lstick{$\ket{i_1}$}&\qw&\qw&\qw\ \ldots \ &\qw&\qw&\qw\rstick{$\ket{i_1}_{n_1+m_1}$}\\    
				\lstick{$\vdots$}&\qw&\qw&\qw\ \ldots \ &\qw&\qw&\qw\rstick{$\vdots$}\\
				\lstick{$\ket{i_d}$}&\qw&\qw&\qw\ \ldots \ &\qw&\qw&\qw\rstick{$\ket{i_d}_{n_1+m_1}$}\\
				\lstick{$\ket{\anc}$}&\qw&\qw&\qw\ \ldots \ &\qw&\qw&\qw\rstick[8]{$\ket{\anc}_\star$}\\
				\lstick{${\ket{\bm{h}_1(\bm{i}) \boxdot \overline{\xi}_1}}$}&\gate[3]{\mathcal{Q}_{(+)}^{(1)}}&\qw&\qw\ \ldots \ &\qw&\qw&\qw\\
				\lstick{${\ket{\bm{h}_2(\bm{i}) \boxdot \overline{\xi}_2}}$}&\qw&\qw&\qw\ \ldots \ &\qw&\qw&\qw\\
				\lstick{$\ket{0}$}&\qw&\gate[3]{\mathcal{Q}_{(+)}^{(2)}}\qw&\qw\ \ldots \ &\qw&\qw&\qw\\
				\lstick{${\ket{\bm{h}_3(\bm{i}) \boxdot \overline{\xi}_3}}$}&\qw&\qw&\qw\ \ldots \ &\qw&\qw&\qw\\
				\lstick{$\ket{0}$}&\qw&\qw&\qw\ \ldots \ &\qw&\qw&\qw\\
				\lstick{$\vdots$}\qw&\qw&\qw&\qw\ \ldots \ &\gate[3,nwires=1]{\mathcal{Q}_{(+)}^{(K)}}\qw&\qw&\qw\\
				\lstick{${\ket{\bm{h}_K(\bm{i}) \boxdot \overline{\xi}_K}}$}&\qw&\qw&\qw\ \ldots \ &\qw&\qw&\qw\\
				\lstick{$\ket{0}$}&\qw&\qw&\qw\ \ldots \ &\qw&\gate[2]{\mathcal{R}_f}&\qw\rstick{$\ket{\bm{h}(\bm{i})}_{n+m+d+K+1}$}\\
				\lstick{$\ket{0}$}&\qw&\qw&\qw\ \ldots \ &\qw&\qw&\qw\rstick{$[\cos(\bar{f}({\bm{h}(\bm{i})})/2)\ket{0}$}\\
				\setwiretype{n} &&&&&&\rstick{$\ +\sin(\bar{f}({\bm{h}(\bm{i})})/2)\ket{1}]$}
			\end{quantikz}
	} 
		\caption{Circuit diagram for $\mathcal{R}_{\bm{h}}$ in Proposition~\ref{prop: loading payoff circuit} (Steps 4--6).}
		\label{fig: payoff circuit diagram2}
\end{figure}

\begin{proof}    
The construction of the quantum circuit $\mathcal{R}_{\bm{h}}$ involves the following steps:
\begin{enumerate}
    \item We first prepare the $K$ component functions $\bm{h}_{k}$ using Proposition \ref{prop: CPWA component circuit}. For $k=1,\ldots,K$, we apply the quantum circuits $(\mathcal{Q}_{h_k})_{k=1,\ldots,K}$ of Proposition \ref{prop: CPWA component circuit} (with $I,d,n_1,m_1,n_2,m_2 \leftarrow I_k,d,n_1,m_1,n_2,m_2$, $a_{l,j} \leftarrow a_{k,l,j}$, and $b_{l} \leftarrow b_{k,l}$ in the notation of Proposition \ref{prop: CPWA component circuit}) followed by an application of the permutation circuit $\Tau_\pi$ from Lemma \ref{lemma: permutation circuit}, where we compute
    \begin{equation}
        \begin{split}
            &\ket{i_1}_{n_1+m_1}\cdots\ket{i_d}_{n_1+m_1}\ket{0}_{q_1+\cdots+q_K}\ket{0}_{2K(n+m+d+5)} \\
            &= \ket{i_1}_{n_1+m_1}\cdots\ket{i_d}_{n_1+m_1}\ket{0}_{q_1}\ket{0}_{n+m+d}\ket{0}_{q_2}\ket{0}_{n+m+d}\cdots\ket{0}_{q_K}\ket{0}_{n+m+d}\ket{0}_{K(n+m+d+10)}\\
            &\xmapsto{\mathcal{Q}_{\bm{h}_1}}\ket{i_1}_{n_1+m_1}\cdots\ket{i_d}_{n_1+m_1}\ket{\anc}_{q_1}\ket{\bm{h}_1(\bm{i})}_{n+m+d}\ket{0}_{q_2}\ket{0}_{n+m+d}\cdots\ket{0}_{q_K}\ket{0}_{n+m+d}\ket{0}_{K(n+m+d+10)}\\
            &\xmapsto{\Tau_{\pi_1}}\ket{i_1}_{n_1+m_1}\cdots\ket{i_d}_{n_1+m_1}\ket{0}_{q_2}\ket{0}_{n+m+d}\ket{\anc}_{q_1}\ket{\bm{h}_1(\bm{i})}_{n+m+d}\cdots\ket{0}_{q_K}\ket{0}_{n+m+d}\ket{0}_{K(n+m+d+10)}\\
            &\xmapsto{\mathcal{Q}_{\bm{h}_2}}\ket{i_1}_{n_1+m_1}\cdots\ket{i_d}_{n_1+m_1}\ket{\anc}_{q_2}\ket{\bm{h}_2(\bm{i})}_{n+m+d}\ket{\anc}_{q_1}\ket{\bm{h}_1(\bm{i})}_{n+m+d}\cdots\ket{0}_{q_K}\ket{0}_{n+m+d}\ket{0}_{K(n+m+d+10)}\\
            &\hspace{100pt}\vdots \hspace{100pt} \vdots\\ &\xmapsto{\mathcal{Q}_{\bm{h}_K}}\ket{i_1}_{n_1+m_1}\cdots\ket{i_d}_{n_1+m_1}\ket{\anc}_{q_K}\ket{\bm{h}_K(\bm{i})}_{n+m+d}\cdots\ket{\anc}_{q_1}\ket{\bm{h}_{1}(\bm{i})}_{n+m+d}\ket{0}_{K(n+m+d+10)}\\
            &\xmapsto{\Tau_{\pi_K}}\ket{i_1}_{n_1+m_1}\cdots\ket{i_d}_{n_1+m_1}\ket{\anc}_{q_1+\cdots+q_K}\ket{\bm{h}_1(\bm{i})}_{n+m+d}\ket{0}_2\ket{0}_{n+m+d+2}\ket{0}_5\\
            &\hspace{100pt} \cdot\ket{\bm{h}_2(\bm{i})}_{n+m+d}\ket{0}_2\ket{0}_{n+m+d+2}\ket{0}_5\cdots\ket{\bm{h}_K(\bm{i})}_{n+m+d}\ket{0}_2\ket{0}_{n+m+d+2}\ket{0}_5\ket{0}_{K} 
        \end{split}
    \end{equation}
    In this step, the number of elementary gates used is at most 
    \begin{equation}
        K\cdot 2N^2 + \sum_{k=1}^K 10651 I_k^3d^3(n+m+1)^3
    \end{equation}
\item Next, we load the strings $\overline{\xi}_k \in \mathbb{F}_{2,0}$ into the qubits $\ket{0}_2$ for $k=1,\ldots,K$ using  Pauli $X$ gates
\begin{equation}
\begin{split}
    &\ket{i_1}_{n_1+m_1}\cdots\ket{i_d}_{n_1+m_1}\ket{\anc}_{q_1+\cdots+q_K}\ket{\bm{h}_1(\bm{i})}_{n+m+d}\ket{0}_2\ket{0}_{n+m+d+2}\ket{0}_5\\
    &\hspace{100pt} \cdot\ket{\bm{h}_2(\bm{i})}_{n+m+d}\ket{0}_2\ket{0}_{n+m+d+2}\ket{0}_5\cdots\ket{\bm{h}_K(\bm{i})}_{n+m+d}\ket{0}_2\ket{0}_{n+m+d+2}\ket{0}_5\ket{0}_{K}\\
    &\xmapsto{}\ket{i_1}_{n_1+m_1}\cdots\ket{i_d}_{n_1+m_1}\ket{\anc}_{q_1+\cdots+q_K}\ket{\bm{h}_1(\bm{i})}_{n+m+d}\ket{\overline{\xi}_1}_2\ket{0}_{n+m+d+2}\ket{0}_5\\
    &\hspace{100pt} \cdot\ket{\bm{h}_2(\bm{i})}_{n+m+d}\ket{\overline{\xi}_2}_2\ket{0}_{n+m+d+2}\ket{0}_5\cdots\ket{\bm{h}_K(\bm{i})}_{n+m+d}\ket{\overline{\xi}_K}_2\ket{0}_{n+m+d+2}\ket{0}_5\ket{0}_{K}.
\end{split}
\end{equation}
The number of elementary gates used in this step is at most
\begin{equation}
    2K.
\end{equation}
    \item We perform $K$ multiplication using the quantum circuit $\widetilde{\mathcal{Q}}_{(\times)}^{(k)} = \widetilde{\mathcal{Q}}_{(\times)}$ from Corollary \ref{corollary: mult circuit} (with $n_1 \leftarrow n+d$, $m_1 \leftarrow m$, $n_2 \leftarrow 2$, $m_2 \leftarrow 0$, $a \leftarrow \bm{h}_k(\bm{i})$, and $b \leftarrow \overline{\xi}_k$ in the notation of Corollary \ref{corollary: mult circuit} for $k=1,\ldots,K$)
    \begin{equation}
        \begin{split}
            &\ket{i_1}_{n_1+m_1}\cdots\ket{i_d}_{n_1+m_1}\ket{\anc}_{q_1+\cdots+q_K}\ket{\bm{h}_1(\bm{i})}_{n+m+d}\ket{\overline{\xi}_1}_2\ket{0}_{n+m+d+2}\ket{0}_5\\
            &\hspace{100pt} \cdot\ket{\bm{h}_2(\bm{i})}_{n+m+d}\ket{\overline{\xi}_2}_2\ket{0}_{n+m+d+2}\ket{0}_5\cdots\ket{\bm{h}_K(\bm{i})}_{n+m+d}\ket{\overline{\xi}_K}_2\ket{0}_{n+m+d+2}\ket{0}_5\ket{0}_{K}\\
            &\xmapsto{\widetilde{\mathcal{Q}}_{(\times)}^{(1)}}\ket{i_1}_{n_1+m_1}\cdots\ket{i_d}_{n_1+m_1}\ket{\anc}_{q_1+\cdots+q_K}\ket{\bm{h}_1(\bm{i})}_{n+m+d}\ket{\overline{\xi}_1}_2\ket{\bm{h}_1(\bm{i}) \boxdot \overline{\xi}_1}_{n+m+d+2}\ket{\anc}_5\\
            &\hspace{100pt} \cdot\ket{\bm{h}_2(\bm{i})}_{n+m+d}\ket{\overline{\xi}_2}_2\ket{0}_{n+m+d+2}\ket{0}_5\cdots\ket{\bm{h}_K(\bm{i})}_{n+m+d}\ket{\overline{\xi}_K}_2\ket{0}_{n+m+d+2}\ket{0}_5\ket{0}_{K}\\
            &\hspace{100pt}\vdots \hspace{100pt} \vdots\\
            &\xmapsto{\widetilde{\mathcal{Q}}_{(\times)}^{(K)}}\ket{i_1}_{n_1+m_1}\cdots\ket{i_d}_{n_1+m_1}\ket{\anc}_{q_1+\cdots+q_K}\ket{\bm{h}_1(\bm{i})}_{n+m+d}\ket{\overline{\xi}_1}_2\ket{\bm{h}_1(\bm{i}) \boxdot \overline{\xi}_1}_{n+m+d+2}\ket{\anc}_5\\
            &\hspace{100pt} \cdots\ket{\bm{h}_K(\bm{i})}_{n+m+d}\ket{\overline{\xi}_K}_2\ket{\bm{h}_K(\bm{i}) \boxdot \overline{\xi}_K}_{n+m+d+2}\ket{\anc}_5\ket{0}_{K}.\\
        \end{split}
    \end{equation}
    The number of elementary gates used in this step is at most
    \begin{equation}
        K\cdot 61(n+m+d+3)^2.
    \end{equation}
    \item We use the permutation circuit $\Tau_\pi$ to reorder the qubits for addition, where 
    \begin{equation}
        \begin{split}
            &\ket{i_1}_{n_1+m_1}\cdots\ket{i_d}_{n_1+m_1}\ket{\anc}_{q_1+\cdots+q_K}\ket{\bm{h}_1(\bm{i})}_{n+m+d}\ket{\overline{\xi}_1}_2\ket{\bm{h}_1(\bm{i}) \boxdot \overline{\xi}_1}_{n+m+d+2}\ket{\anc}_5\\
            &\hspace{100pt} \cdots\ket{\bm{h}_K(\bm{i})}_{n+m+d}\ket{\overline{\xi}_K}_2\ket{\bm{h}_K(\bm{i}) \boxdot \overline{\xi}_K}_{n+m+d+2}\ket{\anc}_5\ket{0}_{K}\\
            &\xmapsto{\Tau_\pi}\ket{i_1}_{n_1+m_1}\cdots\ket{i_d}_{n_1+m_1}\ket{\anc}_{q_1+\cdots+q_K + K(n+m+d+7)}\ket{\bm{h}_1(\bm{i}) \boxdot \overline{\xi}_1}_{n+m+d+2}\ket{\bm{h}_2(\bm{i}) \boxdot \overline{\xi}_2}_{n+m+d+2}\ket{0}\\
            &\hspace{100pt} \ket{\bm{h}_3(\bm{i}) \boxdot \overline{\xi}_3}_{n+m+d+2}\ket{0}\cdots\ket{\bm{h}_K(\bm{i}) \boxdot \overline{\xi}_K}_{n+m+d+2}\ket{0}\ket{0}.
        \end{split}
    \end{equation}
    The number of elementary gates used in this step is at most 
    \begin{equation}
        2N^2.
    \end{equation}
    \item We perform the following addition inductively on the sums for $k=1,\ldots,K-1$
    \begin{equation}
        \boxplus: \mathbb{F}_{n+d+2+k-1,m} \times \mathbb{F}_{n+d+2,m} \to \mathbb{F}_{n+d+2+k,m}, \quad (\bigboxplus_{l=1}^k (\bm{h}_{l}(\bm{i}) \boxdot \overline{\xi}_l),\bm{h}_{k+1}(\bm{i}) \boxdot \overline{\xi}_{k+1}) \mapsto \bigboxplus_{l=1}^{k+1}(\bm{h}_{l}(\bm{i}) \boxdot \overline{\xi}_l),
    \end{equation}
    where we apply the quantum circuit $\mathcal{Q}_{(+)}$ from Corollary \ref{corollary: add circuit} inductively for $k=1,\ldots,K-1$ (with $n_1 \leftarrow n+d+2+k-1$, $m_1 \leftarrow m$, $n_2 \leftarrow n+d+2$, $m_2 \leftarrow m$ in the notation of Corollary \ref{corollary: add circuit}) so that
    \begin{equation}
        \begin{split}
            &\ket{i_1}_{n_1+m_1}\cdots\ket{i_d}_{n_1+m_1}\ket{\anc}_{q_1+\cdots+q_K + K(n+m+d+7)}\ket{\bm{h}_1(\bm{i}) \boxdot \overline{\xi}_1}_{n+m+d+2}\ket{\bm{h}_2(\bm{i}) \boxdot \overline{\xi}_2}_{n+m+d+2}\ket{0}\\
            &\hspace{100pt} \ket{\bm{h}_3(\bm{i}) \boxdot \overline{\xi}_3}_{n+m+d+2}\ket{0}\cdots\ket{\bm{h}_K(\bm{i}) \boxdot \overline{\xi}_K}_{n+m+d+2}\ket{0}\ket{0}\\
            &\xmapsto{\mathcal{Q}_{(+)}^{(1)}}\ket{i_1}_{n_1+m_1}\cdots\ket{i_d}_{n_1+m_1}\ket{\anc}_{q_1+\cdots+q_K + K(n+m+d+7)}\ket{\bm{h}_2(\bm{i}) \boxdot \overline{\xi}_2}_{n+m+d+2}\ket{\bigboxplus_{l=1}^2 (\bm{h}_{l}(\bm{i}) \boxdot \overline{\xi}_l)}_{n+m+d+3}\\
            &\hspace{100pt} \ket{\bm{h}_3(\bm{i}) \boxdot \overline{\xi}_3}_{n+m+d+2}\ket{0}\cdots\ket{\bm{h}_K(\bm{i}) \boxdot \overline{\xi}_K}_{n+m+d+2}\ket{0}\ket{0}\\
            &\hspace{100pt}\vdots \hspace{100pt} \vdots\\
            &\xmapsto{\mathcal{Q}_{(+)}^{(K)}}\ket{i_1}_{n_1+m_1}\cdots\ket{i_d}_{n_1+m_1}\ket{\anc}_{q_1+\cdots+q_K + K(n+m+d+7)}\ket{\bm{h}_2(\bm{i}) \boxdot \overline{\xi}_2}_{n+m+d+2}\ket{\bm{h}_3(\bm{i}) \boxdot \overline{\xi}_3}_{n+m+d+2}\\
            &\hspace{100pt} \cdots\ket{\bm{h}_K(\bm{i}) \boxdot \overline{\xi}_K}_{n+m+d+2}\ket{\bigboxplus_{l=1}^K (\bm{h}_{l}(\bm{i}) \boxdot \overline{\xi}_l)}_{n+m+d+K+1}\ket{0}\\
            &=: \ket{i_1}_{n_1+m_1}\cdots\ket{i_d}_{n_1+m_1}\ket{\anc}_{q_1+\cdots+q_k + K(n+m+d+7)+(K-1)(n+m+d+2)}\ket{\bm{h}(\bm{i})}_{n+m+d+K+1}\ket{0}.
        \end{split}
    \end{equation}
    The  number of elementary gates used in this step is at most
    \begin{equation}
        \sum_{k=1}^{K-1} 29[(n+m+d+2+k-1) + (n+m+d+2) +1]^2
    \end{equation}
    \item Lastly, we apply the quantum circuit $\mathcal{R}_{f}$ from Lemma \ref{lemma: quantum linear rotation} (with $a_1 \leftarrow s$, $a_0 \leftarrow \pi/2$, $n \leftarrow n+d+K+1$, $m\leftarrow m$ in the notation of Lemma \ref{lemma: quantum linear rotation})
    \begin{equation}
        \begin{split}
           &\ket{i_1}_{n_1+m_1}\cdots\ket{i_d}_{n_1+m_1}\ket{\anc}_{q_1+\cdots+q_k + K(n+m+d+7)+(K-1)(n+m+d+2)}\ket{\bm{h}(\bm{i})}_{n+m+d+K+1}\ket{0}\\
           &\xmapsto{\mathcal{R}_f}\ket{i_1}_{n_1+m_1}\cdots\ket{i_d}_{n_1+m_1}\ket{\anc}_{q_1+\cdots+q_k + K(n+m+d+7)+(K-1)(n+m+d+2)}\ket{\bm{h}(\bm{i})}_{n+m+d+K+1}\\
           &\hspace{100pt}\left[ \cos(\bar{f}(\bm{h}(\bm{i}))/2)\ket{0} + \sin(\bar{f}(\bm{h}(\bm{i}))/2)\ket{1} \right]\\
           &=: \ket{i_1}_{n_1+m_1}\cdots\ket{i_d}_{n_1+m_1}\ket{\anc}_{q_1+\cdots+q_K + 2K(n+m+d+5)-1}\left[ \cos(\bar{f}(\bm{h}(\bm{i}))/2)\ket{0} + \sin(\bar{f}(\bm{h}(\bm{i}))/2)\ket{1} \right].
        \end{split}
    \end{equation}
This is the desired state \eqref{eqn: grover state}. The number of elementary gates used in this step is at most
\begin{equation}
    13(n+d+K+1+m+1)^3.
\end{equation}
\end{enumerate}

Note that $p\leq 4d(n+m+1)$ (c.f. \eqref{eqn: bound on p}), hence, for each $k=1,\ldots,K$, 
\begin{equation}
    \begin{split}
        q_k &= I_k(n+m+d+p) + (I_k-2)(n+m+d) + 5(I_k-1)\\
        &\leq I_k(n+m+d+ 4d(n+m+1)) +  (I_k-2)(n+m+d) + 5(I_k-1)\\
        &\leq I_k(d(n+m+1) + 4d(n+m+1)) + 6I_kd(n+m+1)\\
        &= 11I_kd(n+m+1).
    \end{split}
\end{equation}
Hence, 
\begin{equation}\label{eqn: bound on N}
    \begin{split}
        N &= d(n_1+m_1) + \sum_{k=1}^K q_k +  2K(n+m+d+5)\\
        &\leq d(n+m+1) + K\cdot\max_{k=1,\ldots,K}\{I_k\}\cdot 11d(n+m+1) + 12Kd(n+m+1)\\
        &\leq 24K\cdot\max_{k=1,\ldots,K}\{I_k\}\cdot d(n+m+1).
    \end{split}
\end{equation}
Thus, the total number of elementary gates used is at most
\begin{equation}
    \begin{split}
        &K\cdot2N^2 + \sum_{k=1}^K  10651 I_k^3d^3(n+m+1)^3 + 2K \\
        &\qquad + K\cdot 61(n+m+d+3)^2 + 2N^2 \\
        &\qquad + \sum_{k=1}^{K-1}29[(n+m+d+2+k-1) + (n+m+d+2) + 1]^2 + 13(n+d+K+1+m+1)^3\\
        &\leq 2\cdot 24^2K^3\left(\max_{k=1,\ldots,K}\{I_k\}\right)^2 d^2 (n+m+1)^2 + 10651 K \left(\max_{k=1,\ldots,K}\{I_k\}\right)^3 d^3(n+m+1)^3 + 2K\\
        &\qquad + 61\cdot4^2Kd^2(n+m+1)^2 +  2\cdot 24^2K^3\left(\max_{k=1,\ldots,K}\{I_k\}\right)^2 d^2 (n+m+1)^2 \\
        &\qquad + 29K[3Kd(n+m+1)+ 3d(n+m+1) + 1]^2 + 13(4Kd(n+m+1))^3\\
        &\leq 2\cdot 24^2K^3\left(\max_{k=1,\ldots,K}\{I_k\}\right)^2 d^2 (n+m+1)^2 + 10651 K \left(\max_{k=1,\ldots,K}\{I_k\}\right)^3 d^3(n+m+1)^3 + 2K\\
        &\qquad + 61\cdot4^2Kd^2(n+m+1)^2 +  2\cdot 24^2K^3\left(\max_{k=1,\ldots,K}\{I_k\}\right)^2 d^2 (n+m+1)^2 \\
        &\qquad + 29K[7Kd(n+m+1)]^2 + 13\cdot4^3K^3d^3(n+m+1)^3\\
        &\leq (2 \cdot 24^2 + 10651 + 2 + 61 \cdot 4^2 + 2 \cdot 24^2 + 29\cdot7^2 + 13 \cdot 4^3)K^3\left(\max_{k=1,\ldots,K}\{I_k\}\right)^3d^3(n+m+1)^3\\
        &= 16186 K^3\left(\max_{k=1,\ldots,K}\{I_k\}\right)^3d^3(n+m+1)^3.
    \end{split}
\end{equation}
\end{proof}

\section{Error analysis}\label{section: error estimates}
In this section, we provide the detailed error analysis of the steps of Algorithm~\ref{quantum algorithm} outlined in Section \ref{section: outline of algorithm}. We begin the error analysis with a few basic lemmas.
\begin{lemma}[Lipschitz constant of continuous piecewise affine functions]\label{lemma: lipschitz constant}
	Let $h:\R^d_+ \to \R$ be a continuous piecewise affine function given by \eqref{eqn: CPWA payoff}. Let Assumption \ref{assumption: CPWA} hold with corresponding constant $C_2 \in [1,\infty)$. Then, $h:\R^d_+ \to \R$ is Lipschitz continuous with Lipschitz constant 
	\begin{equation}
		L := \sum_{k=1}^{K} \max_{1 \leq l \leq I_k}\{ \lVert \bm{a}_{k,l}  \rVert_\infty\}\sqrt{d} \leq C_2^2 d^{\frac{3}{2}},\label{eqn: Lipchitz constant}
	\end{equation}
    i.e.
    \begin{equation}\label{eqn: lin growth1}
        \forall \bm{x},\bm{y} \in \R^d_+: \quad \vert h(\bm{x}) - h(\bm{y}) \vert \leq C_2^2 d^{\frac{3}{2}} \Vert \bm{x} - \bm{y} \Vert_{2}.
    \end{equation}
	\begin{proof}
		Following the proof of Lemma 3.3 in \cite{neufeld2020modelfree}, the continuous piecewise affine function $h:\R^d_+ \to \R$ admits the following representation:
		\begin{equation}
			h(\bm{x}) = \begin{cases} \bm{a}_1' \cdot \bm{x} + b_1', \quad \text{if $\bm{x} \in \Omega_1$}, \\ \qquad \vdots  \\
				\bm{a}_J' \cdot \bm{x} + b_J', \quad \text{if $\bm{x} \in \Omega_J$},    \end{cases}
		\end{equation}
		where $J := \prod_{k=1}^K I_k \in \N$ and the coefficients $\{\bm{a}_j',b_j': j=1,\ldots,J\}$ are of the form
        \begin{equation}\label{eqn: proof of lipschitz constant1}
            \bm{a}_j' := \sum_{k=1}^K \xi_k \bm{a}_{k,l_k^*}, \quad b_{j}' := \sum_{k=1}^K \xi_k b_{k,l_k^*},
        \end{equation}
        for some $l_k^* \in \{1,2,\ldots, I_k\}$ (\textit{the specific choice of index $l_k^*$ can be found in the proof of Lemma 3.3 in \cite{neufeld2020modelfree}}), and where $\Omega_1,\ldots,\Omega_J$ are  polyhedrons whose union is $\R_+^d$. Note that some of these sets $\Omega_j$ can be empty. We claim that 
		\begin{equation}\label{eqn: proof of lipschitz constant2}
			\forall \bm{x},\bm{y} \in \R_+^d: \quad \lvert h(\bm{x}) - h(\bm{y}) \rvert \leq \max_{1\leq j \leq J} \{ \lVert \bm{a}_j' \rVert_\infty \} \cdot  \sqrt{d}\lVert \bm{x} - \bm{y} \rVert_2.
		\end{equation}
		Let $\bm{x},\bm{y} \in \R_+^d$ be fixed. Consider the line segment from $\bm{x}$ to $\bm{y}$ given by the set $\Gamma := \{\gamma(t) := \bm{x} + t(\bm{y}-\bm{x}): t \in [0,1]\}$. If the line segment $\Gamma$ lies entirely in one of polyhedron $\Omega_{j_*}$, then by linearity of $h$ in $\Omega_{j_*}$ and the H\"older's inequality, it follows that 
		\begin{equation} 
			\lvert h(\bm{x}) - h(\bm{y}) \rvert = \lvert \bm{a}_{j_*}' \cdot (\bm{x}-\bm{y}) \rvert \leq \lVert \bm{a}_{j_*}'\rVert_\infty  \lVert \bm{x} - \bm{y} \rVert_1 \leq \max_{1\leq j \leq J} \{ \lVert \bm{a}_j' \rVert_\infty \} \cdot  \sqrt{d} \lVert \bm{x} - \bm{y} \rVert_2. 
			\end{equation}
		In the general case, the line segment $\Gamma$ may be contained in some $n\geq 2$ polyhedrons $\Omega_{j_1},\ldots,\Omega_{j_n}$, such that there are $n+1$ points $\{ \gamma(t_0),\gamma(t_1),\ldots,\gamma(t_n)\} \subset \Gamma$ with a partition $\{0 =: t_0 < t_1 < \cdots < t_n :=1\}$, satisfying $\gamma(t_m) \in \Omega_{j_m} \cap \Omega_{j_{m+1}}$ for $m=1,\ldots,n-1$. Again by the same argument, it holds for all $m=1,\ldots,n-1$ that 
		\begin{equation} 
			\lvert h(\gamma(t_m)) - h(\gamma(t_{m-1})) \rvert = \lvert \bm{a}_{j_m}' \cdot (\gamma(t_{m})-\gamma(t_{m-1})) \rvert \leq \lVert \bm{a}_{j_m}'\rVert_\infty  (t_m - t_{m-1})\sqrt{d}\lVert \bm{x} - \bm{y} \rVert_2. 
			\end{equation}
		Hence,  summing over $m$, we have 
		\begin{equation}
			\begin{split}
			\lvert h(\bm{x}) - h(\bm{y}) \rvert &\leq \sum_{m=1}^{n} \lVert \bm{a}_{j_m}' \rVert_\infty (t_j - t_{j-1}) \sqrt{d} \lVert \bm{x} - \bm{y} \rVert_2 \\
            &\leq \max_{1\leq j \leq J} \{ \lVert \bm{a}_j' \rVert_\infty \} \cdot \sqrt{d}\lVert \bm{x} - \bm{y} \rVert_2.
		\end{split}
		\end{equation}
        Thus, the claim \eqref{eqn: proof of lipschitz constant2} is proven. Next, by \eqref{eqn: proof of lipschitz constant1}, it follows that 
        \begin{equation}
            \max_{1\leq j \leq J} \{ \lVert \bm{a}_j' \rVert_\infty \} \leq \sum_{k=1}^K \max_{1\leq l \leq I_k} \{\Vert \bm{a}_{k,l}\Vert_\infty\}.
        \end{equation}
        This, \eqref{eqn: proof of lipschitz constant2} and Assumption \ref{assumption: CPWA} concludes the lemma. 
	\end{proof}
\end{lemma}

\begin{lemma}[Linear growth of continuous piecewise affine functions]\label{lemma: lin growth}
	Let $h:\R^d_+ \to \R$ be a continuous piecewise affine function given by \eqref{eqn: CPWA payoff}. Let Assumption \ref{assumption: CPWA} hold with corresponding constant $C_2 \in [1,\infty)$. Then, for all $\bm{x} \in \R_+^d$, it holds that
	\begin{equation}\label{eqn: lin growth2}
		\lvert h(\bm{x}) \rvert \leq C_2^2d^{\frac{3}{2}}(1 + \lVert \bm{x} \rVert_2). 
	\end{equation}
	
	\begin{proof} Note that Lemma \ref{lemma: lipschitz constant} and an application of the triangle equality shows that
        \begin{equation}\label{eqn: proof of lin growth}
            \lvert h(\bm{x}) \rvert \leq  \lvert h(\bm{x}) - h(0) \rvert + \lvert h(0) \rvert \leq C_2^2d^{\frac{3}{2}} \lVert \bm{x} \rVert_2 + \vert h(0) \vert .
        \end{equation}
		Moreover, by \eqref{eqn: CPWA payoff} and Assumption \ref{assumption: CPWA}, we have for all $k=1,\ldots,K$ that 
        \begin{equation}
            \vert h(0)\vert  \leq K \cdot \max\{\vert b_{k,l}\vert : l=1,\ldots,I_k\} \leq C_2^2d.
        \end{equation}
        This and \eqref{eqn: proof of lin growth} implies \eqref{eqn: lin growth2}.
	\end{proof}
\end{lemma}

\subsection{Step 1: Truncation error bounds}
\begin{lemma}\label{lemma: bounds on log-normal 1}
	 Let $Y$ be a log-normal random variable with parameters $\mu \in \R$ and  $\sigma^2 \in (0,\infty)$. Let $Z$ be a standard normal random variable. Then it holds that
	\begin{enumerate}[(i)]
		\item for all $k \in \N$,
		\begin{equation}
			\E[Y^k] = e^{k\mu + \frac{k^2\sigma^2}{2}},
		\end{equation}
		\item for all $y \in [0,\infty)$ that
		\begin{equation}
			\Prob(Z \geq y)  \leq \frac{1}{2}e^{-\frac{y^2}{2}}.
		\end{equation}
	\end{enumerate}
	\begin{proof}
		Item (i) is proven in \cite[Chapter 2.3]{aitchisonbrown}, and Item (ii) is proved in \cite[Eq. 6]{Tailbounds}. 
	\end{proof}
\end{lemma}

\begin{proposition}[Truncation error]\label{prop: truncation error}
	Let $\eps \in (0,1) $,  $d \in \N$, $r,T \in (0,\infty)$, and $(t,\bm{x}) \in [0,T) \times \R_+^d $. Let $h:\R^d_+ \to \R$ be a continuous piecewise affine function given by \eqref{eqn: CPWA payoff}. Let Assumption \ref{assumption: cov matrix} and Assumption \ref{assumption: CPWA} hold with respective constants $C_1,C_2 \in [1,\infty)$. Let $p(\cdot,T;\bm{x},t):\R_+^d \to \R_+$ be the log-normal transition density function given by \eqref{eqn: density formula}. Let $M_{d,\eps} \in [1,\infty)$ satisfy
	\begin{equation}\label{eqn: def M}
		M_{d,\eps} = 2C_2^2 d^{\frac{5}{2}}\eps^{-1} e^{4C_1^2T^2}e^{2rT}\max_{i=1,\ldots,d}\{1,x_i^2\}.
	\end{equation}
	Let $u: [0,T] \times \R_+^d \to \R$ be the option price given by 
	\begin{equation}\label{eqn: prop_truncation sol}
		u(t,\bm{x}) := e^{-r(T-t)}\int_{\R_+^d}h(\bm{y})p(\bm{y},T;\bm{x},t)\,d\bm{y},
	\end{equation}
	and for every $M\geq M_{d,\eps}$, let $\bar{u}_{M,t,\bm{x}} \in \R$ be the truncated solution given by 
	\begin{equation} \label{eqn: prop_truncation truncated sol}
		\bar{u}_{M,t,\bm{x}} := e^{-r(T-t)}\int_{[0,M]^d}h(\bm{y})p(\bm{y},T;\bm{x},t)\, d\bm{y}.
	\end{equation}
	Then, the truncation solution satisfies the following estimate
	\begin{equation}
		\vert u(t,\bm{x}) - \bar{u}_{M,t,\bm{x}} \vert  \leq \eps. \label{eqn: truncation error}
	\end{equation}
	
	\begin{proof}
		Let $\bm{\Sigma} \equiv \bm{\Sigma_d} := (T-t)\bm{C}_d \in \R^{d \times d}$, and $\bm{\mu} \equiv \bm{\mu_d} \in \R^d$ denote the log-covariance and log-mean parameters for the multivariate log-normal random variable $\bm{Y}$ with the probability density function $p(\cdot,T;\bm{x},t)$  given by Lemma~\ref{lemma: density}. 
		Using Lemma \ref{lemma: lin growth}, \eqref{eqn: prop_truncation sol}, \eqref{eqn: prop_truncation truncated sol}, the fact that $e^{-r(T-t)} \leq 1$, and Cauchy-Schwarz inequality, 
		\begin{equation}
			\begin{split}
				\lvert u(t,\bm{x}) - \bar{u}_{M,t,\bm{x}} \rvert &= e^{-r(T-t)}\left\lvert \int_{\R_+^d}h(\bm{y})p(\bm{y},T;\bm{x},t) \, d\bm{y} - \int_{[0,M]^d}h(\bm{y})p(\bm{y},T;\bm{x},t)\,d\bm{y} \right\rvert \\
				&\leq C_2^2 d^{\frac{3}{2}} \E[(1+\Vert \bm{Y} \Vert_2)^2]^{1/2}\Prob(\bm{Y} \notin [0,M]^d)^{1/2},
			\end{split}
		\end{equation}
		where 
		\begin{equation}
			\Prob(\bm{Y} \notin [0,M]^d) = \int_{\R_+^d\setminus [0,M]^d} p(\bm{y},T;\bm{x},t)\,d\bm{y}.
			\end{equation}
		Let $\bm{X} = (X_1,\ldots,X_d)\sim \mathcal{N}(\bm{\mu},\bm{\Sigma})$ be the multivariate Gaussian random variable and recall that $X_i = \ln(Y_i)$ for $i=1,\ldots,d$. Using Lemma \ref{lemma: bounds on log-normal 1} (i), we have 
		\begin{equation}
			\E[\Vert \bm{Y} \Vert_2^2] = \sum_{i=1}^d \E[Y_i^2] = \sum_{i=1}^d \E[e^{2X_i}] = \sum_{i=1}^d e^{2\mu_i + 2\sigma_{ii}^2},
		\end{equation}
		where $\mu_i = \ln(x_i) + (r-\frac{1}{2}\sigma_{ii}^2)(T-t)$ and $\sigma_{ii} = \bm{\Sigma}_{i,i}$. We use the bound $e^{2\mu_i} \leq \max_{i=1,\ldots,d}\{1,x_i^2\}e^{2rT}$ and $e^{2\sigma_{ii}^2}\leq e^{2C_1^2T^2}$ by Assumption \ref{assumption: cov matrix} and conclude that 
		\begin{equation}\label{second-mom-multi-log-normal}
			\E[\Vert \bm{Y}\Vert_2^2] \leq d e^{2C_1^2T^2}e^{2rT}\max_{i=1,\ldots,d}\{1,x_i^2\}.
		\end{equation}
		By Cauchy-Schwarz inequality, we have $(1+ \Vert \bm{Y} \Vert_2)^2 \leq 2(1+ \Vert \bm{Y} \Vert_2^2)$. Also, note that $\max_i\{1,x_i^2\}^{1/2} = \max_i\{1,\vert x_i \vert\}$ and  $1 \leq d e^{2C_1^2T^2}e^{2rT}\max_{i}\{1,x_i^2\}$. Combining the above bounds, we arrive at 
		\begin{equation}
			\E[(1+\Vert \bm{Y} \Vert_2)^2]^{1/2} \leq 2d^{1/2} e^{C_1^2T^2} e^{rT} \max_{i=1,\ldots,d}\{1,\vert x_i\vert \}.
		\end{equation}
		Hence, we have 
		\begin{equation}\label{eqn: proof truncation1}
			\vert u(t,\bm{x}) - \bar{u}_{M,t,\bm{x}} \vert \leq 2C_2^2 d^{2} e^{C_1^2T^2}e^{rT}\max_{i=1,\ldots,d}\{1,\vert x_i\vert \}\cdot \Prob(\bm{Y} \notin [0,M]^d)^{1/2}.
		\end{equation}
		Moreover, for all $i=1,\ldots,d$, we have the inclusions
		\begin{equation}
			\{Y_i \leq M\} = \{X_i \leq \ln (M) \} \supseteq \{\vert X_i - \mu_i \vert  \leq \ln (M) - \mu_i\}.
		\end{equation}
        Note that since $M \geq M_{d,\eps}\geq e^{\mu_i}$, we have that $\ln (M) \geq \mu_i$. 
		By Sidak's correlation inequality \cite[Corollary 1]{sidak}, we have
		\begin{equation}
			\Prob\left(\bigcap_{i=1}^d\{\vert X_i - \mu_i \vert  \leq \ln (M) - \mu_i\}\right) \geq \prod_{i=1}^d \Prob(\{\vert X_i - \mu_i \vert  \leq \ln (M) - \mu_i\}),
		\end{equation}
		and hence, 
		\begin{equation}
			\Prob(\bm{Y} \notin [0,M]^d) = 1 - \Prob\left(\bigcap_{i=1}^d \{Y_i \leq M\}\right)\leq 1 -  \prod_{i=1}^d \Prob(\{\vert X_i - \mu_i \vert  \leq \ln (M) - \mu_i\}).
		\end{equation}
		Denote by $Z$  the standard normal random variable. Using Lemma \ref{lemma: bounds on log-normal 1} (ii) and Assumption \ref{assumption: cov matrix}, we have 
		\begin{equation}
			\Prob(\{\vert X_i - \mu_i \vert  \leq \ln (M) - \mu_i\}) = 1 - 2\Prob(Z \geq \frac{\ln(M) - \mu_i}{\sigma_{ii}})\geq 1 - e^{-\frac{(\ln(M)-\mu_i)^2}{2\sigma_{ii}^2}} \geq 1-e^{-\frac{(\ln(M)-\mu_i)^2}{2C_1^2T^2}}.
		\end{equation}
		Moreover, using $\ln(M)>0$ and $M \geq M_{d,\eps}\geq \max_{i=1,\ldots,d}\{1,x_i^2\}e^{2rT}e^{4C_1^2T^2} \geq e^{2\mu_i}e^{4C_1^2T^2}$, we have 
		\begin{equation}
            \begin{split}
                & \ln(M) \geq 2\mu_i + 4C_1^2T^2\\
                &\Longleftrightarrow (\ln(M))^2 \geq (2\mu_i + 4C_1^2T^2)\ln(M)\\
                &\Longleftrightarrow (\ln(M))^2-2\mu_i\ln(M) +\mu_i^2\geq 4C_1^2T^2\ln(M) + \mu_i^2 \\ 
                &\Longrightarrow (\ln(M)-\mu_i)^2 \geq 4C_1^2T^2\ln(M)\\
                &\Longleftrightarrow-\frac{(\ln(M)-\mu_i)^2}{2C_1^2T^2} \leq -2 \ln(M)\\
                &\Longleftrightarrow e^{-\frac{(\ln(M)-\mu_i)^2}{2C_1^2T^2}}\leq M^{-2}.
            \end{split}
		\end{equation}
		Using Bernoulli's inequality, that is  $(1+z)^d \geq 1+dz$ for any $z \in [-1,\infty)$, and the fact that $-M^{-2}\in [-1,0)$, we have 
		\begin{equation}\label{eqn: proof truncation2}
			\Prob(\bm{Y} \notin [0,M]^d) \leq 1-(1-M^{-2})^d \leq dM^{-2}.
		\end{equation} 
		Thus, by \eqref{eqn: proof truncation1}, \eqref{eqn: proof truncation2}, and \eqref{eqn: def M}, we conclude that 
		\begin{equation}
			\vert u(t,\bm{x}) - \bar{u}_{M,t,\bm{x}} \vert  \leq 2C_2^2 d^{\frac{5}{2}} e^{C_1^2T^2}e^{rT}\max_{i=1,\ldots,d}\{1,\vert x_i\vert \}M^{-1} \leq  2C_2^2 d^{\frac{5}{2}} e^{C_1^2T^2}e^{rT}\max_{i=1,\ldots,d}\{1,\vert x_i\vert \}M_{d,\eps}^{-1} \leq \eps. 
		\end{equation}
	\end{proof}
\end{proposition}

\subsection{Step 2: Quadrature error bounds}
\begin{proposition}[Quadrature error]\label{prop: quadrature error}
	Let $d \in \N$, $r,T\in (0,\infty)$, $n \in \N \cap \{2,3,\ldots\}$, $m \in \N$,  and $(t,\bm{x}) \in [0,T) \times \R_+^d$. Let $M \in [1,\infty)$ be defined by $M := 2^{n-1}$. Let $h:\R^d_+ \to \R$ be the continuous piecewise affine function given by \eqref{eqn: CPWA payoff}. Let Assumption \ref{assumption: CPWA} hold with corresponding constant $C_2 \in [1,\infty)$. Let $p(\cdot,T;\bm{x},t):\R_+^d \to \R_+$ be the log-normal transition density given by \eqref{eqn: density formula}. Let $\widetilde{u}_{n,t,\bm{x}} \in \R$ be the truncated solution given by
	\begin{equation}\label{eqn: quad-prop truncated solution}
		\widetilde{u}_{n,t,\bm{x}} := e^{-r(T-t)}\int_{[0,M]^d}h(\bm{y}) p(\bm{y},T;\bm{x},t)\, d\bm{y},
	\end{equation}
	and let $\widetilde{u}_{n,m,t,\bm{x}} \in \R$ be the truncated quadrature solution given by
	\begin{equation}\label{eqn: quad solution}
		\widetilde{u}_{n,m,t,\bm{x}} := e^{-r(T-t)}\sum_{\bm{j} \in \mathbb{K}_{n,m,+}^d} h(\bm{j})p_{\bm{j},m},
	\end{equation}
	where for $\bm{j} = (j_1,\ldots,j_d)\in \mathbb{K}_{n,m,+}^d$ (c.f Definition \ref{def: two complement}),
	\begin{equation}
		p_{\bm{j},m} := \int_{Q_{\bm{j},m}} p(\bm{y},T;\bm{x},t)\,d\bm{y}, \quad\text{and}\quad Q_{\bm{j},m} := [j_1,j_1+2^{-m}) \times \cdots \times [j_d,j_d+2^{-m}).
	\end{equation}
	Then,
	\begin{equation}
		\lvert \widetilde{u}_{n,t,\bm{x}}  - \widetilde{u}_{n,m,t,\bm{x}} \rvert  \leq C_{2}^2 d^{2} 2^{-m}.
	\end{equation}
	\begin{proof}
		Let $[ \cdot ]_m:\R_+^d \to (2^{-m}\Z)^d$ be a function defined by
		\begin{equation}
			[\bm{y}]_m = \left(\frac{\lfloor 2^m y_1 \rfloor}{2^m},\ldots,\frac{\lfloor 2^m y_d \rfloor}{2^m}\right),
		\end{equation}
		where $\lfloor \cdot \rfloor$ is the floor function. With this function, it holds for every $\bm{j} \in \mathbb{K}_{n,m,+}^d$ that
		\begin{equation}
			\forall \bm{y} \in Q_{\bm{j},m}, \quad [\bm{y}]_m = \bm{j}.
		\end{equation}
		Moreover, since
		\begin{equation}
			[0,M)^d = \bigsqcup_{\bm{j} \in \mathbb{K}_{n,m,+}^d} Q_{\bm{j},m}
		\end{equation}
		is a disjoint union of sets, it follows that 
		\begin{equation}\label{eqn: quad-prop grid solution}
			\widetilde{u}_{n,m,t,\bm{x}} = e^{-r(T-t)}\sum_{\bm{j} \in \mathbb{K}_{n,m,+}^d} \int_{Q_{\bm{j},m}} h([\bm{y}]_m) p(\bm{y},T;\bm{x},t)\,d\bm{y} = e^{-r(T-t)}\int_{[0,M]^d} h([\bm{y}]_m) p(\bm{y},T;\bm{x},t)\,d\bm{y}.
		\end{equation}
		Furthermore, observe that 
		\begin{equation}\label{eqn: quad-prop dyadic approx}
			\forall \bm{y} \in \R_+^d ,\quad \Vert \bm{y} - [\bm{y}]_m \Vert_1 \leq d 2^{-m}.
		\end{equation}
		Hence, by definition of $\widetilde{u}_{n,t,\bm{x}}$ in \eqref{eqn: quad-prop truncated solution}, \eqref{eqn: quad-prop grid solution}, by the fact that $e^{-r(T-t)} \leq 1$, by the Lipschitz continuity of $h(\cdot)$ in Lemma \ref{lemma: lin growth}, by Assumption \ref{assumption: CPWA}, \eqref{eqn: quad-prop dyadic approx}, and by the fact that $p(\bm{y},T;\bm{x},t)$ is a probability density function supported on $\R_+^d$, we conclude that 
		\begin{equation}
			\begin{split}
				\vert \widetilde{u}_{n,t,\bm{x}} - \widetilde{u}_{n,m,t,\bm{x}} \vert &\leq e^{-r(T-t)} \int_{[0,M]^d} \vert h(\bm{y}) - h([\bm{y}]_m) \vert p(\bm{y},T;\bm{x},t)\,d\bm{y}\\
				&\leq \int_{[0,M]^d} C_2^2d \Vert \bm{y} - [\bm{y}]_m\Vert_1 p(\bm{y},T;\bm{x},t)\,d\bm{y} \\
				&\leq C_2^2 d^{2}2^{-m}.
			\end{split}
		\end{equation}
		
	\end{proof}
\end{proposition}

\begin{corollary}[Truncation and quadrature errors]\label{corollary: trun-quad error} 
	Let $\eps \in (0,1)$, $d \in \N$, $r,T \in (0,\infty)$, and $(t,\bm{x}) \in [0,T) \times \R_+^d$. Let $u(t,\bm{x})$ be the option price given by \eqref{eqn: option price formula}. Let $h:\R^d_+ \to \R$ be the continuous piecewise affine function given by \eqref{eqn: CPWA payoff}. Let Assumption \ref{assumption: cov matrix} and Assumption \ref{assumption: CPWA} hold with respective constants $C_1,C_2 \in [1,\infty)$. For every $\eta\in(0,1)$, let $M_{d,\eta} \in [1,\infty)$ be given by \eqref{eqn: def M}. 
	For every $n,m \in \N$ satisfying 
	\begin{equation}\label{eqn: trun-quad def n}
		\begin{split}
			n\geq 1 + \log_2(M_{d,\eps/2} ), 
		\end{split}
	\end{equation}
	\begin{equation}\label{eqn: trun-quad def m}
		m \geq \log_2(C_2^2 d^{2}(\eps/2)^{-1} )
	\end{equation}
	let $\widetilde{u}_{n,m,t,\bm{x}}$ be the truncated quadrature solution given as in \eqref{eqn: quad solution}. Then
	\begin{equation}\label{eqn: trun-quad conclusion}
		\vert u(t,\bm{x}) - \widetilde{u}_{n,m,t,\bm{x}} \vert \leq \eps.
	\end{equation}
	\begin{proof}
		By \eqref{eqn: trun-quad def n}, it holds that 
        \begin{equation}
        	M := 2^{n-1} \geq 2^{n_{d,\eps}-1} \geq M_{d,\eps/2}. 
        	\end{equation}
        Let $\widetilde{u}_{n,t,\bm{x}}$ be the truncated solution given by \eqref{eqn: quad-prop truncated solution}.	By Proposition \ref{prop: truncation error} (with $\eps \leftarrow \eps/2$ and $M \leftarrow 2^{n-1}$ in the notation of Proposition \ref{prop: truncation error}), it follows that 
		\begin{equation}
			\vert u(t,\bm{x}) - \widetilde{u}_{n,t,\bm{x}} \vert \leq \eps/2.
		\end{equation}
		Moreover, by \eqref{eqn: trun-quad def m}, it holds that $C_2^2 d^{2} 2^{-m} \leq \eps/2$. Hence, by 
		Proposition \ref{prop: quadrature error}, it follows that
		\begin{equation}
			\vert \widetilde{u}_{n,t,\bm{x}} - \widetilde{u}_{n,m,t,\bm{x}} \vert \leq \eps/2.
		\end{equation}
		Thus, 
		\eqref{eqn: trun-quad conclusion} follows from triangle inequality. 
	\end{proof}
	
\end{corollary}
\subsection{Step 3: Approximation error bounds for payoff function}
\begin{lemma}[Sublinear property of the $\max$-function]\label{lemma: sublinear property} 
    Let $\eps>0$, $n \in \N$, and let $(a_i)_{i=1}^n, (\tilde{a}_i)_{i=1}^n \subset \R$ satisfy for all $i=1,\ldots,n$:
    \begin{equation}\label{eqn: sublinear1}
        \vert a_i - \tilde{a}_i \vert \leq \eps. 
    \end{equation}
    Then,
    \begin{equation}\label{eqn: sublinear2}
        \left\vert \max_{i=1,\ldots,n}\{a_i\} - \max_{i=1,\ldots,n}\{\tilde{a}_i\} \right\vert \leq \eps.
    \end{equation}
\begin{proof}
    Let $i \in \{1,\ldots,n\}$ be arbitrary. The fact that $\tilde{a}_i \leq \max_{j=1,\ldots,n} \{\tilde{a}_j\}$ and \eqref{eqn: sublinear1} imply that 
    \begin{equation}
        a_i - \max_{j=1,\ldots,n}\{\tilde{a}_j\} \leq a_i - \tilde{a}_i \leq \vert a_i - \tilde{a}_i\vert \leq \eps. 
    \end{equation}
    Since $i$ was arbitrarily chosen, taking maximum over $i=1,\ldots,n$ yields
    \begin{equation}
        \max_{i=1,\ldots,n} \{a_i\} - \max_{j=1,\ldots,n}\{\tilde{a}_j\}  \leq \eps.
    \end{equation}
    Repeating the argument by symmetry concludes \eqref{eqn: sublinear2}.
\end{proof}
\end{lemma}

\begin{proposition}[Approximating payoff function]\label{prop: quad-payoff} Let $d \in \N$, $r,T \in (0,\infty)$, $(t,\bm{x}) \in [0,T) \times \R_+^d$, $n_1 \in \N \cap \{2,3,\ldots\}$, and  $m_1 \in \N_0$. Let $h:\R^d_+ \to \R$ be the continuous piecewise affine function given by \eqref{eqn: CPWA payoff}. i.e.
	\begin{equation}
		 h(\bm{x}) = \sum_{k=1}^{K} \xi_k \max\{\bm{a}_{k,l} \cdot \bm{x} + b_{k,l}: l = 1,\ldots,I_k\}.
		 \end{equation}
Let Assumption \ref{assumption: CPWA} hold with corresponding constant $C_2 \in [1,\infty)$, and let $n_{2} \in \N \cap \{2,3,\ldots\}$ be defined by
\begin{equation}
	n_2 := 1 + \lceil \log_2(C_2)\rceil .
\end{equation}
For every $m_2 \in \N_0$ define the function $[\cdot]_{m_2}:\R \to 2^{-m_2}\Z$ by $	[x]_{m_2} := \frac{\lfloor 2^{m_2}x\rfloor}{2^{m_2}}$, $x\in  \R$. 
For every $m_2 \in \N_0$, 
$k=1,\ldots,K$, 
and $l=1,\ldots,I_k$, 
define  $\widetilde{\bm{a}}_{n_2,m_2,k,l} = (\widetilde{\bm{a}}_{n_2,m_2,k,l,1},\cdots,\widetilde{\bm{a}}_{n_2,m_2,k,l,d})\in \R^d$ and $\widetilde{b}_{n_2,m_2,k,l} \in \R$ by  
\begin{equation}\label{eqn: approx payoff1}
	\begin{split}
		&\widetilde{\bm{a}}_{n_2,m_2,k,l,i} := [(\bm{a}_{k,l})_i]_{m_2}, \quad i=1,\ldots,d,\\
		&\widetilde{b}_{n_2,m_2,k,l} := [b_{k,l}]_{m_2},
	\end{split}
\end{equation}
and define the function $\widetilde{h}_{n_2,m_2}:\R^d \to \R$ by 
\begin{equation}\label{eqn: approx payoff2}
	\widetilde{h}_{n_2,m_2}(\bm{x}) := \sum_{k=1}^K \xi_k \max \{\widetilde{\bm{a}}_{n_2,m_2,k,l} \cdot \bm{x} + \widetilde{b}_{n_2,m_2,k,l}: l=1,\ldots,I_k\}.
\end{equation}
Let $\widetilde{u}_{n_1,m_1,t,\bm{x}} \in \R$ be the truncated quadrature solution given as in \eqref{eqn: quad solution}. i.e.
\begin{equation}
	\widetilde{u}_{n_1,m_1,t,\bm{x}} := e^{-r(T-t)}\sum_{\bm{j} \in \mathbb{K}_{n_1,m_1,+}^d} h(\bm{j})p_{\bm{j},m_1}.
	\end{equation}
For every $m_2 \in \N_0$ let $\widetilde{u}_{n_1,m_1,n_2,m_2,t,\bm{x}} \in \R$ be the truncated quadrature solution with an approximated payoff function  given by
\begin{equation}\label{eqn: approx payoff3}
	\widetilde{u}_{n_1,m_1,n_2,m_2,t,\bm{x}} := e^{-r(T-t)}\sum_{\bm{j} \in \mathbb{K}_{n_1,m_1,+}^d} \widetilde{h}_{n_2,m_2}(\bm{j})p_{\bm{j},m_1}
\end{equation}
Then, for every $m_2 \in \N_0$ it holds that $(\widetilde{\bm{a}}_{n_2,m_2,k,l},\widetilde{b}_{n_2,m_2,k,l}) \in \mathbb{K}_{n_2,m_2}^{d+1}$ and that 
\begin{equation}\label{eqn: approx payoff4}
\vert \widetilde{u}_{n_1,m_1,t,\bm{x}} - \widetilde{u}_{n_1,m_1,n_2,m_2,t,\bm{x}}\vert \leq C_2d^22^{n_1-m_2}.
\end{equation}
\begin{proof}
	By Assumption \ref{assumption: CPWA}, it holds that
	\begin{equation}
		\forall k=1,\ldots,K, \,
		\forall l= 1,\ldots,I_K,\,
		 \forall i=1,\ldots,d:\quad 		\vert (\bm{a}_{k,l})_i\vert, \vert b_{k,l} \vert \leq C_2.
	\end{equation}
    Hence, by Definition \ref{def: two complement}, \eqref{eqn: approx payoff1},  and \eqref{eqn: approx payoff2}, it follows that $(\widetilde{\bm{a}}_{n_2,m_2,k,l},\widetilde{b}_{n_2,m_2,k,l}) \in \mathbb{K}_{n_2,m_2}^{d+1}$. Furthermore, observe that
    \begin{equation}
        \forall c>0: \quad \vert c-[c]_{m_2} \vert \leq 2^{-m_2}.
    \end{equation}
    Hence, for all $k=1,\ldots,K$, $l=1,\ldots,I_k$, and $\bm{x} \in \R_+^d$, it follows that
    \begin{equation}
        \vert \bm{a}_{k,l} \cdot \bm{x} + b_{k,l} - (\widetilde{\bm{a}}_{n_2,m_2,k,l} \cdot \bm{x} + \widetilde{b}_{n_2,m_2,k,l})\vert \leq (d+1)2^{-m_2} \max\{1,\max\{(\bm{x})_i:i=1,\ldots,d\}\}.
    \end{equation}
    Moreover, since $\mathbb{K}_{n_1,m_1} \subset [-2^{n_1-1},2^{n_1-1}]$ (c.f. Definition \ref{def: two complement}), it follows for all $k=1,\ldots,K$, $l=1,\ldots,I_k$ that 
    \begin{equation}
        \forall \bm{j} \in \mathbb{K}_{n_1,m_1,+}^d: \quad \vert  \bm{a}_{k,l} \cdot \bm{j} + b_{k,l} - (\widetilde{\bm{a}}_{n_2,m_2,k,l} \cdot \bm{j} + \widetilde{b}_{n_2,m_2,k,l})\vert \leq (d+1)2^{n_1-1}2^{-m_2}.
    \end{equation}
    Hence, by Lemma \ref{lemma: sublinear property}, it holds for all $k=1,\ldots,K$ that 
    \begin{equation}
        \vert \max\{\bm{a}_{k,l} \cdot \bm{j} + b_{k,l}: l =1,\ldots,I_k\} - \max\{\widetilde{\bm{a}}_{n_2,m_2,k,l} \cdot \bm{j} + \widetilde{b}_{n_2,m_2,k,l}: l=1,\ldots,I_k\}\vert \leq (d+1)2^{n_1-1}2^{-m_2}.
    \end{equation}
    Hence, by Assumption \ref{assumption: CPWA}, it follows that 
    \begin{equation}\label{eqn: step3-approximate h}
        \forall \bm{j} \in \mathbb{K}_{n_1,m_1,+}^d: \quad \vert h(\bm{j}) - \widetilde{h}_{n_2,m_2}(\bm{j}) \vert \leq K(d+1)2^{n_1-1}2^{-m_2} \leq C_2d^2 2^{n_1-m_2}.
    \end{equation}
    Furthermore, since $e^{-r(T-t)}\leq 1$ and 
    \begin{equation}
    	0<\sum_{\bm{j} \in \mathbb{K}_{n_1,m_1,+}^d}  p_{\bm{j},m_1} = \int_{[0,2^{n_1 - 1}]^d} p(\bm{y},T;\bm{x},t)\,d\bm{y} \leq 1,
    	\end{equation}
    we conclude that 
    \begin{equation}
        \vert \widetilde{u}_{n_1,m_1,t,\bm{x}} - \widetilde{u}_{n_1,m_1,n_2,m_2,t,\bm{x}}\vert \leq C_2d^22^{n_1-m_2+1}e^{-r(T-t)}\cdot\sum_{\bm{j} \in \mathbb{K}_{n_1,m_1,+}^d} p_{\bm{j},m_1} \leq C_2d^22^{n_1-m_2}.
    \end{equation}
\end{proof}
\end{proposition}
\begin{corollary}[Truncation and quadrature with approximated payoff errors]\label{corollary: step3-approximate payoff} Let $\eps \in(0,1)$, $d \in \N$, $r,T\in(0,\infty)$, and $(t,\bm{x}) \in [0,T)\times \R_+^d$. Let $u(t,\bm{x})$ be the option price given by \eqref{eqn: option price formula}. Let $h:\R^d_+ \to \R$ be the continuous piecewise affine function given by \eqref{eqn: CPWA payoff}. Let Assumption \ref{assumption: cov matrix} and Assumption \ref{assumption: CPWA} hold with respective constants $C_1,C_2 \in [1,\infty)$, and let $n_{2} := 1 + \lceil \log_2(C_2)\rceil$.	 For every $\eta \in (0,1)$, let $M_{d,\eta} \in [1,\infty)$ be given by \eqref{eqn: def M}.  
For every $n_1,m_1,m_2 \in \N$ satisfying
    \begin{align}
        n_1&\geq 1 + \log_2(M_{d,\eps/3}), \label{eqn: step3-n_1}\\
        %
        m_1 &\geq\log_2(C_2^2 d^{2}(\eps/3)^{-1} ),\label{eqn: step3-m_1}\\
        %
        m_2 &\geq
        1 + \log_2(M_{d,\eps/3}) 
        +  \log_2(C_2 d^2(\eps/3)^{-1}), \label{eqn: step3-m_2}
    \end{align}
    let $\widetilde{u}_{n_1,m_1,n_2,m_2,t,\bm{x}} \in \R$ be the truncated quadrature solution with an approximated payoff function be given by \eqref{eqn: approx payoff3}. Then we have that
    \begin{equation}
        \vert u(t,\bm{x}) - \widetilde{u}_{n_1,m_1,n_2,m_2,t,\bm{x}} \vert \leq \eps.
    \end{equation}
    \begin{proof}
        Let $\tilde{u}_{n_1,m_1,t,\bm{x}}$ be the quadrature solution given by \eqref{eqn: quad solution}. By \eqref{eqn: step3-n_1}, \eqref{eqn: step3-m_1}, and Corollary \ref{corollary: trun-quad error} (with $\eps \leftarrow 2\eps/3$ in the notation of Corollary \ref{corollary: trun-quad error}), it follows that
        \begin{equation}
            \vert u(t,\bm{x}) - \widetilde{u}_{n_1,m_1,t,\bm{x}} \vert \leq 2\eps/3.
        \end{equation} 
        Moreover, by Proposition \ref{prop: quad-payoff} (with $m_2 \leftarrow m_2$ in the notation of Proposition \ref{prop: quad-payoff}) and \eqref{eqn: step3-m_2}, it follows that 
        \begin{equation}
            \vert \widetilde{u}_{n_1,m_1,t,\bm{x}} - \widetilde{u}_{n_1,m_1,n_2,m_2,t,\bm{x}} \vert \leq C_2d^22^{n_1-m_2} \leq \eps/3.
        \end{equation}	
        Thus, the conclusion follows from the triangle inequality. 
        
    \end{proof}

\end{corollary}
\subsection{Step 4: Distribution  loading error bounds}
\begin{proposition}[Distribution loading errors]\label{prop: loading errors}
    Let $\eps \in(0,1)$, $d \in \N$, $r,T\in (0,\infty)$, and $(t,\bm{x}) \in [0,T)\times \R_+^d$. Let $u(t,\bm{x})$ be the option price given by \eqref{eqn: option price formula}. Let $h:\R^d_+ \to \R$ be the continuous piecewise affine function given by \eqref{eqn: CPWA payoff}. Let Assumption~\ref{assumption: cov matrix} and Assumption \ref{assumption: CPWA} hold with respective constants $C_1,C_2 \in [1,\infty)$, and let  $n_2:=1 + \lceil \log_2(C_2)\rceil$.  For every $\eta \in (0,1)$, let $M_{d,\eta} \in [1,\infty)$ be given by \eqref{eqn: def M} and
    for every $n_1,m_1 \in \N$
    let $\{\widetilde{p}_{\bm{j},\eta}: \bm{j} \in \mathbb{K}_{n_1,m_1,+}^d\} \subset [0,1]$ satisfy
    \begin{equation}
        \sum_{\bm{j} \in \mathbb{K}_{n_1,m_1,+}^d} \widetilde{p}_{\bm{j},\eta} = 1
    \end{equation}
    and 
    \begin{equation}\label{eqn: step4 approximate loading}
		\sum_{\bm{j} \in \mathbb{K}_{n_1,m_1,+}^d} \vert\widetilde{p}_{\bm{j},\eta} - \gamma^{-1}p_{{\bm{j}},m_1} \vert \leq \frac{\eta}{C_2^2d^22^{n_1+1}}, 
    \end{equation}
    where for all $\bm{j} = (j_1,\ldots,j_d) \in \mathbb{K}_{n_1,m_1,+}^d$,
    \begin{equation}
        p_{{\bm{j}},m_1} := \int_{Q_{\bm{j},m_1}} p(\bm{y},T;\bm{x},t)\,d\bm{y}, \quad Q_{\bm{j},m_1} := [j_1,j_1+2^{-m_1}) \times \cdots \times [j_d,j_d+2^{-m_1}),
    \end{equation}
    and
    \begin{equation}
        \gamma := \sum_{\bm{j} \in \mathbb{K}_{n_1,m_1,+}^d}  p_{{\bm{j}},m_1} \in (0,1).
    \end{equation}
Moreover, for every $n_1,m_1, m_2 \in \N$ satisfying
  \begin{align}
	n_1&\geq 1 + \log_2(M_{d,\eps/4}), \label{eqn: step4-n_1}\\
	%
	m_1&\geq   \log_2(C_2^2 d^{2}(\eps/4)^{-1}), \label{eqn: step4-m_1}
	\\
	%
	m_2 &\geq  1 + \log_2(M_{d,\eps/4}) 
	 + \log_2(C_2 d^2(\eps/4)^{-1}) \label{eqn: step4-m_2},
\end{align}
let $\tilde{h}_{n_2,m_2}:\R^d \to \R$ be given by \eqref{eqn: approx payoff2}, and let $\widetilde{u}_{n_1,m_1,n_2,m_2,p,t,\bm{x}} \in \R$ be the truncated quadrature solution with approximated payoff and loaded distribution given by 
    \begin{equation}\label{eqn: step4-approximate loading soln}
        \widetilde{u}_{n_1,m_1,n_2,m_2,p,t,\bm{x}} := \gamma e^{-r(T-t)}\sum_{\bm{j} \in \mathbb{K}_{n_1,m_1,+}^d} \widetilde{p}_{\bm{j},\eps/4} \widetilde{h}_{n_2,m_2}(\bm{j}),
    \end{equation}
    Then, 
    \begin{equation}
    	\vert         u(t,\bm{x}) - \widetilde{u}_{n_1,m_1,n_2,m_2,p,t,\bm{x}}\vert \leq \eps.
    \end{equation}
\begin{proof}Let $\widetilde{u}_{n_1,m_1,n_2,m_2,t,\bm{x}} \in \R$ be the quadrature solution with an approximated payoff function be given by \eqref{eqn: approx payoff3}. Note that by Corollary \ref{corollary: step3-approximate payoff} (with $\eps \leftarrow 3\eps/4$ in the notation of Corollary \ref{corollary: step3-approximate payoff}), it holds that 
    \begin{equation}
        \vert u(t,\bm{x}) - \widetilde{u}_{n_1,m_1,n_2,m_2,t,\bm{x}} \vert \leq 3\eps/4.
    \end{equation}
    Moreover, by Lemma \ref{lemma: lin growth} and \eqref{eqn: step3-approximate h}, it holds for every $\bm{j} \in \mathbb{K}_{n_1,m_1,+}^d$ that
    \begin{equation}\label{eqn: bound on tilde h}
        \begin{split}
            \vert \widetilde{h}_{n_2,m_2}(\bm{j}) \vert &\leq \vert h(\bm{j}) - \widetilde{h}_{n_2,m_2}(\bm{j}) \vert + \vert h(\bm{j})\vert \\
            &\leq C_2d^22^{n_1-m_2} + C_2^2d(1+ \Vert \bm{j}\Vert_1)\\
            &\leq C_2d^22^{n_1-m_2} + C_2^2d(1 + d 2^{n_1-1})\\
            &\leq C_2d^22^{n_1-m_2} + C_2^2d^22^{n_1}\\            
            &\leq C_2^2d^22^{n_1+1}.
        \end{split}
    \end{equation}
    Hence, by \eqref{eqn: approx payoff3} and \eqref{eqn: step4-approximate loading soln}, using \eqref{eqn: step4 approximate loading}, the fact that $0<\gamma, e^{-r(T-t)} \leq 1$, and the above estimate, we have
    \begin{equation}
        \begin{split}
            \vert \widetilde{u}_{n_1,m_1,n_2,m_2,t,\bm{x}} - \widetilde{u}_{n_1,m_1,n_2,m_2,p,t,\bm{x}}\vert &\leq \gamma e^{-r(T-t)}\cdot \max_{\bm{j} \in \mathbb{K}_{n_1,m_1,+}^d} \vert \widetilde{h}_{n_2,m_2}(\bm{j}) \vert \cdot \sum_{\bm{j} \in \mathbb{K}_{n_1,m_1,+}^d} \vert  \widetilde{p}_{\bm{j},\eps/4}- \gamma^{-1}p_{\bm{j},m_1}\vert \\
             &\leq \eps/4.
        \end{split}
    \end{equation}
    Thus, the conclusion follows from triangle inequality. 
\end{proof}
\end{proposition}
\subsection{Step 5: Rotation error bounds}
\begin{proposition}[Rotation error]\label{prop: rotation error} 
	Let $\eps \in(0,1)$, $d \in \N$, $r,T\in (0,\infty)$, and  $(t,\bm{x}) \in [0,T)\times \R_+^d$. Let $u(t,\bm{x})$ be the option price given by \eqref{eqn: option price formula}. Let $h:\R^d_+ \to \R$ be the continuous piecewise affine function given by \eqref{eqn: CPWA payoff}. Let Assumption \ref{assumption: cov matrix} and Assumption \ref{assumption: CPWA} hold with respective constants $C_1,C_2 \in [1,\infty)$, and let  $n_2:=1 + \lceil \log_2(C_2)\rceil$. For every $\eta \in (0,1)$ let $M_{d,\eta} \in [1,\infty)$ be given by \eqref{eqn: def M},
	for every $n_1, m_1, m_2 \in \N$ 
    let $\{\widetilde{p}_{\bm{j},\eta}: \bm{j} \in \mathbb{K}_{n_1,m_1,+}^d\} \subset [0,1]$  satisfy
    \begin{equation}
        \sum_{\bm{j} \in \mathbb{K}_{n_1,m_1,+}^d} \widetilde{p}_{\bm{j},\eta} = 1
    \end{equation}
    and 
    \begin{equation}\label{eqn: step5 approximate loading}
        \sum_{\bm{j} \in \mathbb{K}_{n_1,m_1,+}^d} \vert\widetilde{p}_{\bm{j},\eta} - \gamma^{-1}p_{{\bm{j}},m_1} \vert \leq \frac{\eta}{C_2^2d^22^{n_1+1}}, 
    \end{equation}
    where for all $\bm{j} = (j_1,\ldots,j_d) \in \mathbb{K}_{n_1,m_1,+}^d$,
    \begin{equation}
        p_{{\bm{j}},m_1} := \int_{Q_{\bm{j},m_1}} p(\bm{y},T;\bm{x},t)\,d\bm{y}, \quad Q_{\bm{j},m_1} := [j_1,j_1+2^{-m_1}) \times \cdots \times [j_d,j_d+2^{-m_1}),
    \end{equation}
    and
    \begin{equation}
        \gamma := \sum_{\bm{j} \in \mathbb{K}_{n_1,m_1,+}^d}  p_{{\bm{j}},m_1} \in (0,1),
    \end{equation}
    let $\mathfrak{s}_{d,\eta} \in (0,\infty)$ be defined by 
	\begin{equation}
	\mathfrak{s}_{d,\eta} := \sqrt{\frac{\eta}{(C_2^2d^22^{n_1+1})^3}} \label{eqn: step5-s},
	\end{equation} 
	let $\tilde{h}_{n_2,m_2}:\R^d \to \R$ be given by \eqref{eqn: approx payoff2}, and let $a_{n_1,n_2,m_1,m_2,\eta} \in [0,1]$ be the amplitude given by 
	\begin{equation}\label{eqn: step5-a}
		a_{n_1,n_2,m_1,m_2,\eta} = \sum_{\bm{j} \in \mathbb{K}_{n_1,m_1,+}^d}\widetilde{p}_{\bm{j},\eta} \sin^2\left(\tfrac{\mathfrak{s}_{d,\eta} \widetilde{h}_{n_2,m_2}(\bm{j})}{2} + \frac{\pi}{4}\right).
	\end{equation}
Moreover,  
	for every $n_1, m_1, m_2 \in \N$ satisfying 
\begin{align}
	n_1&\geq 1 + \log_2(M_{d,\eps/5}), \label{eqn: step5-n_1}\\
	m_1 &\geq \log_2(C_2^2 d^{2}(\eps/5)^{-1}), \label{eqn: step5-m_1}\\
	m_2&\geq  1 + \log_2(M_{d,\eps/5}) 
	+ \log_2(C_2 d^2(\eps/5)^{-1}) \label{eqn: step5-m_2},
\end{align}
	let $\widetilde{u}_{n_1,m_1,n_2,m_2,p,a,t,\bm{x}} \in \R$ be the truncated quadrature solution with approximated payoff and loaded distribution with rotation given by 
	\begin{equation}\label{eqn: loading and rotation solution}
		\widetilde{u}_{n_1,m_1,n_2,m_2,p,a,t,\bm{x}} := \mathfrak{s}^{-1} \gamma e^{-r(T-t)} (2a -1),
	\end{equation}
where here 
\begin{equation}
	\mathfrak{s}:=\mathfrak{s}_{d,\eps/5}, \qquad \text{and}\qquad a=:a_{n_1,n_2,m_1,m_2,\eps/5}.
\end{equation}	
	Then, the following holds:
	\begin{enumerate}[(i)]
		\item \label{step5-item1}
		\begin{equation} \label{step5-item1-eq1}
			\widetilde{u}_{n_1,m_1,n_2,m_2,p,a,t,\bm{x}} = \mathfrak{s}^{-1} \gamma e^{-r(T-t)}\sum_{\bm{j} \in \mathbb{K}_{n_1,m_1,+}^d} \widetilde{p}_{\bm{j},\eps/5}\sin(\mathfrak{s}\widetilde{h}_{n_2,m_2}(\bm{j})).
		\end{equation}
		\item  \label{step5-item2}
		\begin{equation}
			\vert u(t,\bm{x}) - \widetilde{u}_{n_1,m_1,n_2,m_2,p,a,t,\bm{x}} \vert \leq \eps.
		\end{equation}
	\end{enumerate}
\begin{proof}
First, recall the trigonometric identity that for all $x\in \R$
\begin{equation}
    \sin^2(\frac{x}{2}+\frac{\pi}{4}) = 1 - \cos^2(\frac{x}{2}+\frac{\pi}{4}) = 1 - \frac{1}{2}(1+\cos(x+\pi/2)) = \frac{1}{2} + \frac{1}{2}\sin(x).
\end{equation}
This, \eqref{eqn: step5-a},\eqref{eqn: loading and rotation solution}, 
and the fact that $\sum_{\bm{j} \in \mathbb{K}_{n_1,m_1,+}^d} \widetilde{p}_{\bm{j},\eps/5} = 1$ imply that
\begin{equation}
	\begin{split}
	\widetilde{u}_{n_1,m_1,n_2,m_2,p,a,t,\bm{x}} 
	&=	 \mathfrak{s}^{-1} \gamma e^{-r(T-t)}\sum_{\bm{j} \in \mathbb{K}_{n_1,m_1,+}^d} \widetilde{p}_{\bm{j},\eps/5}\Big[2\sin^2\Big(\tfrac{\mathfrak{s} \widetilde{h}_{n_2,m_2}(\bm{j})}{2} + \frac{\pi}{4}\Big)-1\Big]
	\\
	&= \mathfrak{s}^{-1} \gamma e^{-r(T-t)}\sum_{\bm{j} \in \mathbb{K}_{n_1,m_1,+}^d} \widetilde{p}_{\bm{j},\eps/5}\sin(\mathfrak{s}\widetilde{h}_{n_2,m_2}(\bm{j})),
	\end{split}
\end{equation}
which proves Item~\ref{step5-item1}.
 Next, for the proof of Item \ref{step5-item2}, note that by Taylor expansion, one has for any $x \in [-1,1]$ the estimate
\begin{equation}\label{eqn:sinx-x}
    \vert \sin(x) -x \vert \leq \left(\tfrac{\vert x\vert ^3}{3!} + \tfrac{\vert x\vert ^5}{5!} +\tfrac{\vert x\vert ^7}{7!} + \cdots \right) 
    \leq 
    \vert x\vert^3 \big( e -[1+\tfrac{1}{1!}+\tfrac{1}{2!}] \big)
    \leq
    \vert x\vert^3.
\end{equation}
Moreover, since $\vert \widetilde{h}_{n_2,m_2}(\bm{j}) \vert \leq C_2^2d^2 2^{n_1+1}$ (c.f. \eqref{eqn: bound on tilde h}), by \eqref{eqn: step5-s}, it holds for all $\bm{j} \in \mathbb{K}_{n_1,m_1,+}^d$ that
\begin{equation}
    \vert \mathfrak{s}\widetilde{h}_{n_2,m_2}(\bm{j}) \vert \leq \sqrt{\eps/5}\leq 1.
\end{equation}
This, \eqref{eqn:sinx-x}, \eqref{eqn: step5-s}, and  the fact that $\vert \widetilde{h}_{n_2,m_2}(\bm{j}) \vert \leq C_2^2d^2 2^{n_1+1}$ hence ensure for all $\bm{j} \in \mathbb{K}_{n_1,m_1,+}^d$ that
\begin{equation}\label{step5-item1-eq1-est}
    \mathfrak{s}^{-1}\vert \sin(\mathfrak{s}\widetilde{h}_{n_2,m_2}(\bm{j})) - \mathfrak{s}\widetilde{h}_{n_2,m_2}(\bm{j}) \vert \leq \mathfrak{s}^{-1} \vert \mathfrak{s} \widetilde{h}_{n_2,m_2}(\bm{j}) \vert^3 = \mathfrak{s}^2\vert \widetilde{h}_{n_2,m_2}(\bm{j})\vert^3 \leq \eps/5.
\end{equation}
Let here $\widetilde{u}_{n_1,m_1,n_2,m_2,p,t,\bm{x}} \in \R$ be the truncated quadrature solution with approximated payoff and loaded distribution given by 
\begin{equation}
\widetilde{u}_{n_1,m_1,n_2,m_2,p,t,\bm{x}} := \gamma e^{-r(T-t)}\sum_{\bm{j} \in \mathbb{K}_{n_1,m_1,+}^d} \widetilde{p}_{\bm{j},\eps/5} \widetilde{h}_{n_2,m_2}(\bm{j}).
\end{equation}
This, \eqref{step5-item1-eq1-est}, \eqref{step5-item1-eq1}, and the fact that 
\begin{equation}
0<\gamma e^{-r(T-t)}\sum_{\bm{j} \in \mathbb{K}_{n_1,m_1,+}^d} \widetilde{p}_{\bm{j},\eps/5} \leq 1
\end{equation}
imply that 
\begin{equation}
    \vert \widetilde{u}_{n_1,m_1,n_2,m_2,p,t,\bm{x}} - \widetilde{u}_{n_1,m_1,n_2,m_2,p,a,t,\bm{x}} \vert \leq \eps/5.
\end{equation}
Furthermore, by Proposition \eqref{prop: loading errors} (with $\eps \leftarrow 4\eps/5$ in the notation of Proposition \eqref{prop: loading errors}), it holds that
\begin{equation}
    \vert u(t,\bm{x}) - \widetilde{u}_{n_1,m_1,n_2,m_2,p,t,\bm{x}}\vert \leq 4\eps/5.
\end{equation}
Hence, the conclusion follows from the triangle inequality. 
\end{proof}
\end{proposition}
	
\subsection{Step 6: Quantum amplitude estimation error bounds}
\begin{proposition}[Combined errors]\label{proposition: final error estimate} Let $\eps \in(0,1)$, $d \in \N$, $r,T\in (0,\infty)$, and $(t,\bm{x}) \in [0,T)\times \R_+^d$. Let $u(t,\bm{x})$ be the option price given by \eqref{eqn: option price formula}. Let $h:\R^d_+ \to \R$ be the continuous piecewise affine function given by \eqref{eqn: CPWA payoff}. Let Assumption \ref{assumption: cov matrix} and Assumption \ref{assumption: CPWA} hold with respective constants $C_1,C_2 \in [1,\infty)$, 
	and let 
    \begin{equation}
        n_2:=1 + \lceil \log_2(C_2)\rceil
    \end{equation}
For every $\eta \in (0,1)$, let $M_{d,\eta} \in [1,\infty)$ be given by \eqref{eqn: def M}. Let $n_{1,d,\eps},m_{1,d,\eps},m_{2,d,\eps} \in (0,\infty)$ be defined by 
    \begin{align}
        n_{1,d,\eps} &:= 1 + \log_2(M_{d,\eps/6}), \label{eqn: step6-n_1}\\
        m_{1,d,\eps} &:=   \log_2(C_2^2 d^{2}(\eps/6)^{-1}), \label{eqn: step6-m_1}\\
        m_{2,d,\eps} &:= 1 + \log_2(M_{d,\eps/6}) 
         + \log_2(C_2 d^2(\eps/6)^{-1}) \label{eqn: step6-m_2}.
    \end{align}
Moreover, for every $n_1,m_1,m_2 \in \N$ satisfying $n_1 \geq n_{1,d,\eps}$,  $m_1 \geq m_{1,d,\eps}$, and $m_2 \geq m_{2,d,\eps}$, let $\{\widetilde{p}_{\bm{j},\eps/6}: \bm{j} \in \mathbb{K}_{n_1,m_1,+}^d\} \subset [0,1]$  satisfy
\begin{equation}
    \sum_{\bm{j} \in \mathbb{K}_{n_1,m_1,+}^d} \widetilde{p}_{\bm{j},\eps/6} = 1
\end{equation}
and 
\begin{equation}\label{eqn: step6 approximate loading}
    \sum_{\bm{j} \in \mathbb{K}_{n_1,m_1,+}^d} \vert\widetilde{p}_{\bm{j},\eps/6} - \gamma^{-1}p_{{\bm{j}},m_1} \vert \leq \frac{\eps}{6C_2^2d^22^{n_1+1}}, 
\end{equation}
where for all $\bm{j} = (j_1,\ldots,j_d) \in \mathbb{K}_{n_1,m_1,+}^d$,
\begin{equation}
    p_{{\bm{j}},m_1} := \int_{Q_{\bm{j},m_1}} p(\bm{y},T;\bm{x},t)\,d\bm{y}, \quad Q_{\bm{j},m_1} := [j_1,j_1+2^{-m_1}) \times \cdots \times [j_d,j_d+2^{-m_1}),
\end{equation}
and
\begin{equation}
    \gamma := \sum_{\bm{j} \in \mathbb{K}_{n_1,m_1,+}^d}  p_{{\bm{j}},m_1} \in (0,1), \label{eqn: step6-gamma}
\end{equation}
let $\mathfrak{s}\equiv \mathfrak{s}_{d,\eps/6} \in (0,\infty)$ be defined by 
\begin{equation}
    \mathfrak{s} \equiv \mathfrak{s}_{d,\eps/6} := \sqrt{\frac{\eps/6}{(C_2^2d^22^{n_1+1})^3}} \label{eqn: step6-s},
\end{equation}
   let $\tilde{h}_{n_2,m_2}:\R^d \to \R$ be given by \eqref{eqn: approx payoff2}, let $ a\equiv a_{n_1,n_2,m_1,m_2,\varepsilon/6} \in [0,1]$ be the amplitude given by 
\begin{equation}\label{eqn: error estimate amplititude}
    a \equiv a_{n_1,n_2,m_1,m_2,\mathfrak{s},\varepsilon/6} := \sum_{\bm{j} \in \mathbb{K}_{n_1,m_1,+}^d} \widetilde{p}_{\bm{j},\eps/6} \sin^2\left(\frac{\mathfrak{s}\widetilde h_{n_2,m_2}(\bm{j})}{2} + \frac{\pi}{4}\right),
\end{equation}
let $\hat{a}\in [0,1]$ satisfy
    \begin{equation}\label{eqn:step-6-grover-error}
        \vert a - \hat{a} \vert \leq \frac{\eps \mathfrak{s}}{12},
    \end{equation}
    and let $\widetilde{U}_{t,\bm{x}}$ be the approximated solution given by
    \begin{equation}
        \widetilde{U}_{t,\bm{x}} := \mathfrak{s}^{-1} \gamma e^{-r(T-t)}(2\hat{a} - 1).
    \end{equation}
    Then,
    \begin{equation}\label{eqn:step-6-approximation}
        \vert u(t,\bm{x}) - \widetilde{U}_{t,\bm{x}} \vert \leq \eps.
    \end{equation}
\begin{proof}
    Let here $\widetilde{u}_{n_1,m_1,n_2,m_2,p,a,t,\bm{x}} \in \R$ be the truncated quadrature solution with approximated payoff and loaded distribution with rotation given by 
        \begin{equation}   
        	  \widetilde{u}_{n_1,m_1,n_2,m_2,p,a,t,\bm{x}} := \mathfrak{s}^{-1} \gamma e^{-r(T-t)} (2a -1).
        	  \end{equation}
    By Proposition \ref{prop: rotation error} item (ii) (with $\eps \leftarrow 5\eps/6$ in the notation of Proposition \ref{prop: rotation error}), it holds that
    \begin{equation}\label{eqn:step-6-triangle1}
        \vert u(t,\bm{x}) - \widetilde{u}_{n_1,m_1,n_2,m_2,p,a,t,\bm{x}} \vert \leq 5\eps/6.
    \end{equation}
    Using  $0<\gamma e^{-r(T-t)} \leq 1$ and \eqref{eqn:step-6-grover-error}, it follows that 
    \begin{equation}
        \vert \widetilde{u}_{n_1,m_1,n_2,m_2,p,a,t,\bm{x}} - \widetilde{U}_{t,\bm{x}} \vert \leq 2\mathfrak{s}^{-1} \vert a - \hat{a} \vert \leq\eps/6.
    \end{equation}
Hence, we conclude \eqref{eqn:step-6-approximation}.
  \end{proof}
\end{proposition}

\section{
	Proof of Theorem \ref{main theorem}}\label{sec:AlgoandProof}
In this section, we 
provide the proof of Theorem~\ref{main theorem}.

\begin{proof}[Proof of Theorem \ref{main theorem}.]
    First, let $n_1,n_2,m_1,m_2, N, \gamma, \mathfrak{s}$ be defined as 
    in line~2--3 of Algorithm~\ref{quantum algorithm}, and set $n := n_1 + n_2$ and $m := m_1 + m_2$. Let $\bm{h}: \mathbb{F}_{n_1,m_1}^d \to \mathbb{F}_{n+d+K+1+m}$ be defined as in \eqref{eqn: encoded payoff}, and let $p:= d(2n_2+2m_2+3)$ and $N,\{q_k\}_{k=1}^K \in \N$ be given by \eqref{eqn: def N, q_k} from Proposition \ref{prop: loading payoff circuit}. By Assumption \ref{assumption: distribution loading}, it holds that
   \begin{equation}
       \mathcal{P}\ket{0}_{d(n_1+m_1)} = \sum_{\bm{i} \in \mathbb{F}_{n_1,m_1,+}^d} \sqrt{\widetilde{p}_{i}}\ket{i_1}_{n_1+m_1}\cdots\ket{i_d}_{n_1+m_1}.
   \end{equation}
   Hence, together with Proposition \ref{prop: loading payoff circuit}, the circuit $\mathcal{A}:= \mathcal{R}_h (\mathcal{P} \otimes I_2^{\otimes(N-d(n_1+m_1))})$ satisfies 
   \begin{equation}
       \begin{split}
       \mathcal{A}\ket{0}_N &= \sum_{\bm{i} \in \mathbb{F}_{n_1,m_1,+}^d} \sqrt{\widetilde{p}_{i}}\ket{i_1}_{n_1+m_1}\cdots\ket{i_d}_{n_1+m_1}\ket{\anc}_{q_1+\cdots+q_K + 2K(n+m+d+5)-1}[\cos(\bar{f}(\bm{h}(\bm{i}))/2)\ket{0}+\sin(\bar{f}(\bm{h}(\bm{i}))/2)\ket{1}]\\
       &= \sum_{\bm{i} \in \mathbb{F}_{n_1,m_1,+}^d} \sqrt{\widetilde{p}_{i}}\cos(\bar{f}(\bm{h}(\bm{i}))/2)\ket{i_1}_{n_1+m_1}\cdots\ket{i_d}_{n_1+m_1}\ket{\anc}_{q_1+\cdots+q_K + 2K(n+m+d+5)-1}\ket{0}\\
       &\qquad +\sum_{\bm{i} \in \mathbb{F}_{n_1,m_1,+}^d} \sqrt{\widetilde{p}_{i}}\sin(\bar{f}(\bm{h}(\bm{i}))/2)\ket{i_1}_{n_1+m_1}\cdots\ket{i_d}_{n_1+m_1}\ket{\anc}_{q_1+\cdots+q_K + 2K(n+m+d+5)-1}\ket{1}\\
       &=: \sqrt{1-a}\ket{\psi_0}_{N-1}\ket{0} + \sqrt{a}\ket{\psi_1}_{N-1}\ket{1},
       \end{split}
   \end{equation} 
    where as in Proposition \ref{prop: loading payoff circuit}, $\bar{f}:\mathbb{F}_{n+d+K+1,m}\to\R$ is defined by $\bar{f}(i) = f \circ \mathrm{D}_{n+d+K+1,m}(i)$ and $f(x) = \mathfrak{s}x + \frac{\pi}{2}$, and $a \in [0,1]$ is given by 
    \begin{equation}
        a := \sum_{\bm{i} \in \mathbb{F}_{n_1,m_1,+}^d} {\widetilde{p}_{i}}\sin^2(\bar{f}(\bm{h}(\bm{i}))/2).
    \end{equation}
    By Proposition \ref{prop: loading payoff circuit} and Proposition \ref{prop: quad-payoff}, we note that the function  $\widetilde{h}_{n_2,m_2}$ given in \eqref{eqn: approx payoff2} coincides with the function $\mathrm{D}_{n+d+K+1,m}(\bm{h})$ when restricted to the domain $\mathbb{F}_{n_1,m_1}^d$. Using this and Proposition \ref{proposition: final error estimate} (with $\eps \leftarrow \tfrac{\eps}{6C_2^2d^22^{n_1+1}}$ in the notation of Assumption \ref{assumption: distribution loading}), we have 
    \begin{equation}
        a = \sum_{\bm{i} \in \mathbb{F}_{n_1,m_1,+}^d} {\widetilde{p}_{i}}\sin^2\left(\frac{\mathfrak{s} \mathrm{D}_{n+d+K+1,m}(\bm{h}(\bm{i}))}{2} + \frac{\pi}{4}\right) = a_{n_1,n_2,m_1,m_2,\mathfrak{s},\eps/6},
    \end{equation}
    where $a_{n_1,n_2,m_1,m_2,\mathfrak{s},\eps/6}$ is 
     defined in \eqref{eqn: error estimate amplititude}. By Proposition \ref{prop: m-IQAE}, the output $\widehat{a}$ from line~7 of Algorithm~\ref{quantum algorithm} satisfies the bound 
    \begin{equation}
        \vert a - \widehat{a} \vert \leq \eps\mathfrak{s}/12, \quad \text{with probability at least $1-\alpha$}.
    \end{equation}
    Thus, the estimate \eqref{eqn: theorem estimate} follows from Proposition \ref{proposition: final error estimate}. Next, we count the total number of qubits and elementary gates used to construct circuit $\mathcal{A}$.  Let $M_{d,\eps/6}$ be the constant given by \eqref{eqn: def M}, and note that 
    \begin{equation}
        M_{d,\eps/6} = 6\mathfrak{c} d^{\tfrac{5}{2}}\eps^{-1}.
    \end{equation}
    Moreover, recall that for any $v \in \R$, one has $\lceil v \rceil \leq v+1$. Hence, by \eqref{eqn: step6-n_1}-\eqref{eqn: step6-m_2} and the bound on $M_{d,\eps/6}$, we have the following bounds
    \begin{align}
        &n_1 = \lceil n_{1,d,\eps} \rceil \leq 2 + \log_2(M_{d,\eps/6}) = 2 + \log_2( 2 \cdot 3 \mathfrak{c}d^{\tfrac{5}{2}}\eps^{-1}),\label{eqn: bound on n_1}\\
        &    n_2 \leq 2 + \log_2(C_2)\leq 2 + \log_2(\mathfrak{c}^{\tfrac{1}{2}}),\label{eqn: bound on n_2}\\
        &        m_1 = \lceil m_{1,d,\eps}\rceil \leq 1 + \log_2(C_2^2d^2(\eps/6)^{-1})\leq 1 + \log_2(2 \cdot 3 \mathfrak{c}d^2\eps^{-1}),\label{eqn: bound on m_1}\\
        &
        m_2 = \lceil m_{2,d,\eps}\rceil \leq 2 + \log_2(M_{d,\eps/6}) + \log_2(C_2 d^2(\eps/6)^{-1}) \leq 2 + \log_2(2^2 3^2 \mathfrak{c}^{\tfrac{3}{2}}d^{\tfrac{9}{2}}\eps^{-2})\label{eqn: bound on m_2}.
    \end{align}
    Furthermore,
    \begin{equation}\label{eqn: bound on n+m+1}
        \begin{split}
            n+m+1 = n_1 + n_2 + m_1 + m_2 + 1 \leq 8 + \log_2(2^4 3^4\mathfrak{c}^4 d^{9}\eps^{-4})= \log_2(2^7 3^4\mathfrak{c}^4 d^{9}\eps^{-4}).
        \end{split}
    \end{equation}
   By Proposition \ref{prop: loading payoff circuit} and the fact that $\mathcal{A} = \mathcal{R}_h(\mathcal{P} \otimes I_2^{\otimes N - d(n_1+m_1)})$, the quantum circuit $\mathcal{A}$ uses $N$ qubits, where $N$ is given by \eqref{eqn: def N, q_k}. Using the upper bound \eqref{eqn: bound on N} for $N$, Assumption \ref{assumption: CPWA}, and \eqref{eqn: bound on n+m+1}, we thus have the following bound on the number of qubits used for the circuit $\mathcal{A}$
    \begin{equation}\label{eqn: bound on N qubits}
        \begin{split}
            N &\leq 24 K \cdot \max_{k=1,\ldots,K}I_k \cdot d (n+m+1)\\
            & \leq 24 C_2 d \cdot d (n+m+1)\\
            &\leq 24 C_2 d^2 \log_2(2^7 3^4\mathfrak{c}^4 d^{9}\eps^{-4})\\
            &\leq 24 C_2 d^2(\log_2(2^7 3^4) + 4 \log_2(\mathfrak{c}) + 9\log_2(d\eps^{-1}))\\
            &\leq 24 C_2 d^2(14 + 4 \log_2(\mathfrak{c}) + 9\log_2(d\eps^{-1}))\\
            &\leq 24 C_2 d^2\cdot 27 \log_2(\mathfrak{c})(1+\log_2(d\eps^{-1}))\\
            &\leq 648 C_2 \log_2(\mathfrak{c})  d^2(1+\log_2(d\eps^{-1}))\\
            &=: \mathfrak{C}_1 d^2 (1+\log_2(d\eps^{-1})).
        \end{split}
    \end{equation}
    Next, we count the number of elementary gates used to construct the quantum circuit $\mathcal{A}$, which will be denoted by $N_\mathcal{A}$. By Assumption \ref{assumption: distribution loading}, the number of elementary gates used to construct $\mathcal{P}$ is at most \eqref{eqn: bound on P_d,eps} with ($n \leftarrow n_1$, $m \leftarrow m_1$, $\eps\leftarrow \tfrac{\eps}{6C_2^2d^22^{n_1+1}}$ in the notation of Assumption \ref{assumption: distribution loading}). Hence, using \eqref{eqn: bound on n_1} and \eqref{eqn: bound on m_1}, the number of elementary gates used to construct $\mathcal{P}$ is bounded by 
    \begin{equation}\label{eqn: bound on gates of P_d,eps}
        \begin{split}
        &C_3 (n_1+m_1)^{C_3}d^{C_3}(\log_2(6\eps^{-1}C_2^2d^22^{n_1+1}))^{C_3}\\
        &= C_3 (n_1+m_1)^{C_3}d^{C_3}(\log_2(6C_2^2d^2\eps^{-1}) +  n_1 + 1))^{C_3}\\
        &\leq C_3 (3 + \log_2(2^2 3^2\mathfrak{c}^2d^{\tfrac{9}{2}}\eps^{-2}))^{C_3} d^{C_3}(\log_2(2 \cdot 3\mathfrak{c}d^2\eps^{-1}) + 3 + \log_2(2 \cdot 3 \mathfrak{c}d^{\tfrac{5}{2}}\eps^{-1}))^{C_3}\\
        &= C_3 d^{C_3} (\log_2(2^5 3^2 \mathfrak{c}^2 d^{\tfrac{9}{2}} \eps^{-2}))^{2C_3}.
        \end{split}
    \end{equation}
    By Proposition \ref{prop: loading payoff circuit}, Assumption \eqref{assumption: CPWA}, and the bound \eqref{eqn: bound on n+m+1}, the number of elementary gates used to construct $\mathcal{R}_h$ is estimated by at most
    \begin{equation}\label{eqn: bound on gates of R_h}
        \begin{split}
            16186 K^3 \big(\max\{I_1,\ldots,I_k\}\big)^3d^3(n+m+1)^3\leq 16186 (C_2 d)^3 d^3\left(\log_2(2^7 3^4\mathfrak{c}^4 d^{9}\eps^{-4})\right)^3.
        \end{split}
    \end{equation} 
    Hence, summing up \eqref{eqn: bound on gates of P_d,eps} and \eqref{eqn: bound on gates of R_h}, the number of elementary gates used to construct quantum circuit $\mathcal{A}$ is at most
        \begin{equation}\label{eqn: bound on N_A}
             N_\mathcal{A} \leq C_3 d^{C_3} (\log_2(2^5 3^2 \mathfrak{c}^2 d^{\tfrac{9}{2}} \eps^{-2}))^{2C_3} + 16186 (C_2 d)^3 d^3\left(\log_2(2^7 3^4\mathfrak{c}^4 d^{9}\eps^{-4})\right)^3.
        \end{equation}
    Next, we count the number of elementary gates used in line~7 of Algorithm~\ref{quantum algorithm}, which by
    Remark~\ref{remark: circuit Q^kA}~Item~2.\ coincides with 
    the number of elementary gates used to construct the quantum circuit $\mathcal{Q}^{k_t}\mathcal{A}$ in the Modified IQAE algorithm. Note that this number is also the number of elementary gates used in Algorithm~\ref{quantum algorithm}. Using \eqref{eqn: step6-s} and \eqref{eqn: bound on n_1}, the number $\mathfrak{s}^{-1}$ is bounded by 
    \begin{equation}\label{eqn: bound on s^-1}
        \mathfrak{s}^{-1} = \left(6\eps^{-1}(C_2^2d^22^{n_1+1})^3\right)^{\tfrac{1}{2}} \leq \sqrt{6}\eps^{-\tfrac{1}{2}}C_2d \left(2^{3+\log_2(6\mathfrak{c}d^{\tfrac{5}{2}}\eps^{-1})}\right)^{\frac{3}{2}} = 2^{\tfrac{13}{2}} 3^{2} C_2\mathfrak{c}^{\tfrac{3}{2}} d^{\tfrac{19}{4}}\eps^{-2}.
    \end{equation}
    Thus, by using Proposition \ref{prop: m-IQAE} Item 3. (with $\mathcal{A} \leftarrow \mathcal{A}$, $\eps \leftarrow \eps\mathfrak{s}/12$, $n \leftarrow N-1$, and $N \leftarrow N_\mathcal{A}$ in the Notation of Proposition \ref{prop: m-IQAE}) together with \eqref{eqn: bound on N qubits} and \eqref{eqn: bound on N_A}, we conclude that the number of elementary gates used in Algorithm~\ref{quantum algorithm} is bounded by
            \begin{equation}
        \begin{split}
            &\frac{\pi}{4\tfrac{\eps\mathfrak{s}}{12}}(8N^2+23+N_\mathcal{A})\\
            &\leq  2^{\tfrac{13}{2}} 3^{3}\pi C_2\mathfrak{c}^{\tfrac{3}{2}} d^{\tfrac{19}{4}}\eps^{-3}\big[8\big( 24 C_2 d^2 \log_2(2^7 3^4\mathfrak{c}^4 d^{9}\eps^{-4}) \big)^2 + 23 \\
            &\qquad + C_3 d^{C_3} (\log_2(2^5 3^2 \mathfrak{c}^2 d^{\tfrac{9}{2}} \eps^{-2}))^{2C_3} + 16186 (C_2 d)^3 d^3\left(\log_2(2^7 3^4\mathfrak{c}^4 d^{9}\eps^{-4})\right)^3\big]\\
            &\leq 2^{\tfrac{13}{2}} 3^{3}\pi C_2^4 C_3 d^{\max\{10.75,4.75+C_3\}}\eps^{-3}\big[8 \cdot 24^2 (\log_2(2^7 3^4\mathfrak{c}^4 d^{9}\eps^{-4}))^2 + 23 \\
            &\qquad + (\log_2(2^5 3^2 \mathfrak{c}^2 d^{\tfrac{9}{2}} \eps^{-2}))^{2C_3} + 16186  \left(\log_2(2^7 3^4\mathfrak{c}^4 d^{9}\eps^{-4})\right)^3\big]\\
            &\leq 2^{\tfrac{13}{2}} 3^{3}\pi C_2^4 C_3 d^{\max\{10.75,4.75+C_3\}}\eps^{-3}\Big[8 \cdot 24^2 \big(\log_2(2^7 3^4) + 4\log_2(\mathfrak{c}) + 9\big )^2\big(1+\log_2(d\eps^{-1}) \big)^2 + 23 \\
            &\qquad + \big(\log_2(2^5 3^2) + 2 \log_2( \mathfrak{c}) + \tfrac{9}{2} \big)^{2C_3}\big(1+\log_2 (d  \eps^{-1}))^{2C_3} \\
            &\qquad + 16186  \big(\log_2(2^7 3^4)+ 4\log_2(\mathfrak{c}) +9\big)^3 \big(1+\log_2(d\eps^{-1})\big)^3 \Big]\\
            &\leq 2^{\tfrac{13}{2}} 3^{3}\pi \Big[8 \cdot 24^2 \cdot \big((14 + 4 + 9 )\log_2(\mathfrak{c}) \big)^2 + 23 + \big((8.2 + 2  + 4.5)\log_2( \mathfrak{c}))^{2C_3} \\
            &\qquad + 16186  \big((14+ 4 +9) \log_2(\mathfrak{c})\big)^3  \Big]C_2^4 C_3  d^{\max\{10.75,4.75+C_3\}}\eps^{-3} (1+\log_2(d\eps^{-1}))^{\max\{3,2C_3\}}\\
            &\leq 2^{\tfrac{13}{2}} 3^{3}\pi \Big[8\cdot 24^2 + 23 + 1 + 16186  \Big] \\
            &\qquad \cdot  C_2^4 C_3\big(27\log_2(\mathfrak{c})\big)^{\max\{3,2C_3\}}d^{\max\{10.75,4.75+C_3\}}\eps^{-3} (1+\log_2(d\eps^{-1}))^{\max\{3,2C_3\}}\\
            &\leq (1.6\times 10^8)C_2^4 C_3\big(27\log_2(\mathfrak{c})\big)^{\max\{3,2C_3\}}d^{\max\{10.75,4.75+C_3\}}\eps^{-3} (1+\log_2(d\eps^{-1}))^{\max\{3,2C_3\}}\\
            &=: \mathfrak{C}_2 d^{\max\{10.75,4.75+C_3\}}\eps^{-3} (1+\log_2(d\eps^{-1}))^{\max\{3,2C_3\}}.
        \end{split}
    \end{equation}

    Lastly, we count the number of applications on $\mathcal{A}$.  Using \eqref{eqn: bound on queries of A in mIQAE} (with $\eps \leftarrow \eps \mathfrak{s}/12$ and $\alpha \leftarrow \alpha$ in the notation Proposition \ref{prop: m-IQAE}) and bound for $\mathfrak{s}^{-1}$ (c.f. \eqref{eqn: bound on s^-1}), the number of applications on $\mathcal{A}$ is at most 
\begin{equation}
    \tfrac{62\cdot 12}{\eps\mathfrak{s}}\ln\left(\tfrac{21}{\alpha}\right) \leq 62 \cdot 12 \cdot 2^{\tfrac{13}{2}} 3^{2} C_2\mathfrak{c}^{\tfrac{3}{2}} d^{\tfrac{19}{4}}\eps^{-3}\ln\left(\tfrac{21}{\alpha}\right) \leq (6.1\times 10^5) C_2\mathfrak{c}^{\tfrac{3}{2}}d^{4.75}\eps^{-3}\ln(\tfrac{21}{\alpha}) =: \mathfrak{C}_3 d^{4.75}\eps^{-3}\ln(\tfrac{21}{\alpha}) .
\end{equation}
\end{proof}

\section{Conclusion}
In this paper we have developed with Algorithm~\ref{quantum algorithm} a quantum Monte Carlo algorithm to approximately solve multidimensional 
Black-Scholes PDEs. The contributions of this paper are the following. 

 First, our algorithm allows the payoff function to be of general form and is only required to be continuous piecewise affine. From a financial point of view, this is not very restrictive, as most European options are continuous piecewise affine, see also the various relevant examples provided in Example~\ref{example: call options}. This extends the existing quantum algorithms which typically require the continuous piecewise affine payoff function to be either one-dimensional or to be a basket option. 
 
 Moreover, we provided a mathematical rigorous error and complexity analysis of Algorithm~\ref{quantum algorithm}, which we see as our main contribution of the paper. This allows us to prove that the computational complexity of the algorithm only grows polynomially in the space dimension $d$ of the PDE and the prescribed reciprocal of the accuracy $\varepsilon$. In addition, we see that for continuous piecewise affine payoff functions which are uniformly bounded, the computational running time of Algorithm~\ref{quantum algorithm} scales $O(\eps^{-\nicefrac{3}{2}})$. Therefore, compared to classical (i.e.\ non quantum-based) Monte Carlo algorithms which scale $O(\eps^{-2})$, we indeed have proved that Algorithm~\ref{quantum algorithm} provides a speed-up.
 
 Furthermore, we have developed a package we named \texttt{qfinance} within the \texttt{Qiskit} framework which can be used to run Algorithm~\ref{quantum algorithm} on a computer  for the case $d=1,2$. The \texttt{OptionPricing} class within this package enables the user to input  all  parameters of the underlying stocks, to choose the class of continuous piecewise affine payoff functions within the ones presented in   Example~\ref{example: call options}, as well as to specify to error tolerance level. We numerically demonstrated the applicability of our algorithm in this low-dimensional setting. Moreover, we have discussed the scalability of our algorithm by explaining how one could extend our code for the general $d$-dimensional setting.
 %
 	We also highlighted that the limitation of the numerical simulation are not caused by our developed Algorithm~\ref{quantum algorithm}, but due to the limited quantum computing hardware currently available.
 
 We emphasize that the outline of Algorithm~\ref{quantum algorithm}, namely to approximate the solution of the Black-Scholes PDE via its Feynman-Kac representation by first uploading the transition probability of the underlying (log-normally distributed) SDE, followed by the uploading of the payoff function, and then applying a Quantum amplitude estimation algorithm to estimate the solution of the PDE is not new and has been already applied, e.g., in \cite{chakrabarti2021threshold,QC5_Patrick,QC4_optionpricing}.
 However, so far, no mathematical rigorous error and complexity analysis of such a quantum Monte Carlo algorithm to solve  Black-Scholes PDEs, or any quantum based algorithm to solve PDEs, has been provided in the literature, 
   which was the main goal of this paper.

\section*{Acknowledgment}
Financial support by the 
Nanyang Assistant Professorship Grant (NAP Grant) \textit{Machine Learning based Algorithms in Finance and Insurance} 
and the grant NRF2021-QEP2-02-P06 
 is
gratefully acknowledged.

\printbibliography
\end{document}